\documentclass[11pt]{article}

\usepackage[margin=1in]{geometry}
\usepackage{fullpage}
\usepackage{tgtermes}
\usepackage[T1]{fontenc}
\usepackage[colorlinks,citecolor=blue,linkcolor=blue,urlcolor=red]{hyperref}
\usepackage{graphicx}
\usepackage{amsfonts,amsmath,amsthm,amssymb,dsfont,mathtools}
\usepackage{xfrac,nicefrac}
\usepackage{mathdots}
\usepackage{bm,bbm}
\usepackage{url}
\usepackage{caption}
\usepackage{paralist}
\usepackage{enumerate}
\usepackage[normalem]{ulem}
\usepackage{enumitem}
\usepackage{xspace}
\xspaceaddexceptions{]\}}
\usepackage[capitalise]{cleveref}
\crefname{algocf}{Algorithm}{Algorithms}
\Crefname{algocf}{Algorithm}{Algorithms}
\usepackage{comment}
\usepackage{tabu}
\usepackage{framed}
\usepackage{float,wrapfig}
\usepackage{subfigure}
\usepackage{tikz}
\usepackage{bm}
\usepackage{soul}
\usepackage{enumitem}
\usepackage[usenames,dvipsnames]{pstricks}
\usepackage{thmtools, thm-restate}
\usepackage[linesnumbered,boxed,ruled,vlined]{algorithm2e}
\usepackage{algpseudocode}
\usetikzlibrary{decorations.pathreplacing}

\usepackage{regexpatch}

\theoremstyle{plain}

\newtheorem{theorem}{Theorem}[section]
\newtheorem{lemma}[theorem]{Lemma}

\newtheorem{cor}[theorem]{Corollary}
\theoremstyle{definition}
\newtheorem{claim}[theorem]{Claim}
\newtheorem{definition}[theorem]{Definition}
\newtheorem{remark}[theorem]{Remark}

\newenvironment{proofof}[1]{\begin{proof}[Proof of #1]}{\end{proof}}

\DeclarePairedDelimiter{\tvdbasic}{\lVert}{\rVert}
\makeatletter
\newcommand{\@tvdstar}[2]{\tvdbasic*{#1 - #2}_{\mathrm{TV}}}
\newcommand{\@tvdnostar}[3][]{\tvdbasic[#1]{#2 - #3}_{\mathrm{TV}}}
\newcommand{\tvd}{\@ifstar\@tvdstar\@tvdnostar}

\makeatother

\usepackage[breakable]{tcolorbox}

\usepackage{mleftright}

\renewcommand{\left}{\mleft}
\renewcommand{\right}{\mright}

\renewcommand{\epsilon}{\varepsilon}

\renewcommand{\Pr}{\operatorname*{\mathbf{Pr}}}

\newcommand{\F}{\mathbbm{F}}

\newcommand{\Z}{\mathbbm{Z}}

\renewcommand{\hat}{\widehat}

\newtcolorbox{construction}[2][]
{
	breakable,
	colframe = gray!50,
	colback  = gray!10,
	coltitle = gray!10!black,
	before skip = 10pt,
	after skip = 10pt,
	title    = \textbf{#2},
	#1,
}
\newtcolorbox{graphview}[2][]
{
	breakable,
	colframe = black!30,
	colback  = black!0,
	coltitle = gray!10!black,
	before skip = 10pt,
	after skip = 10pt,
	title    = \textbf{#2},
	#1,
}
\newtcolorbox{notation}[2][]
{
	breakable,
	colframe = black!30,
	colback  = black!0,
	coltitle = gray!10!black,
	before skip = 10pt,
	after skip = 10pt,
	title    = \textbf{#2},
	#1,
}

\newcommand{\floor}[1]{\left\lfloor {#1} \right\rfloor}

\newcommand{\bk}[1]{\left( #1 \right)}
\newcommand{\bigbk}[1]{\big({#1}\big)}
\newcommand{\midbk}[1]{({#1})}
\newcommand{\Bigbk}[1]{\Big({#1}\Big)}
\newcommand{\biggbk}[1]{\bigg({#1}\bigg)}
\newcommand{\Bk}[1]{\left[ #1 \right]}

\newcommand{\bigBk}[1]{\big[ #1 \big]}

\newcommand{\BK}[1]{\left\{ #1 \right\}}
\newcommand{\midBK}[1]{\{ #1 \}}

\newcommand{\abs}[1]{\left|{#1}\right|}
\newcommand{\midabs}[1]{|{#1}|}
\newcommand{\bigabs}[1]{\big|{#1}\big|}

\newcommand{\angbk}[1]{\left\langle{#1}\right\rangle}
\newcommand{\E}{\mathop{\mathbb{E}}}
\renewcommand{\Pr}{\mathop{\mathrm{Pr}}}

\newcommand{\poly}{\mathrm{poly}}
\newcommand\numberthis{\addtocounter{equation}{1}\tag{\theequation}}

\newcommand{\CW}{\mathrm{CW}}
\newcommand{\sym}{\textup{sym}}
\newcommand{\degen}{\unrhd}

\newcommand{\Aut}{\textup{Aut}}

\newcommand{\BKmid}[2]{{\mleft\{ #1 ~\middle|~ #2 \mright\}}}

\renewcommand{\split}{\textsf{split}}
\newcommand{\Split}{\textsf{split}}

\newcommand{\eqdef}{\eqqcolon}
\newcommand{\defeq}{\coloneqq}

\newcommand{\rot}{\textup{rot}}
\newcommand{\swap}{\textup{swap}}

\newcommand{\tauthm}{\ensuremath{\text{Sch\"onhage's $\tau$ theorem (\Cref{thm:Schonhage})}}}

\newcommand{\A}{\textup{A}}
\newcommand{\B}{\textup{B}}
\newcommand{\T}{\mathcal{T}}
\renewcommand{\S}{\mathcal{S}}

\newcommand{\revtext}{\textup{rev}}
\newcommand{\nrot}{\textup{(nrot)}}

\makeatletter
\let\save@mathaccent\mathaccent
\newcommand*\if@single[3]{%
  \setbox0\hbox{${\mathaccent"0362{#1}}^H$}%
  \setbox2\hbox{${\mathaccent"0362{\kern0pt#1}}^H$}%
  \ifdim\ht0=\ht2 #3\else #2\fi
  }
\newcommand*\rel@kern[1]{\kern#1\dimexpr\macc@kerna}
\newcommand*\widebar[1]{\@ifnextchar^{{\wide@bar{#1}{0}}}{\wide@bar{#1}{1}}}
\newcommand*\wide@bar[2]{\if@single{#1}{\wide@bar@{#1}{#2}{1}}{\wide@bar@{#1}{#2}{2}}}
\newcommand*\wide@bar@[3]{%
  \begingroup
  \def\mathaccent##1##2{%
    \let\mathaccent\save@mathaccent
    \if#32 \let\macc@nucleus\first@char \fi
    \setbox\z@\hbox{$\macc@style{\macc@nucleus}_{}$}%
    \setbox\tw@\hbox{$\macc@style{\macc@nucleus}{}_{}$}%
    \dimen@\wd\tw@
    \advance\dimen@-\wd\z@
    \divide\dimen@ 3
    \@tempdima\wd\tw@
    \advance\@tempdima-\scriptspace
    \divide\@tempdima 10
    \advance\dimen@-\@tempdima
    \ifdim\dimen@>\z@ \dimen@0pt\fi
    \rel@kern{0.6}\kern-\dimen@
    \if#31
      \overline{\rel@kern{-0.6}\kern\dimen@\macc@nucleus\rel@kern{0.4}\kern\dimen@}%
      \advance\dimen@0.4\dimexpr\macc@kerna
      \let\final@kern#2%
      \ifdim\dimen@<\z@ \let\final@kern1\fi
      \if\final@kern1 \kern-\dimen@\fi
    \else
      \overline{\rel@kern{-0.6}\kern\dimen@#1}%
    \fi
  }%
  \macc@depth\@ne
  \let\math@bgroup\@empty \let\math@egroup\macc@set@skewchar
  \mathsurround\z@ \frozen@everymath{\mathgroup\macc@group\relax}%
  \macc@set@skewchar\relax
  \let\mathaccentV\macc@nested@a
  \if#31
    \macc@nested@a\relax111{#1}%
  \else
    \def\gobble@till@marker##1\endmarker{}%
    \futurelet\first@char\gobble@till@marker#1\endmarker
    \ifcat\noexpand\first@char A\else
      \def\first@char{}%
    \fi
    \macc@nested@a\relax111{\first@char}%
  \fi
  \endgroup
}
\makeatother

\makeatletter
\xpatchcmd\thmt@restatable{%
\csname #2\@xa\endcsname\ifx\@nx#1\@nx\else[{#1}]\fi
}{%
\ifthmt@thisistheone
\csname #2\@xa\endcsname\ifx\@nx#1\@nx\else[{#1}]\fi
\else
\csname #2\@xa\endcsname[{Restated}]
\fi}{}{}
\makeatother

\renewcommand{\tilde}{\widetilde}
\renewcommand{\bar}{\widebar}

\renewcommand{\F}{\mathbb F}
\newcommand{\Span}{\textup{span}}

\newcommand{\ourbound}{2.371866}

\renewcommand{\l}{\ell}

\newcommand{\smallsub}{\scriptscriptstyle}

\newcommand{\ind}{\mathbbm{1}}
\newcommand{\mymiddle}{\;\middle|\;}
\newcommand{\bigbinom}[2]{\left( \begin{array}{c} #1\\#2 \end{array} \right)}

\DeclareMathOperator*{\argmax}{arg\,max}

\newcommand{\hashx}{h_{\textup{X}}}  %
\newcommand{\hashy}{h_{\textup{Y}}}
\newcommand{\hashz}{h_{\textup{Z}}}

\newcommand{\alphx}[1][\alpha]{{#1}_{\smallsub \textup{X}}}
\newcommand{\alphy}[1][\alpha]{{#1}_{\smallsub \textup{Y}}}
\newcommand{\alphz}[1][\alpha]{{#1}_{\smallsub \textup{Z}}}
\let\alphax\alphx
\let\alphay\alphy
\let\alphaz\alphz

\newcommand{\valthree}{V_{\tau}^{(3)}}
\newcommand{\valone}{V_{\tau}^{(\textup{nrot})}}
\newcommand{\valsix}{V_{\tau}^{(6)}}

\newcommand{\lvl}{\ensuremath{\l}}  %
\newcommand{\nextlvl}{\ensuremath{(\l \! + \! 1)}}
\newcommand{\lastlvl}{\ensuremath{(\l \! - \! 1)}}

\newcommand{\numxblock}{N_{\textup{BX}}}
\newcommand{\numyblock}{N_{\textup{BY}}}
\newcommand{\numzblock}{N_{\textup{BZ}}}
\newcommand{\alphabx}[1][\alpha]{#1_{\smallsub \textup{BX}}}
\newcommand{\alphaby}[1][\alpha]{#1_{\smallsub \textup{BY}}}
\newcommand{\alphabz}[1][\alpha]{#1_{\smallsub \textup{BZ}}}
\newcommand{\alphant}[1][\alpha]{#1_{\smallsub \textup{N}}}
\newcommand{\numtriple}{N_{\alphax, \alphay, \alphaz}}
\newcommand{\numretain}{N_{\textup{ret}}}
\newcommand{\numalpha}[1][\alpha]{N_{#1}}
\newcommand{\distShareMargin}[1][\alpha]{D_{#1}}

\newcommand{\alphap}[1][\alpha]{#1_{\smallsub \textup{P}}}  %
\newcommand{\alphaval}[1][\alpha]{#1_{\smallsub \textit{V}_\tau}}  %
\newcommand{\pcomp}{p_{\textup{comp}}}

\newcommand{\splres}{\tilde{\alpha}}  %
\newcommand{\splresA}{\splres_{\textup{A}}}
\newcommand{\splresB}{\splres_{\textup{B}}}
\newcommand{\splresAvg}{\splres_{\textup{avg}}}
\newcommand{\Thole}{T_{\mathrm{hole}}}
\let\splresAverage\splresAvg

\newcommand{\Sym}{\mathcal{S}}
\newcommand{\shuf}{\phi}
\newcommand{\shufi}[1][i]{\shuf_{#1}}

\newcommand{\fracnonholeIJK}{\eta_{\smallsub I,J,K}}
\newcommand{\splresX}{\splres_{\smallsub \textup{X}}}
\newcommand{\splresY}{\splres_{\smallsub \textup{Y}}}
\newcommand{\splresZ}{\splres_{\smallsub \textup{Z}}}

\newcommand{\typicalset}[1][K]{B_{\textup{typical}, #1}}
\newcommand{\plusplusk}[1][k]{\textup{+}, \textup{+}, #1}
\newcommand{\splresavg}{\bar{\alpha}}

\newcommand{\complementshape}{i^{(r)} \!-\! i', \, j^{(r)} \!-\! j', \, k^{(r)} \!-\! k'}
\newcommand{\usefulset}[1][X_I, Y_J, Z_K]{B_{\textup{useful}}(#1)}
\newcommand{\strongset}[1][X_I, Y_J, Z_K]{B_{\textup{strong}}(#1)}

\newcommand{\param}{\mathcal{X}}
\newcommand{\paramglob}{\mathcal{X}_{\textup{glob}}}
\newcommand{\Vglob}{V_{\textup{glob}}}
\newcommand{\objfun}{\textsf{obj}}

\newcommand{\numzerotwotwo}{N_{022}}
\newcommand{\splresavgZ}{\bar{\alpha}_{\smallsub \textup{Z}}}

\begin{document}

\author{
  Ran Duan \thanks{Email: duanran@mail.tsinghua.edu.cn.} \\
  \small Tsinghua University \\
  \and
  Hongxun Wu \thanks{Email: wuhx@berkeley.edu.} \\
  \small UC Berkeley \\
  \and
  Renfei Zhou \thanks{Email: zhourf20@mails.tsinghua.edu.cn.} \\
  \small Tsinghua University
}

\title{Faster Matrix Multiplication via Asymmetric Hashing}
\date{}
\setcounter{page}{0} \clearpage
\maketitle
\thispagestyle{empty}
\begin{abstract}
  Fast matrix multiplication is one of the most fundamental problems in algorithm research. The exponent of the optimal time complexity of matrix multiplication is usually denoted by $\omega$. This paper discusses new ideas for improving the laser method for fast matrix multiplication. We observe that the analysis of higher powers of the Coppersmith-Winograd tensor [Coppersmith \& Winograd 1990] incurs a ``combination loss'', and we partially compensate for it using an asymmetric version of CW's hashing method. By analyzing the eighth power of the CW tensor, we give a new bound of $\omega < \ourbound$, which improves the previous best bound of $\omega < 2.372860$ [Alman \& Vassilevska Williams 2020]. Our result breaks the lower bound of $2.3725$ in [Ambainis, Filmus \& Le Gall 2015] because of the new method for analyzing component (constituent) tensors.
\end{abstract} \newpage

\maketitle

\section{Introduction}
The time complexity of multiplying two $n\times n$ matrices is usually denoted by $O(n^{\omega+o(1)})$ for some real number $\omega$. (It is easy to see that $2\leq \omega\leq 3$.) Although it has been studied for more than 50 years, the exact value of $\omega$ is still unknown, and many people believe that $\omega=2$. The determination of the constant $\omega$ would have wide implications. Not only do many matrix operations have similar complexities as fast matrix multiplication (FMM) algorithms, such as LUP decomposition, matrix inversion, determinant~\cite{AAH_book,Bunch1974}, algorithms for many combinatorial problems can also be accelerated by FMM algorithms, including graph problems like transitive closure (see~\cite{AAH_book}), unweighted all-pair shortest paths~\cite{seidel1995,zwick2002}, all-pair bottleneck path~\cite{APBP}, and other problems such as context-free grammar parsing~\cite{valiant1975} and RNA-folding~\cite{bringmann2019}.

The first algorithm breaking cubic time bound is by Strassen in 1969~\cite{strassen1969}, who showed that $\omega\leq\log_2 7<2.8074$. After that, there is a series of research (e.g. \cite{pan1978,bini1979,schonhage1981partial,romani1982,coppersmith1981,strassen1986,coppersmith1987matrix,stothers2010,williams2012,legall2014}) improving the upper bound for $\omega$. The current best bound is $\omega<2.3728596$ given by the refined laser method of Alman and Vassilevska Williams~\cite{alman2021}. Very recently, DeepMind~\cite{deepmind22} used a reinforcement learning method to improve the ranks of several small matrix multiplication tensors. They gave a time bound of $O(n^{2.778})$ over $\mathbb{F}_2$.

The most recent improvements along this line are all based on applying Strassen's laser method \cite{strassen1986} to Coppersmith-Winograd tensors \cite{coppersmith1987matrix}, $\mathrm{CW}_q$. The laser method first picks a tensor $T$ and takes its $N$-th power $T^{\otimes N}$. This tensor $T$ is usually efficient to compute (having a small asymptotic rank). Then it forces a subset of variables in $T^{\otimes N}$ to be zero (zeroing-out) and degenerates it into the direct sum of independent matrix multiplication tensors. This gives an efficient algorithm for computing matrix multiplications. 

The original work of Coppersmith and Winograd \cite{coppersmith1987matrix} not only analyzed $\mathrm{CW}_q$, but also its second power, $\mathrm{CW}_q^{\otimes 2}$. By a clever hashing technique using the Salam-Spencer set \cite{salem1942, Behrend1946}, their analysis of $\mathrm{CW}_q$ shows that $\omega < 2.38719$. For the second power, they applied the laser method to $\mathrm{CW}_q$ after merging some subtensors of $\mathrm{CW}_q^{\otimes 2}$ into larger matrix multiplications. This gives a better bound of $\omega < 2.375477$. Stothers \cite{stothers2010} and Vassilevska Williams \cite{williams2012} independently improved the analysis to the fourth and eighth powers of $\mathrm{CW}_q$ while using computer programs to automate complicated calculations. Le Gall \cite{legall2014} further improved it to the $32$-th power by formulating it as a series of convex optimization problems. 

Although analyzing higher and higher powers of $\mathrm{CW}_q$ gives improved bounds for $\omega$, the analysis is not tight. Alman and Vassilevska Williams \cite{alman2021} focus on the extra loss in hash moduli in the higher-power analysis, and they compensate for this loss with their refined laser method. This gives the current best algorithm. Note that these analyses are all recursive, namely, the value of the $q$-th power depends on the bounds for the $(q - 1)$-th power. 

In this paper, we identify a more implicit ``combination loss'', which is our main contribution. Such loss arises not from a single recursion level, but from the structure of adjacent levels. To demonstrate that this observation indeed leads to an improved algorithm, we compensate for this loss using an asymmetric hashing method. This improves the analysis of the second power by Coppersmith and Winograd \cite{coppersmith1987matrix} to $\omega < 2.374631$. We also generalize it to higher powers and obtain the improved bound of $\omega < \ourbound$.\footnote{This bound is slightly better than the previous version of this paper due to more flexibility in optimizing parameters.} A similar asymmetric hashing method was used in \cite{coppersmith1987matrix} to analyze an asymmetric tensor. Fast rectangular matrix multiplication algorithms (e.g. \cite{coppersmith1982rapid, coppersmith1997rectangular, huang1998fast, le2012faster, gall2018improved}) also use asymmetric hashing to degenerate the tensor power $T^{\otimes N}$ in the laser method into independent rectangular matrix multiplications. In this paper, we use asymetric hashing for a different purpose, that is, to compensate for the ``combination loss''.

On the limitation side, Ambainis, Filmus and Le Gall \cite{ambainis2015fast} showed that if we analyze the lower powers of $\mathrm{CW}_q$ using previous approaches, including the refined laser method, we cannot give a better upper bound than $\omega < 2.3725$ even if we go to arbitrarily high powers. Our bound of $\omega < \ourbound$ breaks such a limitation by analyzing only the $8$th power. This is because our analysis improves the values for lower powers and this limitation no longer applies. They also proved a limitation of $2.3078$ for a wider class of algorithms, which includes our algorithm. There are also works that prove more general barriers \cite{doi:10.1137/19M124695X, alman_et_al:LIPIcs:2018:8360,v017a001,v017a002,DBLP:conf/mfcs/BlaserL20}. Our approach is also subject to these barriers.

Another approach to fast matrix multiplication is the group theoretical method of Cohn and Umans \cite{1238217, 1530730, cohn2013fast}. There are also works on its limitations \cite{alon2012sunflowers, BCC, arxiv.1712.02302, DBLP:journals/corr/abs-2204-03826}. 

\subsection{Organization of this Paper}

In Section~\ref{sec:overview}, we will first give an overview of the laser method and our improved algorithm. In Section~\ref{sec:prelim}, we introduce the concepts and notations that we will use in this paper, as well as some basic building blocks, for example, the asymmetric hashing method. To better illustrate our ideas, in Section~\ref{sec:2nd}, we give an improved analysis of the second power of the CW tensor. In Sections~\ref{sec:hole_lemma}, \ref{sec:global}, and \ref{sec:component}, we will extend the analysis to higher powers. In Section~\ref{sec:result}, we discuss our optimization program and give the numerical results.

\section{Technical Overview}
\label{sec:overview}
Our improvement is based on the observation that a hidden ``combination loss'' exists. In the overview, we will explain our ideas based on the second-power analysis of Coppersmith-Winograd \cite{coppersmith1987matrix}. Their analysis implies that $\omega < 2.375477$ which is still the best upper bound via the second power of $\mathrm{CW}_q$. We will point out the ``combination loss'' in their analysis, and present our main ideas which serve to compensate for such loss. With certain twists, the same ideas also apply to higher powers and allow us to obtain the improved bound for $\omega$.

\subsection{Coppersmith-Winograd Algorithm} \label{sec:CW}

It is helpful to start with a high-level description of Coppersmith-Winograd~\cite{coppersmith1987matrix}. Our exposition here differs a little from their original work. Specifically, their work uses the values of subtensors in a black-box way. We open up this black box and look at the structure inside those subtensors. This change will make the combination loss visible. 

In the following, we will assume some familiarity with the Coppersmith-Winograd Algorithm. For an exposition of their algorithm, the reader may also refer to the excellent survey by Bläser~\cite{gs005}. 

\paragraph*{The 2nd-power of the CW tensor.}

To begin, we first set up some minimum notation. The starting point of the laser method is a tensor $T$ that can be efficiently computed (having a small asymptotic rank). The \textit{Coppersmith-Winograd Tensor} $\mathrm{CW}_q$, whose asymptotic rank is $q+2$, is widely used in previous works. The exact definition of $\CW_q$ is the following.\footnote{Readers may also refer to \Cref{sec:CWtensor}.}
\[\CW_q = \sum_{i=1}^q (x_i y_i z_0 + x_i y_0 z_i + x_0 y_i z_i) + x_0 y_{0} z_{q+1} + x_{0} y_{q+1} z_0 + x_{q+1} y_0 z_{0} .\]
We will only use several properties of $\CW_q$. First, it is a tensor over variable sets $X = \{x_0, x_1, \dots, x_{q + 1}\}$, $Y = \{y_0, y_1, \dots, y_{q + 1}\}$, and $Z = \{z_0, z_1, \dots, z_{q + 1}\}$. We let $X_0 = \{x_0\}, X_1 = \{x_1, x_2, \dots, x_q\}, X_2 = \{x_{q+1}\}$ and define $\{Y_j\}_{j = 0,1,2}, \{Z_k\}_{k= 0,1,2}$ similarly. Second, for the partition $X = X_0 \cup X_1 \cup X_2$, $Y = Y_0 \cup Y_1 \cup Y_2$, and $Z = Z_0 \cup Z_1 \cup Z_2$, the following holds:

\begin{itemize}
\item For all $i + j + k \neq 2$, the subtensor of $\mathrm{CW}_q$ over $X_i$, $Y_j$, $Z_k$ is zero. 
\item For all $i + j + k = 2$, the subtensor of $\mathrm{CW}_q$ over $X_i$, $Y_j$, $Z_k$ is a matrix multiplication tensor, denoted by $T_{i,j,k}$.
\end{itemize}

We call this partition the \emph{level-1 partition}. $\mathrm{CW}_q$ is the summation of all such $T_{i,j,k}$'s:
$$\mathrm{CW}_q = \sum_{i + j + k = 2} T_{i,j,k} = T_{0,1,1}+T_{1,0,1}+T_{1,1,0}+T_{0,0,2}+T_{0,2,0}+T_{2,0,0}.$$

For any two sets $A$ and $B$, their Cartesian product $A \times B$ is $\{(a, \, b) \mid a \in A, \, b \in B\}$. The second power $\mathrm{CW}_q^{\otimes 2}$ is a tensor over variable sets $\tilde{X} = X \times X$, $\tilde{Y} = Y \times Y$, and $\tilde{Z} = Z \times Z$:
$$\mathrm{CW}_q^{\otimes 2} = \sum_{i + j + k = 2} \ \sum_{i' + j' + k' = 2} T_{i,j,k} \otimes T_{i',j',k'}.$$

For $\tilde{X} = X \times X$, the product of two level-1 partitions gives $\tilde{X} = \bigcup_{i', i''} X_{i'} \times X_{i''}$. We define the level-2 partition to be a coarsening of this. For each $i \in \{0, 1, 2, 3, 4\}$, we define
\[\tilde{X}_i \defeq \bigcup_{i' + i'' = i} X_{i'} \times X_{i''}, \quad
  \tilde{Y}_j \defeq \bigcup_{j' + j'' = j} Y_{j'} \times Y_{j''}, \quad
  \tilde{Z}_k \defeq \bigcup_{k' + k'' = k} Z_{k'} \times Z_{k''}.\]
For example, $\tilde{X}_2 = (X_0 \times X_2) \cup (X_1 \times X_1) \cup (X_2 \times X_0)$. The level-2 partition is given by $\tilde{X} = \tilde{X}_0 \cup \tilde{X}_1 \cup \cdots \cup \tilde{X}_4$, $\tilde{Y} = \tilde{Y}_0 \cup \tilde{Y}_1 \cup \cdots \cup \tilde{Y}_4$, and $\tilde{Z} = \tilde{Z}_0 \cup \tilde{Z}_1 \cup \cdots \cup \tilde{Z}_4$.

For any $i + j + k = 4$, let $T_{i,j,k}$ be the subtensor of $\mathrm{CW}_q^{\otimes 2}$ over $\tilde{X}_i, \tilde{Y}_j, \tilde{Z}_k$. We have
$$T_{i,j,k} = \sum_{i' + j' + k' = 2} T_{i', j', k'} \otimes T_{i - i',j - j',k - k'}.$$
For example, $T_{1,1,2} = T_{1,1,0} \otimes T_{0,0,2} + T_{0,0,2} \otimes T_{1,1,0} + T_{1,0,1} \otimes T_{0,1,1} + T_{0,1,1} \otimes T_{1,0,1}$.\footnote{For notational convenience, we let $T_{i,j,k} = 0$ if any of $i,j,k$ is negative.}

The tensor $\mathrm{CW}_q^{\otimes 2}$ is the sum of all such $T_{i,j,k}$'s:
\begin{align*}
\mathrm{CW}_q^{\otimes 2} = &\sum_{i + j + k = 4} T_{i,j,k} \\ = &\phantom{{}+{}} T_{0,0,4} + T_{0,4,0} + T_{4,0,0} \\ &+ T_{0,1,3} + T_{0,3,1} + T_{1,0,3} + T_{1,3,0} + T_{3,0,1} + T_{3,1,0} \\ &+ T_{0,2,2} + T_{2,0,2} + T_{2,2,0} \\ &+ T_{1,1,2} + T_{1,2,1} + T_{2,1,1}.
\end{align*}
Note that these $T_{i,j,k}$'s ($i + j + k = 4$) may not be matrix multiplication tensors. For the subtensors of $\mathrm{CW}_q^{\otimes 2}$, it is true that $T_{0,0,4}, T_{0,1,3}, T_{0,2,2}$ (and their permutations) are indeed all matrix multiplication tensors, while $T_{1,1,2}$ (and its permutations) are not. We will need this important fact in the two-level analysis of Coppersmith and Winograd. 

\paragraph*{The Laser Method.} We say a set of subtensors $T_1, T_2, \dots, T_\ell$ of $T$ is \underline{independent} if and only if the following two conditions hold:
\begin{itemize}
\item These subtensors are supported on disjoint variables. 
\item The restriction of $T$ over the union of the supports of the subtensors $T_1, \ldots, T_\ell$ is exactly $T_1 + T_2 + \cdots + T_\ell$. (That is, there cannot be any additional term in $T$ other than those from $T_1, T_2, \dots, T_\ell$.)
\end{itemize}
As we have seen, the tensor $\mathrm{CW}_q$ is the sum of many matrix multiplication tensors $T_{i,j,k}$ ($i + j + k = 2$). If these $T_{i,j,k}$'s were independent, we would have succeeded, as we would be able to efficiently compute many independent matrix multiplications using $\mathrm{CW}_q$. A direct application of Sch\"onhage's $\tau$ theorem (See Theorem \ref{thm:Schonhage}) would give us an upper bound on $\omega$. 

Clearly, these $T_{i,j,k}$'s are not even supported on disjoint variables. The first step of the laser method is to take its $n$-th tensor power. Since the asymptotic rank of $\mathrm{CW}_q$ is $q + 2$, the $n$-th power needs $(q + 2)^{n}$ many multiplications to compute. ($n$ will later go to infinity.) Within $\mathrm{CW}^{\otimes n}_q$, there are certainly many non-independent subtensors that are matrix multiplication tensors. We will carefully zero out some variables, i.e. setting some of them to zero. After zeroing out, all remaining matrix multiplication tensors will be independent. Therefore, we can apply Sch\"onhage's $\tau$ theorem. This is the single-level version of the laser method.

Such a single-level analysis gives a bound of $\omega < 2.38719$. By considering the second power, Coppersmith and Winograd \cite{coppersmith1987matrix} get an improved bound of $\omega < 2.375477$. Note in the analysis above, for any two non-independent matrix multiplication tensors, only one of them will survive the zeroing-out. Looking ahead, the two-level analysis will further exploit the tensor power structure and ``merge'' some non-independent matrix multiplication tensors into a larger one.

\paragraph*{Variable Blocks.} In the two-level analysis, we take the $2n$-th power of $\mathrm{CW}_q$. After that, the level-1 partition of the X variables becomes $(X_{\hat i_1} \times X_{\hat 
 i_2}\times \cdots \times X_{\hat i_{2n}})_{\hat i_1,\hat i_2, \dots, \hat i_{2n} \in \{0,1,2\}}$ (similarly for the Y and Z variables). Let $\hat I\in \{0,1,2\}^{2n}$ be a sequence. We call each $X_{\hat I} \defeq X_{\hat i_1} \times X_{\hat i_2}\times \cdots \times X_{\hat i_{2n}}$ a \underline{level-1 variable block}. Same for $Y_{\hat J}$ and $Z_{\hat K}$.\footnote{In the overview, we always let $I=(i_1,\cdots,i_n)$ and $\hat I =(\hat i_1,\cdots,\hat i_{2n})$. Other sequences ($I',\hat{I}', J, K$, etc.) are defined similarly. }

In level-2, we first equivalently view $(\mathrm{CW}_q)^{\otimes 2n}$ as $(\mathrm{CW}_q^{\otimes 2})^{\otimes n}$, i.e., the $n$-th power of $\mathrm{CW}_q^{\otimes 2}$. For each sequence $I \in \{0,1,2,3,4\}^{n}$, we call its corresponding block $X_{I} \defeq \tilde{X}_{i_1} \times \tilde{X}_{i_2} \times \cdots \times \tilde{X}_{i_n}$ a \underline{level-2 variable block}. We define $Y_J, Z_K$ similarly. (Below, we will always state our definitions in terms of $X$ variable blocks. The same definitions always apply to Y/Z blocks as well.)

\paragraph*{Two-level Analysis.} We are now ready to describe the two-level analysis of Coppersmith-Winograd in detail.
\begin{itemize}
\item \textbf{The Analysis for Level-2.} At this level, our goal is to first zero out $(\mathrm{CW}_q)^{\otimes 2n}$ into independent subtensors, each isomorphic to 
\[\T^{\alpha} \defeq \bigotimes_{i + j + k = 4} T_{i,j,k}^{\otimes \alpha(i,j,k) n}\]
for some distribution $\alpha : \{(i,j,k)\}_{i + j + k = 4} \rightarrow [0,1]$. Moreover, each $\alpha(i,j,k)$ has to be a multiple of $1/n$, so that the exponent $\alpha(i,j,k) n$ is always an integer. There are at most $\poly(n)$ many such $\alpha$. We will select a typical distribution $\alpha$ and only keep the independent subtensors isomorphic to that corresponding $\T^\alpha$. By doing so, we incur a factor of $\frac{1}{\poly(n)}$ that is negligible compared to $(q + 2)^n$.

We use $[\alphx(i)]_{i \in \{0,1,\dots,4\}}$ to denote the X marginal distribution of $\alpha$. We say that a block $X_I$ obeys $\alphx$ if and only if the distribution of $i_1, i_2, \dots, i_n$ is exactly $\alphx$. (The same applies to Y and Z.) For each triple $(X_I, Y_J, Z_K)$ satisfying (1) $i_t + j_t + k_t = 4$ for all $t \in [n]$ and (2) the distribution $[(i_t, j_t, k_t)]_{t \in [n]}$ equals $\alpha$, the subtensor of $(\CW_q^{\otimes 2})^{\otimes n}$ over $(X_I, Y_J, Z_K)$ is exactly one copy of $\T^\alpha$. However, as $X_I$ may belong to many such triples, these copies are not independent.

We will zero out some level-2 blocks to obtain independent copies of $\T^\alpha$. By the symmetry of X, Y, and Z variables, we can assume that $\alphx = \alphy
 = \alphz$. So we only focus on Z-variable blocks. Let $\numzblock$ be the number of $Z_K$'s that obey $\alphz$. In order to make the Z blocks of all isomorphic copies of $\T^\alpha$ independent, there can be at most $\numzblock$ copies of them.

In fact, using an elegant construction using hashing and the Salem-Spencer set~\cite{salem1942,Behrend1946}, Coppersmith and Winograd~\cite{coppersmith1987matrix} showed that one could get $\numzblock^{1 - o(1)}$ such independent tensors. The $o(1)$ factor is negligible for our purpose.

\item \textbf{The Concrete Formula for Level-2.} To be more specific, let us spell out the concrete formula for $\T^\alpha$ in level 2. We define $\mathrm{sym}(T_{i,j,k}^{\alpha(i,j,k)n}) = T_{i,j,k}^{\alpha(i,j,k)n} \otimes T_{k,i,j}^{\alpha(k,i,j)n} \otimes T_{j,k,i}^{\alpha(j,k,i)n}$ and $\mathrm{sym}_6(T_{i,j,k}^{\alpha(i,j,k)n}) =\mathrm{sym}(T_{i,j,k}^{\alpha(i,j,k)n}) \otimes \mathrm{sym}(T_{j,i,k}^{\alpha(j,i,k)n})$. So we have
\begin{equation}\T^\alpha = \mathrm{sym}\big(T_{0,0,4}^{\otimes \alpha(0,0,4)n}\big) \otimes \mathrm{sym}_6\big(T_{0,1,3}^{\otimes \alpha(0,1,3)n}\big)\otimes \mathrm{sym}\big(T_{0,2,2}^{\otimes \alpha(0,2,2)n}\big) \otimes \mathrm{sym}\big(T_{1,1,2}^{\otimes \alpha(1,1,2)n} \big). \label{equ:Talpha} \end{equation}

Generally, for any tensor $T$, we define $\mathrm{sym}(T)$ as the tensor product of three tensors obtained by rotating $T$'s X, Y, and Z variables. We will need such a notation later.

\item \textbf{The Analysis for Level-1.} In this level, we are given many independent copies of $\T^\alpha$, and our goal is to obtain independent matrix multiplication tensors. Note that $\T^\alpha$ itself is not a matrix multiplication tensor because $\mathrm{sym}\big(T_{1,1,2}^{\otimes \alpha(1,1,2) n}\big)$ is not a matrix multiplication tensor. So now we will focus on $T_{1,1,2}$ only. Let $m = \alpha(1,1,2)n$ for convenience.

We further zero out $\mathrm{sym}\big(T_{1,1,2}^{\otimes m}\big)$ into independent matrix multiplication tensors. We select a distribution $\splres^{(1,1,2)}$ over $\{(0,0,2), (0,1,1),(1,0,1), (1,1,0)\}$ which we call the \underline{split distribution} of $(1,1,2)$. (Following the same spirit as before, we only need to consider one typical $\splres$.) We call it the split distribution because we will split $T_{1,1,2}^{\otimes m}$ into isomorphic copies of
\begin{equation}
\T_{1,1,2}^{\splres, m} \; \defeq \bigotimes_{\hat{i}+\hat{j}+\hat{k}=2} \bk{T_{\, \hat{i}, \, \hat{j}, \, \hat{k}} \otimes T_{1 - \hat{i}, \, 1 - \hat{j}, \, 2 -\hat{k}}}^{\otimes \splres^{(1,1,2)} \! (\hat{i}, \ \hat{j}, \ \hat{k}) \,\cdot\, m}.
\label{equ:tT112}
\end{equation}

Here, each $\T_{1,1,2}^{\splres, m}$ is a subtensor of $T_{1,1,2}^{\otimes m}$ that splits according to $\splres^{(1,1,2)}$. A single $T_{1,1,2}^{\otimes m}$ contains many such isomorphic copies. We also get $\splres^{(1,2,1)}$ and $\splres^{(2,1,1)}$ from the symmetry under rotation, which defines $\T_{1,2,1}^{\splres, m}$ and $\T_{2,1,1}^{\splres, m}$.

\bigskip

For concreteness, let us first specify the minimum number of parameters that determine $\splres^{(1,1,2)}$. Since there is a symmetry between $T_{0,0,2} \otimes T_{1,1,0}$ and $T_{1,1,0} \otimes T_{0,0,2}$, we can w.l.o.g$.$ assume that $\splres^{(1,1,2)}(0,0,2) = \splres^{(1,1,2)}(1,1,0) = \mu$ and $\splres^{(1,1,2)}(0,1,1) = \splres^{(1,1,2)}(1,0,1) = 1/2 - \mu$ for some $0 \leq \mu \leq \frac{1}{2}$. 

To obtain many independent matrix multiplication tensors from $T^{\alpha}$ (defined in \eqref{equ:Talpha}), we now degenerate the only term in it that is not a matrix multiplication tensor, $\mathrm{sym}\big(T_{1,1,2}^{\otimes m}\big)$, into independent matrix multiplication tensors. 

We zero out $\mathrm{sym}\big(T_{1,1,2}^{\otimes m}\big)$ into many independent subtensors isomorphic to
\begin{align}
  \mathrm{sym} \big(\T_{1,1,2}^{\splres, m}\big) = \mathrm{sym} \Big(
  &\left(T_{0,0,2} \otimes T_{1,1,0} \right)^{\otimes \mu m} \otimes  \left(T_{1,1,0} \otimes T_{0,0,2}  \right)^{\otimes \mu m}  \otimes \notag \\
  &\left(T_{0,1,1} \otimes T_{1,0,1} \right)^{\otimes (1/2-\mu) m} \otimes \left(T_{1,0,1} \otimes T_{0,1,1} \right)^{\otimes (1/2-\mu) m} \Big).
\end{align}

In $T_{1,2,1}^{\otimes m}$ and $T_{2,1,1}^{\otimes m}$, the marginal distribution for Z-variables is $\splresZ^{(1,2,1)}(0) = \splresZ^{(1,2,1)}(1) = \splresZ^{(2,1,1)}(0) = \splresZ^{(2,1,1)}(1) = 1/2$. In $T_{1,1,2}^{\otimes m}$, the marginal for Z-variables is $\splresZ^{(1,1,2)}(0) = \splresZ^{(1,1,2)}(2) = \mu, \; \splresZ^{(1,1,2)}(1) = 1 - 2\mu$.

Denote by $\numzblock^{(1,1,2)}$ the number of level-1 Z-variable blocks in $T_{1,1,2}^{\otimes m}$ obeying $\splresZ^{(1,1,2)}$, and similar for $(1,2,1)$ and $(2,1,1)$. With the same construction as level-$2$, there is a way to zero out $\mathrm{sym}\big(T_{1,1,2}^{\otimes m}\big)$, and get $\big(\numzblock^{(1,1,2)} \numzblock^{(1,2,1)} \numzblock^{(2,1,1)}\big)^{1 - o(1)}$ many such independent isomorphic copies of $\mathrm{sym} \big(\T_{1,1,2}^{\splres, m}\big)$.

\item \textbf{Putting everything together.} We first zero out $\mathrm{CW}_q^{\otimes 2n}$ to isomorphic copies of $\T^\alpha$ and then to independent isomorphic copies of tensor products of matrix multiplication tensors (which is still a matrix multiplication tensor). We can compute
  \[
    \numzblock^{1-o(1)} \cdot \left(\numzblock^{(1,1,2)} \numzblock^{(1,2,1)} \numzblock^{(2,1,1)}\right)^{1-o(1)}
  \]
  many independent matrix multiplication tensors with $(q + 2)^{2n}$ many multiplications. The distributions $\alpha$ and $\splres^{(1,1,2)}$ are carefully chosen to balance between the size of each matrix multiplication tensor and the total number of them.
\end{itemize}

\paragraph*{Merging after splitting.} In the analysis above, we did not further zero out $\mathrm{sym}\left(T_{0,j,k}\right)^{\otimes \alpha(0, j, k) n}$  ($j + k = 4$) in level-1 because $T_{0,j,k}$ (without symmetrization) is already a matrix multiplication tensor. If we instead use the same approach as the analysis of $T_{1,1,2}$, i.e., splitting and zeroing out into independent matrix multiplication tensors, we would get a much worse result. Intuitively, the reason is that, if we insist on (1) splitting it into non-independent matrix multiplication subtensors and (2) then zeroing out the subtensors into independent ones, we would have wasted all the terms we zeroed out. Here, we can avoid such waste by the fact that $T_{0,j,k}$ is already a matrix multiplication tensor. 

Essentially, this can be equivalently viewed as (1) splitting them into non-independent matrix multiplication subtensors and (2) merging these non-independent subtensors back into a single matrix multiplication tensor. As pointed out by \cite{ambainis2015fast}, such merging is the reason why higher power analyses improve the bound of $\omega$.

In our algorithm, we cannot view $T_{0,j,k}$ as a large matrix multiplication, since we need its split distribution $\tilde \alpha^{(0,j,k)}$. Hence we have to split it and merge it back. This gives a result as good as directly treating $T_{0,j,k}$ as a single matrix multiplication tensor.

\bigskip

Let us simply look at $T_{0,2,2}$ as an example. Let $m = \alpha(0, 2, 2) n$. We know that $T_{0,2,2}^{\otimes m}$ is isomorphic to the matrix multiplication tensor $\langle 1, 1, (q^2 + 2)^m\rangle$.

Now we are going to split it. First, choose a distribution $\splres^{(0,2,2)}$ over $\{(0,0,2),(0,2,0),(0,1,1)\}$. We pick $\splres^{(0,2,2)}(0,0,2) = \splres^{(0,2,2)}(0,2,0) = \lambda$, and $\splres^{(0,2,2)}(0,1,1) = 1 - 2\lambda$ for some $0 \leq \lambda \leq \frac{1}{2}$. Then, the tensor $T_{0,2,2}^{\otimes m}$ alone contains $\numzerotwotwo = \binom{m}{\lambda m, \lambda m, (1 - 2\lambda)m}$ many (non-disjoint) isomorphic copies of
\begin{equation} T^{\splres,m}_{0,2,2} = \left(T_{0,0,2} \otimes T_{0,2,0}\right)^{\otimes \lambda m} \otimes \left(T_{0,2,0} \otimes T_{0,0,2}\right)^{\otimes \lambda m} \otimes (T_{0,1,1} \otimes T_{0,1,1})^{\otimes (1 - 2\lambda)m}. \label{equ:tT022}\end{equation}

We need the fact that each $T_{0,0,2}$ is isomorphic to the matrix multiplication tensor $\langle 1, 1, 1\rangle$ while each $T_{0,1,1}$ is isomorphic to the matrix multiplication tensor $\langle 1, 1, q\rangle$. This implies each $T^{\splres,m}_{0,2,2}$ is isomorphic to the matrix multiplication tensor $\langle 1,1,q^{2(1 - 2\lambda)}\rangle$. 

After zeroing out all the X, Y, and Z variables that do not obey our choice of $\splres^{(0,2,2)}$, we merge all these $\numzerotwotwo$ non-disjoint isomorphic copies back into a matrix multiplication tensor $\langle 1, 1, \numzerotwotwo \cdot q^{2(1 - 2\lambda)} \rangle$. Taking $\lambda = \frac{1}{2 + q^2}$, we get $\langle 1, 1, \numzerotwotwo \cdot q^{2(1 - 2\lambda)} \rangle = \langle1,1,(q^2 + 2)^{m (1- o(1))}\rangle$, which is as good as $\langle 1, 1, (q^2 + 2)^m\rangle$ when $m$ goes to infinity.

\subsection{Combination Loss}\label{sec:combination-loss}

The key insight in our paper is that this analysis is actually wasteful. To see this, let us first summarize all the zeroing-outs we performed in the above analysis. 

\paragraph*{All zeroing-outs.}Taking the merging point of view for $T_{0,j,k}$'s, we get the following procedure:
\begin{itemize}
\item In level-2, we first zero out all the blocks $X_I$ that do not obey $\alphx$. Same for Y and Z-blocks. Then we further zero out some level-2 blocks according to hashing and the Salem-Spencer set. For each remaining block triple $(X_I, Y_J, Z_K)$, the subtensor of $(\CW_q^{\otimes 2})^{\otimes n}$ over $X_I, Y_J, Z_K$ gives one isomorphic copy of $\T^\alpha$ (defined in \eqref{equ:Talpha}). Let us call this subtensor $\T_{I,J,K}$.

\item Fix any $X_I, Y_J, Z_K$ and the corresponding subtensor $\T_{I,J,K}$. Consider level-1 blocks $X_{\hat{I}} \in X_I, \; Y_{\hat{J}} \in Y_J, \; Z_{\hat{K}} \in Z_K$. For sequence $\hat{K}$, we define its \underline{split distribution over set $S$} as
$$\split(\hat{K},S)(\hat{k}', \hat{k}'') = \frac{1}{|S|}\left| \big\{t \in S \mid \hat{k}_{2t - 1} = \hat{k}', \; \hat{k}_{2t} = \hat{k}''\big\}\right|.$$
Let $S_{i,j,k} = \{t \mid i_t = i, j_t = j, k_t = k\}$ be the set of positions that belong to $T_{i,j,k}$. We say $Z_{\hat{K}}$ obeys $\splresZ$ if and only if for all $i,j,k$, we have $\split(\hat{K}, S_{i,j,k})(\hat{k}', k - \hat{k}') = \splresZ^{(i,j,k)}(\hat{k}')$.\footnote{For simplicity, we write $\split(\hat{K}, S_{i,j,k}) = \splresZ^{(i,j,k)}$ when there is no ambiguity.} (Since we are taking the merging viewpoint, $T_{0,1,1}$ and $T_{0,0,2}$ have their split distributions as well.)

\bigskip

First, we zero out all the blocks $Z_{\hat{K}}$ that do not obey $\splresZ$. Same for the X and Y blocks. Then we further zero out some level-1 blocks according to hashing and the Salem-Spencer set. Note that this hashing is only over indices in $S_{1,1,2} \cup S_{1,2,1} \cup S_{2,1,1}$ because $T_{0,j,k}$'s are handled differently.

The subtensor of $\T_{I,J,K}$ over each remaining triple $(X_{I}, Y_{J}, Z_{K})$ is an isomorphic copy of 
\[\T^{\splres} \defeq \bigotimes_{i + j + k = 4} \T^{\splres, \alpha(i,j,k)n}_{i,j,k},\]
where $\T_{1,1,2}^{\splres, m}$ (and $\T_{1,2,1}^{\splres, m}$, $\T_{2,1,1}^{\splres, m}$) is defined in \eqref{equ:tT112} and other $\T_{0,j,k}^{\splres, m}$'s are defined in the same way as \eqref{equ:tT022}. Finally, some copies of $\T^{\splres}$ are merged together, and we get independent matrix multiplication tensors.
\end{itemize}

Fix any remaining level-2 triple $M = (X_I, Y_J, Z_K)$. We will show that many level-1 Z-blocks $Z_{\hat{K}} \in Z_K$ are actually not used.

\paragraph*{Which $Z_{\hat{K}}$'s are used?} To answer this, note that in the second step, we zeroed out the blocks $Z_{\hat{K}}$'s that do not obey $\splresZ$. By definition, all remaining $Z_{\hat{K}}$'s satisfy the following:
\begin{equation} \forall i + j + k = 4, \quad \split(\hat{K}, S_{i,j,k}) = \splresZ^{(i,j,k)}. \label{equ:split_condition}\end{equation}
We let $Z^M = \{Z_{\hat{K}} \in Z_K \mid \eqref{equ:split_condition}\text{ holds for }\hat{K}\}$ denote this set of $Z_{\hat{K}}$'s. This is the set of $Z_{\hat{K}}$'s that we actually used in $M = (X_I, Y_J, Z_K)$. We will argue that there are some other level-1 Z-blocks outside $Z^M$, that in a certain sense, is as ``useful'' as those in $Z^M$. To identify these blocks, we first define the \underline{average split distribution} of $k$ as
\begin{equation*}
  \splresavgZ^{(k)}(\hat{k}') = \frac{1}{\alphz(k)} \sum_{i + j = 4 - k} \alpha(i, j, k) \cdot \splresZ^{(i, j, k)}(\hat{k}').
\end{equation*}

Let $S_k = \cup_{i+j=4-k} \, S_{i,j,k}$ be the set of positions where $Z_t = k$. By definition, the condition \eqref{equ:split_condition} implies the following weaker condition which is independent of $I$ and $J$:
\begin{equation}
  \forall k \in \{0,1,\dots, 4\}, \quad \split(\hat{K}, S_{k}) = \splresavgZ^{(k)}. \label{equ:avg_condition}
\end{equation}

Let $Z' = \{Z_{\hat{K}} \in Z_K \mid \eqref{equ:avg_condition}\text{ holds for }\hat{K}\}$ be the set of level-1 Z-blocks that satisfy this weaker condition. Clearly, $Z^M \subseteq Z'$; below we will further show that $|Z^M| = |Z'| \cdot 2^{-\Theta(n)}$. However, all $Z_{\hat{K}} \in Z'$ are ``equivalent'' in the sense that they are isomorphic up to a permutation over $[2n]$. We could take a bold guess: Those $Z_{\hat{K}}$'s in $Z' \setminus Z^M$ should be as useful as those in $Z^M$! We call such ratio $|Z'| / |Z^M|$ the \underline{combination loss}.

\paragraph*{A Closer Look.} How large is the combination loss, $|Z'| / |Z^M|$? In order to affect the bound on $\omega$, the loss needs to be exponentially large. Let us examine the number of ways to split $K$ according to \eqref{equ:split_condition}, and compare it with that of \eqref{equ:avg_condition}. Recall that $S_{i,j,k} = \{t \mid i_t = i, j_t = j, k_t = k\}$ and $S_k = \cup_{i + j = 4 - k} \, S_{i,j,k}$. We use $h(p)$ to denote the binary entropy function $h(p) = - p \log p - (1 - p) \log (1 - p)$.

\begin{itemize}
\item When $k_t \in \{0, 4\}$, there is only one way to split it (i.e., $0 = 0 + 0, \; 4 = 2 + 2$). So such $t$ has no contribution\footnote{Here ``contribution'' means giving a multiplicative factor to $|Z'|$ or $|Z^M|$. $|Z'|$ equals the product of all contributions to it; so does $|Z^M|$.} to either $|Z'|$ or $|Z^M|$. 
\item When $k_t \in \{1, 3\}$, there are two symmetric ways to split it (i.e., $1 = 0 + 1 = 1 + 0$ and $3 = 1 + 2 = 2 + 1$). By this symmetry, we know $\splresZ^{(i,j,k)}$ is simply the uniform distribution, half and half. Fix $k = 1$ or $3$, the contribution to $|Z^M|$ from all $t \in [n]$ such that $k_t = k$ is
  \[\prod_{i+j = 4 - k} \binom{|S_{i,j,k}|}{|S_{i,j,k}| / 2} \approx \prod_{i+j = 4 - k} 2^{|S_{i,j,k}|} = 2^{|S_k|},\]
which equals their contribution to $|Z'|$. So in this case, their contributions to $|Z'|$ and $|Z^M|$ are equal. 

\item The only nontrivial case is when $k_t = 2$. There are three ways to split it: $2 = 0 + 2 = 1 + 1 = 2 + 0$. Taking its symmetry into account, there is still one degree of freedom. Recall that $\splresZ^{(0,2,2)}(0) = \splresZ^{(0,2,2)}(2) = \lambda$ and $ \splresZ^{(0,2,2)}(1) = 1 - 2\lambda$, while $\splresZ^{(1,1,2)}(0) = \splresZ^{(1,1,2)}(2) = \mu, \splresZ^{(1,1,2)}(1) = 1 - 2\mu$.

  Each $t \in S_{0,2,2} \cup S_{2,0,2}$ is either $0 + 2$ or $2 + 0$ with probability $2\lambda$ and is $1 + 1$ with probability $1 - 2 \lambda$. So the logarithm of their contribution to $|Z^M|$ approximately equals the entropy $h(2\lambda) \cdot |S_{0,2,2} \cup S_{2,0,2}|$. Similarly, the logarithm of the contribution of $t \in S_{1,1,2}$ to $|Z^M|$ is approximately $h(2 \mu) \cdot |S_{1,1,2}|$. In total, the contribution of all $t \in S_2$ to $|Z^M|$ is approximately
  $$\exp(h(2\lambda) \cdot |S_{0,2,2} \cup S_{2,0,2}| + h(2 \mu) \cdot |S_{1,1,2}|).$$

On the other hand, let $\bar p = \frac{\lambda |S_{0,2,2} \cup S_{2,0,2}| + \mu |S_{1,1,2}|}{|S_2|}$ be the weighted average of $\lambda$ and $\mu$. (That is, $\splresavgZ^{(2)}(0)=\splresavgZ^{(2)}(2)=\bar p$ and $\splresavgZ^{(2)}(1)=1-2\bar p$.) The contribution of all $t \in S_2$ to $|Z'|$ is approximately $$\exp(h(2 \bar{p}) \cdot |S_2|).$$ This is larger than their contribution to $Z^M$ because the even distribution has the maximum entropy. 

In the analysis of Coppersmith and Winograd, $\lambda = \frac{1}{2 + q^2}$, $\mu = \frac{1}{2 + q^{3 \tau}}$ (where $\tau = \frac{\omega}{3}$). So there is a constant gap between $\lambda$ and $\mu$ when $\omega > 2$. This implies an exponential gap between $|Z'|$ and $|Z^M|$. The combination loss is indeed exponentially large. Hence, compensating for such loss might improve $\omega$. 
\end{itemize}

\subsection{Compensate for Combination Loss}
\label{sec:compensate}

As we discussed above, since a level-2 variable block $Z_K$ is in only one independent copy of $\T^\alpha$, some index $2$'s in it will split according to $\lambda$ while other $2$'s will split according to $\mu$. This causes the combination loss. The same holds for all X, Y and Z dimensions. One natural attempt to compensate for it is to match a level-2 variable block multiple times, i.e., to let it appear in multiple triples.

\paragraph*{Sketch of the Idea.} During the hashing step, we randomly match a level-2 Z-variable block $Z_K$. Any index 2 in $Z_K$ will be in the parts of $T_{0,2,2}$, $T_{2,0,2}$, or $T_{1,1,2}$ randomly. Those index 2 of $Z_K$ in $T_{0,2,2}$ and $T_{2,0,2}$ parts will split according to $\lambda$ while those in $T_{1,1,2}$ part will split according to $\mu$. Even if $Z_K$ is matched multiple times, since each time the positions of these three parts are different, we will be using different level-1 blocks in $Z_K$. This observation allows us to obtain a subtensor that is mostly disjoint from other subtensors from each matching.

\paragraph*{Asymmetric Hashing.} In the original two-level analysis (\Cref{sec:CW}), the marginal distributions $\alphx=\alphy
=\alphz$ are picked to balance the number of level-2 blocks and the number of variables in a block. So after zeroing-out, every remaining block can only be in one triple. We carefully pick $\alphx = \alphy
$ and $\alphz$ so that there will be more level-2 X and Y-blocks than Z-blocks. (This asymmetric hashing method also appears in~\cite{coppersmith1987matrix}, but the base tensor is Strassen's tensor~\cite{strassen1986,pan1978}.) In return, in each level-2 Z-block, we can now have more variables. We will match each Z-variable block to multiple pairs of X and Y-blocks while keeping the matching for each X and Y-block unique, that is, every remaining X or Y-block is only in one triple but a remaining Z-block can be in multiple triples. Such uniqueness for X and Y-blocks is necessary for our method of removing the interfering terms.

\paragraph*{Sanity Check.} As a sanity check, let us first try to match a level-2 variable block $Z_K$ twice. (Recall that we defined the notations for variable blocks in Section \ref{sec:CW}.) We say $Z_K$ can be matched to $X_I$ and $Y_J$ with respect to $\alpha$ if and only if (1) $i_t + j_t + k_t = 4$ for all $t \in [n]$; and (2) the distribution $[(i_t, j_t, k_t)]_{t \in [n]}$ equals $\alpha$.

Suppose $Z_K$ is first matched to $X_I$ and $Y_J$ with respect to $\alpha$. Now among all pairs of level-2 blocks $X_{I'}, Y_{J'}$ that can be matched with $Z_K$ with respect to $\alpha$, we uniformly sample one of them. Then we let $Z_K$ also be matched to that $X_{I'}$ and $Y_{J'}$. We define $M = (X_I, Y_J, Z_K)$ and $M' = (X_{I'}, Y_{J'}, Z_{K})$. The subtensors of $\mathrm{CW}_q^{\otimes 2n}$ over $M$ and $M'$ give two copies of tensor $\T^\alpha$, denoted by $\T^\alpha_M$ and $\T^\alpha_{M'}$, respectively. As explained in Section \ref{sec:combination-loss}, zeroing out in $\T^\alpha_M$ (resp$.$ $\T^\alpha_{M'}$) would give us its subtensor $\T^{\splres}_M$ (resp$.$ $\T^{\splres}_{M'}$). If we actually perform these two zeroing-out procedures one after another, they would interfere with each other. So we instead just keep all the level-1 X, Y, Z variable blocks in the union of the supports of $\T^{\splres}_M$ and $\T^{\splres}_{M'}$ and zero out other level-1 blocks. \underline{We call the remaining subtensor after such zeroing out $\T^{\splres}_{M \cup M'}$.}

\begin{figure}
\centering
\includegraphics[width=0.9\textwidth]{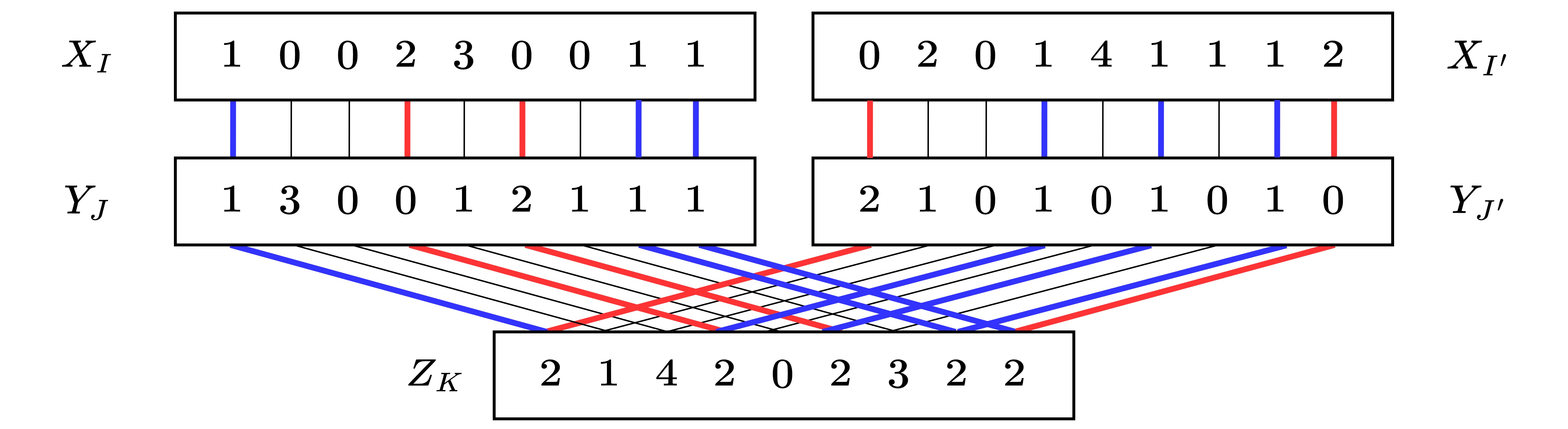}
\caption{Matching level-2 variable block $Z_K$ twice. The positions of $T_{0,2,2}/T_{2,0,2}$'s are marked with red color, while the positions of $T_{1,1,2}$'s are marked with blue. Suppose the index starts from $1$. In this case, $S_{1,1,2} = \{1, 8, 9\}$ and $S'_{1,1,2} = \{4, 6, 8\}$.}
\label{fig:matching}
\end{figure}

Let $Z^M$ be the set of level-1 Z-variable blocks in $\T^{\splres}_M$, and let $Z^{M'}$ be that of $\T^{\splres}_{M'}$. Our key observation is that $Z^M$ and $Z^{M'}$ are mostly disjoint. Recall that $S_{i,j,k} = \{t \in [n] \mid (i_t, j_t, k_t) = (i,j,k)\}$ is the positions of the $T_{i,j,k}$ part in $M$. Similarly, let $S'_{i,j,k}$ be that of $M'$. (See Figure \ref{fig:matching}.) For any set $S \subseteq [n]$ and a level-1 Z-variable block $Z_{\hat{K}} = Z_{(\hat{k}_1, \hat{k}_2, \dots, \hat{k}_{2n})}$, we use $\split(\hat{K},S)$ to denote the split distribution for $Z_{(\hat{k}_1, \hat{k}_2, \dots, \hat{k}_{2n})}$ restricted to set $S$. Recall that it is defined as
\[\split(\hat{K},S)(k', k'') = \frac{1}{|S|}\left| \big\{t \in S \; \big| \; \hat{k}_{2t - 1} = k', \hat{k}_{2t} = k''\big\}\right|.\]
 Fix any $Z_{\hat{K}} \in Z^M$. We claim that, by the randomness of $S'_{1,1,2}$, w.h.p$.$ $\split(\hat{K},S'_{1,1,2}) = \splresavgZ^{(2)}$ which is the average split distribution. On the contrary, for any $Z_{\hat{K}'} \in Z^{M'}$, we must have $\split(\hat{K}',S'_{1,1,2}) = \splresZ^{(1,1,2)}$. Since $\splresavgZ^{(2)} \ne \splresZ^{(1,1,2)}$, this shows w.h.p$.$ $Z_{\hat{K}} \not\in Z^{M'}$, which implies that $Z^M$ and $Z^{M'}$ are mostly disjoint.

Now let us justify our claim. Since we uniformly sampled the pair $(X_{I'}, Y_{J'})$, by symmetry, $S'_{1,1,2}$ is uniform among all $\alpha(1,1,2)n$-sized subsets of $S_2 = \{t \in [n] \mid k_t = 2\}$. Regardless of whether a position $t \in S_2$ is in $S_{0,2,2}, S_{2,0,2},$ or $S_{1,1,2}$, such position is in $S'_{1,1,2}$ with equal probability. As we fixed $Z_{\hat{K}} \in Z^M$, the corresponding split distribution $\split(\hat{K},S'_{1,1,2})$ is mostly likely to be the weighted average of $\splresZ^{(0,2,2)}$, $\splresZ^{(2,0,2)}$, and $\splresZ^{(1,1,2)}$, i.e., the average distribution $\splresavgZ^{(2)}$.

The idea above is generalized in our main algorithm so that each level-2 Z-block can be matched in $2^{\Theta(n)}$ triples. However, there are two remaining challenges:

\begin{itemize}
\item In order to be independent, being mostly disjoint is not sufficient. $\T^{\splres}_M$ and $\T^{\splres}_{M'}$ has to be completely disjoint. An easy fix would be zeroing out the intersecting variables. But this introduces missing Z-variables in the final matrix multiplication tensors we get. We use a random shuffling technique similar to that of \cite{doi:10.1137/1.9781611975482.31} to fix these ``holes'' in Z-variables.
\item Moreover, even being perfectly disjoint does not guarantee independence. There is also a second condition in the definition of independence: no additional terms, i.e., we have to make sure that we get exactly $\T^{\splres}_M + \T^{\splres}_{M'}$ without any extra terms. For example, let $X_{\hat I'}, Y_{\hat J'}, Z_{\hat K}$ be level-1 blocks in $X_{I'}$, $Y_{J'}$, and $Z^M$, respectively, then the subtensor of $\T^{\splres}_{M \cup M'}$ over these level-1 blocks should be zero. Vice versa for level-1 blocks in $X_I, Y_J$, and $Z^{M'}$. If these conditions are not satisfied, $T_M^{\tilde\alpha}$ and $T_{M'}^{\tilde\alpha}$ are not independent.
\end{itemize}

\paragraph*{Fixing Holes.} We now address the first challenge. Suppose we have a broken matrix multiplication tensor of size $\bar N \times \bar M \times \bar P$ (i.e., it corresponds to the matrix multiplication of an $\bar N \times \bar M$ matrix and an $\bar M \times \bar P$ matrix), in which half of the Z-variables are zeroed out. Let us denote the two multiplying matrices as $X$ and $Y$. The product is $(XY)_{i,j} := \sum_{k = 1}^{\bar M} X_{i,k} Y_{k,j}$, but we can only get half of the entries in $XY$. If we randomly select three permutations $\pi_1, \pi_2, \pi_3$ of $[\bar N], [\bar M], [\bar P]$, respectively, and fill in $X_{i,k} \gets A_{\pi_1\!(i), \pi_2\!(k)}, \; Y_{k,j} \gets B_{\pi_2\!(k), \pi_3\!(j)}$ instead, we would get the correct answer to a random half of the entries in $AB$. Then we just repeat this multiple times using other broken matrix multiplication tensors. Combine the answers to the entries in $AB$ together, we would get the correct answer for $AB$ with high probability. In other words, we solve the first challenge by ``gluing'' many randomly permuted broken matrix multiplication tensors together.

\paragraph*{Compatibility.} For the second challenge, we need the following key observation. Let $X_{\hat{I}'} \in X_{I'}, \; Y_{\hat{J}'}\in Y_{J'}, \; Z_{\hat{K}} \in Z^{M}$ be three level-1 blocks. If there is an interfering term involving variables in $X_{\hat{I}'}$, $Y_{\hat{J}'}$ and $Z_{\hat{K}}$, for those $t \in S'_{2,0,2}$, we must have $(\hat{i}'_{2t - 1}, \hat{i}'_{2t}) = (2 - \hat{k}_{2t - 1}, \, 2 - \hat{k}_{2t})$, because $\hat{j}'_{2t-1} = \hat{j}'_{2t} = 0$ and $(\hat{i}'_{2t - 1}, \hat{i}'_{2t}) + (\hat{j}'_{2t - 1}, \hat{j}'_{2t}) + (\hat{k}_{2t - 1}, \hat{k}_{2t}) = (2,2)$. This implies
\[
  \Split(\hat{I}', S'_{2,0,2})(a, 2 - a) = \Split(\hat{K}, S'_{2,0,2})(2 - a, a) \qquad \textup{for all} \quad a \in \{0,1,2\}.
  \numberthis \label{eq:split202agree}
\]
If this holds, we say $\Split(\hat I', S'_{2,0,2})$ agrees with $\Split(\hat K, S'_{2,0,2})$.

We claim that for any fixed $X_{\hat I'} \in X_{I'}$, $Y_{\hat J'} \in Y_{J'}$, and $Z_{\hat K} \in Z^{M}$ that are retained in $\T^{\splres}_{M \cup M'}$, the condition \eqref{eq:split202agree} is not satisfied with high probability:
\begin{enumerate}
\item $\Split(\hat I', S'_{2,0,2})$ agrees with $\splresZ^{(2,0,2)}$, the split distribution of Z indices for $T_{2,0,2}$. This is a necessary condition for $X_{\hat I'}$ to form a triple with $Z_{\hat K} \in Z^{M'}$, because $\Split(\hat K, S'_{2,0,2}) = \splresZ^{(2,0,2)}$ holds for all $Z_{\hat K} \in Z^{M'}$. Otherwise, if this condition does not hold, $X_{\hat I'}$ cannot form any triple with Z-blocks $Z_{\hat K} \in Z^{M'}$, so it has been zeroed out before forming $\T^{\splres}_{M \cup M'}$. (In this argument, it is crucial that $X_{I'}$ is only matched in a unique triple $M'$.)
\item $\Split(\hat K, S'_{2,0,2})$ is equal to $\splresavgZ^{(2)}$, the average split distribution of index 2, with high probability. This can be deduced from a similar argument as in the sanity check.
\end{enumerate}
These two arguments, combined with the fact that $\splresavgZ^{(2)} \ne \splresZ^{(2,0,2)}$, concludes that \eqref{eq:split202agree} is unlikely to hold.

Similarly, one could also look at the positions $S'_{0,2,2}$, which is symmetric to the case above.\footnote{The same argument would not work with $t \in S'_{1,1,2}$, because now $(\hat{j}'_{2t - 1}, \hat{j}'_{2t})$ could be either $(1,0)$ or $(0,1)$. This degree of freedom prevents us from getting a compatibility constraint between $\split(\hat{I}',S'_{1,1,2})$ and $\split(\hat{K},S'_{1,1,2})$.} For any level-2 blocks $X_{I}$, ${Y}_{J}$, and level-1 block $Z_{\hat{K}} \in Z^M$ where $M = (X_I, Y_J, Z_K)$, when $\split(\hat{K},S_{2,0,2})$ and $\split(\hat{K},S_{0,2,2})$ are both equal to $\splresZ^{(2,0,2)}$ ($S_{2,0,2}$ and $S_{0,2,2}$ are defined w.r.t$.$ $M$), we say $Z_{\hat{K}}$ is \underline{compatible} with $X_{I}$ and ${Y}_{J}$. Our argument above shows that w.h.p$.$ a level-1 block will only be compatible with one pair of $ X_{I}$ and ${Y}_{J}$, and there is no interfering term involving $Z_{\hat{K}}$ and $ X_{I}$, ${Y}_{J}$ if they are not compatible. If $Z_{\hat{K}}$ happens to be compatible with two pairs, we will zero out $Z_{\hat{K}}$ and leave it as a hole. Then it will be fixed by our hole-fixing technique.  See \cref{sec:2nd} and \cref{sec:hole_lemma} for more details.

\subsection{Beyond the Second Power}

In the analysis of higher and higher powers of the CW tensor, because we can perform merging for $T_{0,j,k}$ at each level, we get better and better upper bounds of $\omega$. Together with such gain, we also incur combination loss at each level. Our approach generalizes to high powers as well. For higher powers of the CW tensor, we apply this method to analyze both the global value (i.e., the value of the CW tensor) and the component values (e.g., the values of $T_{1,2,5},T_{1,3,4},T_{2,2,4},T_{2,3,3}$ and their permutations for the 4th power). Same as in previous works, we optimize all parameters by a computer program to obtain an upper bound of $\omega$.

\paragraph{Fixing holes in tensors.} Recall that in our two-level analysis, we obtained a subtensor $\T^{\splres}_M$ for each retained level-2 triple $M = (X_I, Y_J, Z_K)$, in which some Z-blocks $Z_{\hat K} \in Z^M$ are zeroed out and become holes. We first pretend there are no holes, zeroing out these tensors $\T^{\splres}_M$ to form matrix multiplication tensors, and then fix the holes in the matrix multiplication tensors. However, for higher powers, this approach meets a difficulty: The process of transforming $\T^{\splres}_M$ to matrix multiplication tensors is more involved\footnote{It will be a more general degeneration instead of just zeroing out.}, so it is not clear how one can control the number of holes in the final matrix multiplication tensors.

To solve this difficulty, we will fix the holes in $\T^{\splres}_M$ before transforming them to matrix multiplication tensors. Suppose all $\T^{\splres}_M$ are isomorphic to some tensor $\T^*$ except for the holes in their Z-variables. As long as $\T^*$ has a desired symmetric structure, we can shift the holes in each copy $\T^{\splres}_M$ to random places, like we did in \cref{sec:compensate} to repair the holes in matrix multiplication tensors. Then, by gluing several copies together, we can fix all the holes, resulting in a copy of $\T^*$ without holes. In \cref{sec:hole_lemma}, we will show how to fix the holes: We will first define the desired tensor structure $\T^*$ which will naturally appear in \cref{sec:global,sec:component}, then define a group of permutations used to move the holes, and finally fix the holes in $\T^*$ using the idea we discussed above.

\paragraph{Compatibility for higher powers.} Recall that in \cref{sec:compensate}, we defined the compatibility for the second power. A level-1 Z-block $Z_{\hat K}$ is compatible with $X_{I}, Y_{J}$, if the following conditions hold:
\begin{itemize}
\item $\split(\hat K, S_{0,2,2}) = \split(\hat K, S_{2,0,2}) = \splresZ^{(0,2,2)}$, where $S_{0,2,2}$ and $S_{2,0,2}$ are defined with respect to the triple $(X_I, Y_J, Z_K)$;
\item $\split(\hat K, S_2) = \splresavgZ^{(2)}$.
\end{itemize}
Here, we can make a constraint for the split distribution of $T_{0,2,2}$ because the index 0 ensures a one-to-one correspondence between the split distributions of Y and Z-indices. Similar for $T_{2,0,2}$. We generalize this definition to higher powers, creating a similar constraint for every component $T_{i,j,k}$ with $i = 0$ or $j = 0$:
\begin{itemize}
\item For all components $T_{i,j,k}$ with $i = 0$ or $j = 0$, $\split(\hat K, S_{i,j,k}) = \splresZ^{(i,j,k)}$, where $S_{i,j,k}$ is defined with respect to the triple $(X_I, Y_J, Z_K)$.
\item $\split(\hat K, S_k) = \splresavgZ^{(k)}$ for all $k$.
\end{itemize}
Based on this definition, we will show in \cref{sec:global} that there are no interfering terms between $Z_{\hat K}$ and $X_I, Y_J$ that are incompatible. Once some $Z_{\hat K}$ is compatible with two remaining triples, we zero it out as we did in \cref{sec:compensate}.

The other steps for higher powers are similar to \cref{sec:compensate}: First, we choose a distribution $\alpha$ over level-$\lvl$ components $T_{i,j,k}$ ($i + j + k = 2^\lvl$). Second, we apply the asymmetric hashing method to obtain many triples that do not share X or Y-blocks (sharing Z-blocks is allowed). Then, if some level-$\lastlvl$ Z-block $Z_{\hat K}$ is compatible with two remaining triples, it is zeroed out. (Several additional zeroing-out steps are applied for technical reasons.) Finally, all triples become independent, then we fix the holes and degenerate each triple independently, forming a desired lower bound of the value.

\paragraph{Analyzing component values.} Our approach is applied to obtain not only the value of the CW tensor $\CW_q$, but also the value of high-level components $T_{i,j,k}$, which is a subtensor of the CW tensor power. One more complication rises in this scenario: Denote by $\numxblock, \numyblock, \numzblock$ be the number of X, Y, and Z-blocks before the hashing process. The asymmetric method requires $\numxblock = \numyblock > \numzblock$ holds. However, if we apply the hashing method on
\[
  \sym(T_{i,j,k}^{\otimes n}) \defeq T_{i,j,k}^{\otimes n} \otimes T_{j,k,i}^{\otimes n} \otimes T_{k,i,j}^{\otimes n}
\]
like in previous works, we would have $\numxblock = \numyblock = \numzblock$ since X, Y, Z variables are symmetric. Our solution is to pick three different numbers $A_1, A_2, A_3$ and apply hashing on
\begin{align*}
  \T \coloneqq{} & T_{i,j,k}^{\otimes A_1n} \otimes T_{j,k,i}^{\otimes A_2n} \otimes T_{k,i,j}^{\otimes A_3n} \otimes{} \\
                 & T_{j,i,k}^{\otimes A_1n} \otimes T_{k,j,i}^{\otimes A_2n} \otimes T_{i,k,j}^{\otimes A_3n}.
\end{align*}
This makes Z-variables asymmetric from X and Y. Starting with this tensor $\T$ and performing similar analysis as in \cref{sec:global}, we obtain desired lower bounds for the components.

\paragraph*{Why did we break the 2.3725 lower bound?} By compensating for the combination loss at each level, we get the new upper bound of $\omega < 2.37187$ from the eighth power of the CW tensor. In the paper by Ambainis, Filmus, and Le Gall \cite{ambainis2015fast}, they proved that certain algorithms could not give a better upper bound than $\omega < 2.3725$. These algorithms include previous improvements in analyzing higher powers and the refined laser method \cite{alman2021}. Our algorithm is the first algorithm to break this lower bound.

In their lower bound, they start with an estimated partitioned tensor, which is a partitioned tensor with an estimated value for each subtensor. Then they defined the merging value of it, which, roughly speaking, is the maximum value one can get from (1) zeroing it out into independent subtensors; and (2) merging non-independent subtensors that are matrix multiplication tensors into larger matrix multiplications.

Then they start from the tensor $\mathrm{CW}_{q}^{\otimes 16}$. The values of its subtensors are estimated using previous algorithms \cite{williams2012, legall2014} but ignoring a penalty term: the penalty term arises in the laser method when multiple joint distributions correspond to the same marginal distribution. They proved that from such a starting point, the merging value of $\mathrm{CW}_{q}^{\otimes 16}$ could not give a better bound than $\omega < 2.3725$. Since the penalty is exactly the term that the refined laser method improved, their lower bound also applies to the refined laser method.

Our algorithm uses the same variable block multiple times starting from level 2, and this would already give improved value for level-3 subtensors $T_{1,2,5},T_{1,3,4},T_{2,2,4}$ and $T_{2,3,3}$. We break their lower bound because our lower bounds for values of the subtensors of $\mathrm{CW}_{q}^{\otimes 16}$ are better than the upper bounds they used for the estimation. 

However, at the end of the day, our algorithm is still zeroing out level-1 variable blocks and merging non-independent matrix multiplications. So it is still subject to their second lower bound starting from $\mathrm{CW}_{q}$, which says that such algorithms cannot prove $\omega < 2.3078$~\cite{ambainis2015fast}.

\section{Preliminaries}
\label{sec:prelim}
In this paper, $\log x$ means $\log_2 x$ by default, and $[n]=\{1,\ldots,n\}$. For a sequence $a_1,\ldots,a_k$ which sums to 1, $\binom{n}{a_1n, \, \cdots, a_kn}$ can be written as $\binom{n}{[a_i n]_{i\in[k]}}$, or simply $\binom{n}{[a_i n]}$ if there is no ambiguity. The notation $A_1\sqcup A_2\sqcup\cdots\sqcup A_k$ means the disjoint union of sets $A_1, \ldots, A_k$.

We use $\ind[P]$ as the indicator function in this paper. For a statement $P$,
\[
\ind[P] = \begin{cases}
  1 & \textup{if $P$ is true}, \\
  0 & \textup{otherwise}.
\end{cases}
\]
So $\sum_{x}\ind[P(x)]=\big|\{x|P(x)\}\big|$, that is, the number of $x$ satisfying the statement $P(x)$.

Most of our notations about tensors are similar to \cite{alman2021}.

\subsection{Tensors, Operations, and Ranks}

\paragraph*{Tensors.} Let $X = \{x_1, x_2, \dots, x_n\}$, $Y = \{y_1, y_2, \dots, y_m\}$, and $Z = \{z_1, z_2, \dots, z_p\}$ be three sets of variables. A \emph{tensor} $T$ over $X, Y, Z$ and field $\mathbb{F}$ is the summation
\[T \coloneqq \sum_{i = 1}^n \sum_{j=1}^m \sum_{k=1}^p a_{i,j,k} \cdot x_i y_j z_k,\]
where all $a_{i,j,k}$'s are from field $\mathbb{F}$. %

The \emph{matrix multiplication tensor} $\langle n, m, p\rangle$ is a tensor over sets $\{x_{i,j}\}_{i \in [n], j \in [m]}$, $\{y_{j,k}\}_{j \in [m], k \in [p]}$, and $\{z_{k,i}\}_{k \in [p], i \in [n]}$, defined as
\[ \langle n, m, p\rangle \coloneqq \sum_{i=1}^n \sum_{j=1}^m \sum_{k=1}^p x_{i,j} \cdot y_{j,k} \cdot z_{k,i}. \]

\paragraph*{Tensor isomorphisms.} If two tensors $T$ and $T'$ are equal up to a renaming of their variables, we say they are \emph{isomorphic} or \emph{equivalent}, denoted as $T \cong T'$. Formally, we have the following definition.

\begin{definition}
  Let $T$ be a tensor over variables $X, Y, Z$, written as
  \[T = \sum_{x \in X} \sum_{y \in Y} \sum_{z \in Z} a_{x,y,z} \cdot xyz,\]
  and $T'$ be a tensor over variables $X', Y', Z'$. We say $T$ is \emph{isomorphic to} $T'$ if there are bijections
  \[ \phi_1: X \to X',\qquad \phi_2: Y \to Y',\qquad \phi_3: Z \to Z', \]
  satisfying
  \[
    T' = \sum_{x \in X} \sum_{y \in Y} \sum_{z \in Z} a_{x,y,z} \cdot \phi_1(x) \cdot \phi_2(y) \cdot \phi_3(z).
  \]
  $\phi = (\phi_1, \phi_2, \phi_3)$ is called an \emph{isomorphism} from $T$ to $T'$, denoted by $\phi\big(T\big) = T'$.
\end{definition}

Moreover, an isomorphism $\phi$ from $T$ to itself is called an \emph{automorphism} of $T$. All automorphisms of $T$ form a group, called the \emph{group of automorphisms} of $T$, written $\Aut(T)$.

\paragraph*{Tensor operations.} Let $T$ and $T'$ be two tensors over $X, Y, Z$ and $X', Y', Z'$, respectively, written as
\[
  \displaystyle T = \sum_{x_i \in X}  \sum_{y_j \in Y} \sum_{z_k \in Z} a_{i,j,k} \cdot x_i y_j z_k, \qquad \displaystyle T' = \sum_{x'_{i'} \in X'} \sum_{y'_{j'} \in Y'} \sum_{z'_{k'} \in Z'} a'_{i',j',k'} \cdot x'_{i'} y'_{j'} z'_{k'}.
\]
We define the following operations:
\begin{itemize}
\item The \emph{direct sum} $T \oplus T'$ is a tensor over $(X \sqcup X'), (Y \sqcup Y')$, and $(Z \sqcup Z')$, defined as
  \[
    T \oplus T' \defeq \sum_{x_i \in X} \sum_{y_j \in Y} \sum_{z_k \in Z} a_{i,j,k} \cdot x_i y_j z_k + \sum_{x'_{i'} \in X'}  \sum_{y'_{j'} \in Y'} \sum_{z'_{k'} \in Z'} a'_{i',j',k'} \cdot x'_{i'} y'_{j'} z'_{k'}.
  \]
  When performing the direct sum $T \oplus T'$, we always regard the variables in $T$ and $T'$ as distinct variables.
  Specifically, $T^{\oplus n}$ is defined as $T \oplus T \oplus \cdots \oplus T$, i.e., the direct sum of $n$ copies of $T$.

\item The \emph{tensor product} $T \otimes T'$ is a tensor over $(X \times X'),~ (Y \times Y'),~ (Z \times Z')$ defined as
  \[
    T \otimes T' \defeq \sum_{x_i \in X} \sum_{y_j \in Y} \sum_{z_k \in Z} \sum_{x'_{i'} \in X'}  \sum_{y'_{j'} \in Y'} \sum_{z'_{k'} \in Z'} a_{i,j,k} \cdot a'_{i',j',k'} \cdot (x_i, x'_{i'}) \cdot (y_j, y'_{j'}) \cdot (z_k, z'_{k'}).
  \]
  Specifically, the \emph{tensor power} $T^{\otimes n}$ is defined as the tensor product $T \otimes T \otimes \cdots \otimes T$ of $n$ copies. 
\item The \emph{summation} $T + T'$ is well-defined only when $(X,Y,Z) = (X', Y', Z')$. We define $T + T'$ to be a tensor over $X, Y, Z$:
  \[T + T' \defeq \sum_{x_i \in X} \sum_{y_j \in Y} \sum_{z_k \in Z} (a_{i,j,k} + a'_{i,j,k}) \cdot x_i y_j z_k.\]
\end{itemize}

\paragraph*{Rotation, swapping, and symmetrization.} Let
\[
  T = \sum_{i \in [n]} \sum_{j \in [m]} \sum_{k \in [p]} a_{i,j,k}\cdot x_iy_jz_k
\]
be a tensor over $X = \BK{x_1, \ldots, x_n}$, $Y = \BK{y_1, \ldots, y_m}$, and $Z = \BK{z_1, \ldots, z_p}$. We define the \emph{rotation} of $T$, denoted by $T^\rot$, as
\[
  T^\rot \defeq \sum_{i \in [n]} \sum_{j \in [m]} \sum_{k \in [p]} a_{i,j,k} \cdot x_j y_k z_i.
\]
over $X' = \BK{x_1, \ldots, x_m}$, $Y' = \BK{y_1, \ldots, y_p}$, and $Z' = \BK{z_1, \ldots, z_n}$. Intuitively, rotation is changing the order of dimensions from $(|X|, |Y|, |Z|)$ to $(|Y|, |Z|, |X|)$ while keeping the structure of the tensor unchanged.

Similarly, we define the \emph{swapping} of $T$, denoted by $T^\swap$, as
\[
  T^\swap \defeq \sum_{i \in [n]} \sum_{j \in [m]} \sum_{k \in [p]} a_{i,j,k} \cdot x_j y_i z_k.
\]
over $X' = \BK{x_1, \ldots, x_m}$, $Y' = \BK{y_1, \ldots, y_n}$, and $Z' = \BK{z_1, \ldots, z_p}$. Swapping is changing the order of dimensions from $(|X|, |Y|, |Z|)$ to $(|Y|, |X|, |Z|)$, i.e., swapping X and Y dimensions.

Based on these two operations, we define the \emph{symmetrization} of $T$. The \emph{rotational symmetrization}, or \emph{3-symmetrization} of $T$, is defined by $\sym_3(T) = T\otimes T^\rot \otimes T^{\rot \; \rot}$. In $\sym_3(T)$, the X, Y and Z variables are symmetric. Further, we define the \emph{full symmetrization}, or \emph{6-symmetrization} of $T$, as $\sym_6(T) = \sym_3(T) \otimes \sym_3(T)^\swap$.

\paragraph*{Tensor rank.} The \emph{rank} $R(T)$ of a tensor $T$ is the minimum integer $r \geq 0$ such that we can decompose $T$ into 
\[
  T = \sum_{t=1}^r \bigg(\sum_{x_i \in X} a_{t,i} x_i \bigg) \cdot \bigg(\sum_{y_j \in Y} b_{t,j} y_j\bigg) \cdot \bigg(\sum_{z_k \in Z} c_{t,k} z_k\bigg).
\]
This equation is also called the \emph{rank decomposition} of $T$.

The \emph{asymptotic rank} $\tilde{R}(T)$ is defined as
\[
  \tilde{R}(T) \defeq \lim_{n \rightarrow \infty} R\left(T^{\otimes n}\right)^{1/n}.
\]

We need the following theorem linking asymptotic rank to the matrix multiplication exponent $\omega$.

\begin{theorem} [Sch\"onhage's $\tau$ theorem \cite{schonhage1981partial}] \label{thm:Schonhage}
  Let tensor $T = \bigoplus_{i=1}^\ell \langle n_i, m_i, p_i \rangle$ be a direct sum of matrix multiplication tensors where $n_i, m_i, p_i \geq 1$ are positive integers. Suppose for $\tau \in [2/3, 1]$,
  $$\sum_{i=1}^\ell (n_i m_i p_i)^{\tau} = \tilde{R}(T).$$
We have that $\omega \leq 3 \tau$.
\end{theorem}

\subsection{Restrictions, Degenerations, and Values}
\label{sec:degen_and_value}

\paragraph*{Restrictions of a tensor.} Let $T$ be a tensor over $X, Y, Z$ and $T'$ be a tensor over $X', Y', Z'$. We say $T'$ is a \emph{restriction} of $T$ (or $T$ restricts to $T'$) if there is a mapping $f_1: X \to \Span_\F(X')$, where $\Span_\F(X')$ is the set of linear combinations over $X'$ with coefficients in $\F$; and similarly $f_2: Y \to \Span_\F(Y')$, $f_3: Z \to \Span_\F(Z')$, satisfying
\[T' = \sum_{i=1}^{|X|} \sum_{j=1}^{|Y|} \sum_{k=1}^{|Z|} a_{i,j,k} \cdot f_1(x_i) \cdot f_2(y_j) \cdot f_3(z_k).\]
It is easy to verify that $R(T') \le R(T)$ and $\tilde R(T') \le \tilde R(T)$ by applying $f_1, f_2, f_3$ on each variable that appeared in the rank decomposition.

\paragraph*{Degenerations.} Suppose $T$ is a tensor over $X, Y, Z$ while $T'$ is a tensor over $X', Y', Z'$. Also, there are mappings
\[
  f_1 : X \to \Span_{\F[\lambda]}(X'), \qquad f_2 : Y \to \Span_{\F[\lambda]}(Y'), \qquad f_3 : Z \to \Span_{\F[\lambda]}(Z'),
\]
where $\F[\lambda]$ is the ring of polynomials of a formal variable $\lambda$. If there exists $d \in \mathbb N$, such that
\[
  T' = \lambda^{-d} \cdot \bigg(\sum_{i=1}^{|X|} \sum_{j=1}^{|Y|} \sum_{k=1}^{|Z|} a_{i,j,k} \cdot f_1(x_i) \cdot f_2(y_j) \cdot f_3(z_k)\bigg) + O(\lambda),
\]
then we say $T'$ is a degeneration of $T$, written $T \degen T'$. It is clear that restriction is a special type of degeneration. One can also verify that $\tilde R(T') \le \tilde R(T)$.

\paragraph*{Zeroing out.} \emph{Zeroing out} is a special case of restrictions. While zeroing out, we select a subset $X' \subseteq X$ and set all X-variables outside $X'$ to zero. Similarly, $Y' \subseteq Y$ and $Z' \subseteq Z$ are chosen and all other Y and Z-variables are set to zero. Namely, zeroing out is the degeneration with
\[
  f_1(x_i) =
  \begin{cases}
    x_i, & x_i \in X', \\
    0, & x_i \in X \backslash X',
  \end{cases}
\]
and $f_2, f_3$ are defined similarly. The resulting tensor $T'$ is also called the \emph{subtensor of} $T$ \emph{over} $X', Y', Z'$, written $T' = T \vert_{X', Y', Z'}$.

Zeroing out suffices for previous works, but we need one more type of degeneration below for technical reasons.

\paragraph*{Identifications.} Let $X^{(1)}$ and $X^{(2)}$ be two disjoint sets that identify with the same set $X$, in the sense that there exists bijections $i_{X^{(1)}} : X^{(1)} \rightarrow X$ and $i_{X^{(2)}}: X^{(2)} \rightarrow X$. Similarly, for $Y,\; Y^{(1)}, Y^{(2)}$ and $Z,\; Z^{(1)}, Z^{(2)}$, there are bijections $i_{Y^{(1)}}, i_{Y^{(2)}}$ and $i_{Z^{(1)}}, i_{Z^{(2)}}$ similarly.

Suppose $T^{(1)}$ is a tensor defined on the sets $X^{(1)}, Y^{(1)}, Z^{(1)}$, and $T^{(2)}$ is defined on $X^{(2)}, Y^{(2)}, Z^{(2)}$. For their direct sum $T^{(1)} \oplus T^{(2)}$, we can define the following degeneration.
$$f_1(x_i) = \begin{cases} i_{X^{(1)}}(x_i) & x_i \in X^{(1)}, \\ i_{X^{(2)}}(x_i) & x_i \in X^{(2)}.
\end{cases}$$
The definitions of $f_2, f_3$ are similar. The resulting tensor after degeneration is exactly $T^{(1)} + T^{(2)}$ as if they were both defined on $X, Y, Z$. We call such a degeneration an \emph{identification} because it identifies the different copies of the same variable. 

Moreover, for $m$ tensors $T^{(1)}, T^{(2)}, T^{(3)}, \dots, T^{(m)}$, we can similarly define their identification, written as 
\[
  \bigoplus_{i=1}^m T^{(i)} \degen \sum_{i=1}^{m}  T^{(i)}.
\]

\paragraph*{Values.} The value of a tensor captures the asymptotic ability of its symmetrization, $\sym_3(T)$ or $\sym_6(T)$, in computing matrix multiplication.

\begin{definition}
  \label{def:value}
  The 3-symmetrized $\tau$-value of a tensor $T$, denoted as $V^{(3)}_\tau(T)$, is defined as
  \[ V^{(3)}_\tau(T) = \limsup_{n \to \infty} \max \BKmid{\bk{\sum_{i=1}^m (a_i b_i c_i)^{\tau}}^{\frac{1}{3n}}}{\sym_3(T)^{\otimes n} \degen \bigoplus_{i=1}^m \angbk{a_i, b_i, c_i}}. \]
  The 6-symmetrized $\tau$-value is defined by replacing $\sym_3$ with $\sym_6$ (and replacing $3n$ with $6n$) in the above definition, denoted by $V^{(6)}_\tau(T)$.
\end{definition}

Note that the previous works only use the 3-symmetrized values; however, we need 6-symmetrized values due to technical reasons. $V^{(6)}_\tau(T) \ge V^{(3)}_\tau(T)$ holds for any tensor $T$.

It is easy to verify that, for tensors $T$ and $T'$, their values satisfy $V^{(6)}_\tau(T \otimes T') \ge V^{(6)}_\tau(T) \cdot V^{(6)}_\tau(T')$ and $V^{(6)}_\tau(T \oplus T') \ge V^{(6)}_\tau(T) + V^{(6)}_\tau(T')$. Similar properties hold for 3-symmetrized values. These properties are called the \emph{super-multiplicative} and \emph{super-additive} properties, which allow us to bound the values of some complex tensors based on the values of their ingredients.

\subsection{Partitions of a Tensor}
\label{sec:partition-tensor}

\paragraph*{Partitions of a tensor.} The \emph{partition} of variable sets $X, Y, Z$ is defined as the disjoint unions
\[
  X = X_0 \sqcup X_1 \sqcup \cdots \sqcup X_{\ell_1 - 1}, \qquad Y = Y_0 \sqcup Y_1 \sqcup \cdots \sqcup Y_{\ell_2 - 1}, \qquad Z = Z_0 \sqcup Z_1 \sqcup \cdots \sqcup Z_{\ell_3 - 1}.
\]
Then, we define the corresponding partition of $T$ to be $T = \sum_{i = 0}^{\ell_1 - 1} \sum_{j = 0}^{\ell_2 - 1} \sum_{k = 0}^{\ell_3 - 1} T_{i,j,k}$ where $T_{i,j,k} \defeq T \vert_{X_i, Y_j, Z_k}$ is the subtensor of $T$ over subsets $X_i, Y_j, Z_k$. In this case, we call $T$ a \emph{partitioned tensor} and call $T_{i,j,k}$ a \emph{component} of $T$. For a variable $x \in X_i$, we say $i$ is the \emph{index} of $x$. It is similar for Y and Z-variables.

An important class of partitioned tensors for matrix multiplication is the \emph{$P$-partitioned tensors}, which has been heavily used in prior works because it enables the use of the laser method. A partitioned tensor $T = \sum_{i = 0}^{\ell_1 - 1} \sum_{j = 0}^{\ell_2 - 1} \sum_{k = 0}^{\ell_3 - 1} T_{i,j,k}$ is called $P$-partitioned if $T_{i,j,k} = 0$ for all $i + j + k \ne P$. Most partitioned tensors we consider in this paper are $P$-partitioned tensors.

\paragraph*{Partitions of a tensor power.} Consider the tensor power $\mathcal{T} \defeq T^{\otimes n}$, which is a tensor over $X^{n}$, $Y^{n}$, and $Z^{n}$. Here $T$ is called the \emph{base tensor}, and the partition of $T$ we take is called the \emph{base partition}. Given any base tensor together with a base partition, it naturally induces the partition for the tensor power $\mathcal{T}$, as described below.

Given the base partition of $T$, we let $I = (I_1, I_2, \dots, I_n)$ be a sequence in $[\ell_1]^n$. We call $X_I \defeq X_{I_1} \times X_{I_2} \times \cdots \times X_{I_n}$ an \emph{X-variable block}, or equivalently an \emph{X-block}. Similarly, we define Y-blocks and Z-blocks. These blocks form a partition of variables in $\T$. Moreover, the sequence $I$ here is called the \emph{index sequence} of block $X_I$. We define the index sequences for Y-blocks and Z-blocks similarly.

Using such notations, $\T$ is partitioned into
\[
  \T = \sum_{I \in [\ell_1]^n}\sum_{J \in [\ell_2]^n}\sum_{K \in [\ell_3]^n} \mathcal{T}_{I, J, K},
\]
where $\T_{I, J, K} = \T \vert_{X_I, Y_J, Z_K}$ is the subtensor of $\T$ over blocks $X_I$, $Y_J$, and $Z_K$. The index sequence of a variable is defined as the index sequence of the block it belongs to, i.e., any variable $x \in X_I$ has the index sequence $I$ (similar for Y and Z-variables).

\subsection{Coppersmith-Winograd Tensor}
\label{sec:CWtensor}

The most important tensor in the fast matrix multiplication literature is the Coppersmith-Winograd tensor \cite{coppersmith1987matrix}. It is a partitioned tensor over
\begin{align*}
  X &= \{x_0\} \sqcup \{x_1, x_2, \dots, x_{q}\} \sqcup \{x_{q+1}\}, \\
  Y &= \{y_0\} \sqcup \{y_1, y_2, \dots, y_{q}\} \sqcup \{y_{q+1}\}, \\
  Z &= \{z_0\} \sqcup \{z_1, z_2, \dots, z_{q}\} \sqcup \{z_{q+1}\}.
\end{align*}
We define $X_0 = \{x_0\}, X_1 = \{x_1, x_2, \dots, x_{q}\}$, and $X_2 =  \{x_{q+1}\}$. The partition is therefore $X = X_0 \sqcup X_1 \sqcup X_2$.\footnote{Here we let the index start from $0$ to be consistent with previous works.} Similarly, we define the partitioned sets for $Y$ and $Z$.

The tensor $\CW_q$ is defined as
\[\CW_q = \sum_{i=1}^q (x_i y_i z_0 + x_i y_0 z_i + x_0 y_i z_i) + x_0 y_{0} z_{q+1} + x_{0} y_{q+1} z_0 + x_{q+1} y_0 z_{0} .\]
By the definition of partition, $\CW_q$ can be written as
\[ \CW_q = T_{1, 1, 0} + T_{1, 0, 1} + T_{0, 1, 1} + T_{0, 0, 2} + T_{0, 2, 0} + T_{2, 0, 0}, \]
where $T_{i,j,k} = \CW_q \vert_{X_i, Y_j, Z_k}$, i.e., the corresponding subtensor of $\CW_q$.

Note here the nonzero $T_{i,j,k}$'s all satisfy that $i + j + k = 2$, i.e., $\CW_q$ is a 2-partitioned tensor. This is a crucial property of $\CW_q$, which enables the use of the laser method. Another important property is that its asymptotic rank $\tilde R(\CW_q) = q + 2$~\cite{coppersmith1987matrix}. This benefits the use of Theorem~\ref{thm:Schonhage}.

Similar to the previous subsection, the above partition of $\CW_q$ induces a partition of $\T = \CW_q^{\otimes N}$:
\[
  \T = \sum_{I, J, K \in \BK{0, 1, 2}^{N}} \T_{I,J,K}, \qquad \textup{where} \qquad \T_{I,J,K} \defeq \T \vert_{X_I, Y_J, Z_K}.
\]
Because $\CW_q$ is 2-partitioned, $\T_{I,J,K}$ is nonzero only when $I + J + K = (2, 2, \ldots, 2)$. We call it the \emph{level-1 partition} of the tensor power $\T$. (We will introduce higher-level partitions of $\T$ in the next subsection.) The \emph{variable blocks} $X_I, Y_J, Z_K$ and \emph{index sequences} $I, J, K$ are also defined according to the previous subsection.

\subsection{Leveled Partitions of CW Tensor Power}
\label{subsec:leveled}

The prior works obtained better bounds on $\omega$ by applying the laser method on the tensor powers of the CW tensor, $\CW_q^{\otimes 2^{\lvl}}$, instead of directly on $\CW_q$. Such analyses are recursive: The analysis of the $2^{\lvl}$-th power makes use of the bounds of $\CW_q^{2^{\lvl - 1}}$'s components, namely $T_{i,j,k}$, where $i + j + k = 2^{\lvl}$. When analyzing the tensor $T = \CW_q^{\otimes 2^{\lvl}}$, we fix a large number $N = n \cdot 2^{\lvl}$, and view the tensor power $\T = (\CW_q^{\otimes 2^{\lvl}})^{\otimes n} \cong \CW_q^{\otimes N}$ under a proper variable partition (described below). In this paper, we also place our analysis under the same view, while further fixing $N$ to be the same for all levels $\lvl$ (thus $n = N / 2^{\lvl}$ varies between levels). After that, the partitions in different levels become partitions of the same tensor $\T$ of different granularities. Below, we formally define such multi-level partitions of $\T$.

\paragraph*{Level-$1$ partition.} We consider the $N$-th power of the CW tensor, $\mathcal T = \CW_q^{\otimes N}$, for a large integer $N$. As we discussed before, the level-1 partition of $\T$ is written as $\mathcal T = \sum_{I, J, K \in \{0,1,2\}^N} \mathcal T_{I, J, K}$. The summation is over all $(I, J, K)$ where $I + J + K = (2, 2, \ldots, 2)$. %

\paragraph*{Level-$\lvl$ partition.} Since we will let $N$ go to infinity, without loss of generality, we can assume that $N$ is a multiple of $2^{\lvl - 1}$, i.e., $N = n \cdot 2^{\lvl - 1}$. Then we can alternatively view $\CW_q^{\otimes N}$ as $\big(\CW_q^{\otimes 2^{\lvl - 1}}\big)^{\otimes n}$. If we take $T = \CW_q^{\otimes 2^{\lvl - 1}}$ as the base tensor and adopt the base partition defined below, we will get another partition of $\CW_q^{\otimes N}$, which we call the \emph{level-$\lvl$ partition}.

The base partition of $T = \CW_q^{\otimes 2^{\lvl - 1}}$ is defined as follows. Note that $T$ is a tensor over $X^{2^{\lvl - 1}}, Y^{2^{\lvl - 1}}, Z^{2^{\lvl - 1}}$. If we view $T$ as a power of $\CW_q$ and observe its level-1 partition, each variable $x \in X^{2^{\lvl - 1}}$ has an index sequence $(i_1, \ldots, i_{2^{\lvl - 1}}) \in \BK{0, 1, 2}^{2^{\lvl - 1}}$. The base partition of $T$ is then a coarsening of this level-1 partition formed by grouping variables according to $\sum_{t=1}^{2^{\lvl - 1}} i_t$. Formally, we partition the variable set $X^{2^{\lvl - 1}}$ into $X^{2^{\lvl - 1}} = \tilde{X}_0 \sqcup \tilde{X}_1 \sqcup \cdots \sqcup \tilde{X}_{2^{\lvl}}$, where
\[
  \tilde{X}_i \defeq \bigsqcup_{\substack{(i_1, \ldots, i_{2^{\lvl-1}}) \in \BK{0, 1, 2}^{2^{\lvl-1}} \\ i_1 + \cdots + i_{2^{\lvl-1}} = i}} X_{i_1} \times X_{i_2} \times \cdots \times X_{i_{2^{\lvl-1}}}.
\]
We partition $Y^{2^{\lvl-1}} = \tilde{Y}_0 \sqcup \tilde{Y}_1 \sqcup \tilde{Y}_2 \sqcup \cdots \sqcup \tilde{Y}_{2^\lvl}$ and $Z^{2^{\lvl-1}} = \tilde{Z}_0 \sqcup \tilde{Z}_1 \sqcup \tilde{Z}_2 \sqcup \cdots \sqcup \tilde{Z}_{2^\lvl}$ similarly. 
Then the base partition of $T$ is written as
\[
  T = \sum_{i, j, k \in \{0, 1, \dots, 2^\lvl\}} T_{i,j,k}
\]
where $T_{i,j,k} = T \vert_{\tilde{X}_i, \tilde{Y}_j, \tilde{Z}_k}$. Note that by the property of CW tensor, $T_{i,j,k}$ is non-zero only when $i + j + k = 2^\lvl$. Here $T_{i,j,k}$ is called a \emph{level-$\lvl$ component}. Sometimes, we also write $(i, j, k)$ to represent the component $T_{i,j,k}$ if there is no ambiguity.

The level-$\lvl$ partition of $\mathcal T \defeq T^{\otimes n}$ is induced naturally by such base partition. For sequence $I = \bk{I_1, \ldots, I_{n}} \in \BK{0, 1, \ldots, 2^\lvl}^{n}$, the variable block $X_I$ is $\tilde{X}_{I_1} \times \tilde{X}_{I_2} \times \cdots \times \tilde{X}_{I_n}$. Similarly, we define $Y_J$ and $Z_K$ for $J, K \in \BK{0, 1, \ldots, 2^\lvl}^{n}$. $X_I, Y_J, Z_K$ are called \emph{level-$\lvl$ variable blocks} (or \emph{level-$\lvl$ blocks}). Formally, the level-$\lvl$ partition of $\T$ is written as
\[ \mathcal T = \sum_{I, J, K \in \BK{0, 1, \ldots, 2^\lvl}^{n}} \mathcal T_{I, J, K} \]
where $\mathcal T_{I,J,K} = \T \vert_{X_I, X_J, Z_K}$. Note that $\mathcal T_{I, J, K} \ne 0$ only if $\forall t \in [n]$, $I_t + J_t + K_t = 2^{\lvl}$. We call such $(I, J, K)$, or corresponding variable blocks $(X_I, Y_J, Z_K)$, a \emph{triple}. (These two notations are equivalent for a triple.)

\paragraph*{Level-$\lvl$ index sequences.}
Recall that the partitions of different levels are different partitions of the same tensor $\mathcal T = \CW_q^{\otimes N}$ over $X^N, Y^N, Z^N$ (recall that $X, Y, Z$ are variable sets of $\CW_q$). Therefore, the same variable $x \in X^N$ has different index sequences in different levels. We use \emph{level-$\lvl$ index sequence} to denote its index sequence in the level-$\lvl$ partition. The level-$\lvl$ index sequence of any variable is a sequence in $\BK{0, 1, \ldots, 2^\lvl}^{\bk{N/2^{\lvl-1}}}$.

The level-$\lvl$ partition is a coarsening of the level-$\lastlvl$ partition, because each level-$\lastlvl$ index sequence $I = (I_1, \ldots, I_{(N/2^{\lvl-2})})$ corresponds to a unique level-$\lvl$ index sequence
\[
  I' = \bk{I_1 + I_2,\ I_3 + I_4,\ \ldots\ ,\ I_{\midbk{N/2^{\lvl-2} - 1}} + I_{\midbk{N/2^{\lvl-2}}}}
\]
by adding every two consecutive entries together. Conversely, each level-$\lvl$ index sequence corresponds to a collection of level-$\lastlvl$ index sequences.

Similar many-to-one correspondence also appears between level-$\lastlvl$ blocks and level-$\lvl$ blocks: Each level-$\lastlvl$ block belongs to a unique level-$\lvl$ block, while each level-$\lvl$ block is the disjoint union of several level-$\lastlvl$ blocks. A level-$\lvl$ block can be regarded as not only a set of variables but also a collection of level-$\lastlvl$ blocks. If $X_I$ is some level-$\lvl$ block and $X_{\hat I}$ is level-$\lastlvl$ block, we use $X_{\hat I} \in X_I$ to denote that $X_{\hat I}$ is contained in $X_I$. In this case, we say $X_I$ is the \emph{parent} of $X_{\hat I}$.

(In many places of this paper, we observe two adjacent levels $\lvl$ and $\lvl - 1$ simultaneously. We often use notations $I, J, K$ to represent level-$\lvl$ index sequences while using $\hat I, \hat J, \hat K$ to represent level-$\lastlvl$ index sequences.)

\subsection{Distributions and Entropy}

Throughout this paper, we only need to consider discrete distributions supporting on a finite set. For such a distribution $\alpha$ supporting on $S$, we require it to be normalized (i.e. $\sum_{s \in S} \alpha(s) = 1$) and non-negative (i.e., $\forall s \in S, \ \alpha(s) \geq 0$). We define its entropy in the standard way:
\[ H(\alpha) \defeq -\sum_{\substack{s \in S \\ \alpha(s) > 0}} \alpha(s) \log \alpha(s). \]
We need the following lemma in our analysis.
\begin{lemma}\label{thm:stirling1}
  Let $\alpha$ be a discrete distribution that $\alpha(1)+\cdots+\alpha(k)=1$, then
  $$\binom{N}{\alpha(1)N,\cdots,\alpha(k)N}=2^{N(H(\alpha)+o(1))}.$$
\end{lemma}
\begin{proof}
  By Stirling's approximation, $\log(N!)=N\log N-N\log e+O(\ln N)$, so
  \begin{align*}
    &\log \binom{N}{\alpha(1)N,\cdots,\alpha(k)N} \\ 
    ={} & N\log N-N\log e+o(N)-\sum_{i=1}^k\big(\alpha(i)N\log(\alpha(i)N)-\alpha(i)N\log e+o(\alpha(i)N)\big) \\
    ={} &-N\sum_{i=1}^k\alpha(i)\log\alpha(i)+o(N) = NH(\alpha) + o(N). \qedhere
  \end{align*}
\end{proof}

\subsection{Distributions of Index Sequences}
\label{sec:def_distributions}

As in all prior works, when the laser method is applied on $T^{\otimes n}$ for some $P$-partitioned tensor $T$, the first step is to choose a distribution over all components. It is the same for this paper where we choose $T = \CW_q^{\otimes 2^{\lvl - 1}}$. Below, we clarify the terminologies and notations about these distributions, which are basically consistent with prior works.

\begin{definition}[Component Distributions]
  A level-$\lvl$ (joint) \emph{component distribution} $\alpha$ is a distribution over all level-$\lvl$ components $T_{i,j,k}$ where $i+j+k=2^\lvl$. In the level-$\lvl$ partition of the tensor $\T = \CW_q^{\otimes N}$ where $N = n \cdot 2^{\lvl-1}$, let $(X_I, Y_J, Z_K)$ be a level-$\lvl$ triple with index sequences $I,J,K$. We say $(X_I, Y_J, Z_K)$ is \emph{consistent} with a joint component distribution $\alpha$, if for each level-$\lvl$ component $T_{i,j,k}$, the proportion of index positions $t \in [n]$ where $(I_t, J_t, K_t) = (i,j,k)$ equals $\alpha(i,j,k)$.\footnote{Without loss of generality, as $n \to \infty$, we assume that each entry of $\alpha$ is a multiple of $1/n$. Therefore, there are always some triples consistent with $\alpha$.}

  A level-$\lvl$ X-marginal component distribution $\alphx$ is a distribution over $\BK{0, 1, \ldots, 2^\lvl}$. We say a level-$\lvl$ X-block $X_I$ is consistent with $\alphx$, if for every $i = 0, 1, \ldots, 2^\lvl$, the proportion of positions $t \in [n]$ where $I_t = i$ equals $\alphx(i)$. The Y and Z-marginal distributions are defined similarly, written as $\alphy$ and $\alphz$.

  A joint component distribution $\alpha$ induces marginal component distributions $\alphx$, $\alphy$, and $\alphz$, by
  \[
    \alphx(i) \defeq \sum_{j + k = 2^\lvl - i} \alpha(i,j,k), \qquad
    \alphy(j) \defeq \sum_{i + k = 2^\lvl - j} \alpha(i,j,k), \qquad
    \alphz(k) \defeq \sum_{i + j = 2^\lvl - k} \alpha(i,j,k).
    \numberthis \label{eq:shape_joint_to_margin}
  \]
  These induced marginal distributions are called the \emph{marginals} of $\alpha$.
\end{definition}

Denote the set of all distributions $\alpha(i,j,k)$ as $D$. As in \cite{coppersmith1987matrix,williams2012}, we have the following fact:

\begin{lemma}
  \label{thm:unique-dis-level1}
  In the level-1 partition, given marginal distributions $\alphx(i), \alphy(j), \alphz(k)$, the joint distribution $\alpha(i,j,k)$ is uniquely determined if exists.
\end{lemma}
\begin{proof}
  Suppose the marginals $\alphx$, $\alphy$, and $\alphz$ are given. We can determine $\alpha(0,0,2) = \alphz(2)$ and $\alpha(0,1,1) = \alphx(0) - \alphy(2) - \alphz(2)$. Other entries can be determined similarly.
\end{proof}

In higher levels, marginal distributions usually do not uniquely determine the joint distribution.\footnote{This is the cause of a loss in moduli in the analysis of higher powers, which can be reduced by the refined laser method~\cite{alman2021}.} Given marginal distributions $\alphx,\alphy,\alphz$, define $D(\alphx,\alphy,\alphz)\subseteq D$ to be the set of joint distributions inducing marginal distributions $\alphx,\alphy,\alphz$:
\[ D(\alphx,\alphy,\alphz) \defeq \{\alpha' \in D \mid \text{the marginal distribution of $\alpha'$ is $\alphx, \alphy, \alphz$} \}. \]
For convenience, we further define $\distShareMargin \defeq D(\alphx, \alphy, \alphz)$ to be the collection of distributions that share the same marginals with $\alpha$; define $D^*(\alphx, \alphy, \alphz) \defeq \argmax_{\alpha' \in D(\alphx, \alphy, \alphz)} H(\alpha')$.

\paragraph{Split distributions.} Suppose we view a level-$\lvl$ component $T_{i,j,k}$'s tensor power $\T \defeq T_{i,j,k}^{\otimes n}$ as a subtensor of $(\CW_q^{\otimes 2^{\lvl - 1}})^{\otimes n}$, in the sense that $\T = (\CW_q^{\otimes 2^{\lvl - 1}})^{\otimes n} \vert_{X_I, Y_J, Z_K}$ where $(I_t, J_t, K_t) = (i, j, k)$ for all $t \in [n]$. When we look at the level-$\lastlvl$ partition of this tensor, each factor $T_{i,j,k}$ ($i + j + k = 2^\lvl$) further splits into
\[
  T_{i,j,k} = \sum_{i_l + j_l + k_l = 2^{\lvl-1}} T_{i_l, j_l, k_l} \otimes T_{i-i_l,\, j-j_l,\, k-k_l}.
\]
It corresponds to the fact that, as explained at the end of \Cref{subsec:leveled}, each index in the level-$\lvl$ index sequence is the sum of two consecutive indices in the level-$\lastlvl$ index sequence. A \emph{split distribution} is a distribution over the terms on the right hand side. (Such concept is similar to the component distribution and was used in prior works.) Formally:

\begin{definition}[Split Distributions]
  A (joint) \emph{split distribution} of level-$\lvl$ component $T_{i,j,k}$, namely $\alpha^{(i, j, k)}$, is a distribution over all $(i_l, j_l, k_l)$ such that both $T_{i_l, j_l, k_l}$ and $T_{i - i_l, \, j - j_l, \, k - k_l}$ are level-$\lastlvl$ components. (Equivalently, it satisfies $0 \le i_l \le i$, $0 \le j_l \le j$, $0 \le k_l \le k$, and $i_l + j_l + k_l = 2^{\lvl - 1}$.)

  Consider the level-$\lastlvl$ partition of $\T \defeq T_{i,j,k}^{\otimes n}$ and let $(X_{\hat I}, Y_{\hat J}, Z_{\hat K})$ be a level-$\lastlvl$ triple of $\T$. We say $(X_{\hat I}, Y_{\hat J}, Z_{\hat K})$ is \emph{consistent} with the joint split distribution $\alpha^{(i, j, k)}$, if for every $(i_l, j_l, k_l)$,
  \[
    \abs{\BK{t \in [n] \;\middle|\; (\hat I_{2t-1}, \hat J_{2t-1}, \hat K_{2t-1}) = (i_l, j_l, k_l)}} = \alpha^{(i, j, k)}(i_l, j_l, k_l) \cdot n.
  \]
  (I.e., $\alpha^{(i, j, k)}(i_l, j_l, k_l)$ fraction of the factors $T_{i,j,k}$ split into $T_{i_l, j_l, k_l} \otimes T_{i - i_l, \, j - j_l, \, k - k_l}$.)

  The marginals of $\alpha^{(i,j,k)}$ on variables $i_l, j_l, k_l$ are called the \emph{marginal split distributions}, written $\alphx^{(i,j,k)}$, $\alphy^{(i,j,k)}$, and $\alphz^{(i,j,k)}$, respectively. We say a level-$\lastlvl$ X-block $X_{\hat I}$ is consistent with $\alphx^{(i,j,k)}$, if for every $i_l$, we have $|\{t \in [n] \mid \hat I_{2t-1} = i_l\}| = \alphx^{(i,j,k)}(i_l) \cdot n$. Similar for Y and Z-blocks.
\end{definition}

When we focus on an arbitrary level-$\lvl$ triple $(X_I, Y_J, Z_K)$ of $\T' \defeq (\CW_q^{\otimes 2^{\lvl - 1}})^{\otimes n}$ (not necessarily $T_{i,j,k}^{\otimes n}$), $\T' \vert_{X_I, Y_J, Z_K}$ still contains some factors $T_{i,j,k}$. We use the notation $S_{i,j,k}(I,J,K) \defeq \{t \in [n] \mid (I_t, J_t, K_t) = (i, j, k)\}$ to represent the set of positions $t$ where the component $T_{i,j,k}$ appears. (When $I, J, K$ are clear from context, we also write $S_{i,j,k}$ for short.) Based on this notation, we define the joint split distribution $\alpha^{(i,j,k)}$ of a level-$\lastlvl$ triple $(X_{\hat I}, Y_{\hat J}, Z_{\hat K})$ within the position set $S_{i,j,k}$ as
\[
  \alpha^{(i,j,k)}(i_l, j_l, k_l) = \frac{1}{|S_{i,j,k}|} \abs{\BK{t \in S_{i,j,k} \;\middle|\; (\hat I_{2t-1}, \hat J_{2t-1}, \hat K_{2t-1}) = (i_l, j_l, k_l) }}.
\]
(The only difference from the earlier variant is that we replaced $[n]$ with $S_{i,j,k}$.)

\subsection{Salem-Spencer Set}

As in previous works, we also need the Salem-Spencer set to construct independent matrix products. 

\begin{theorem}[\cite{salem1942,Behrend1946}]
  \label{thm:salem-spencer}
  For any positive integer $M$, there is a set $A\subset \{0,\cdots,M-1\}$ with no three-term arithmetic progression modulo $M$, which satisfies $|A|>M^{1-o(1)}$. (Namely if $a, b, c \in A$ satisfy $a + c \equiv 2b \pmod M$, then $a=b=c$.) $A$ is called a \emph{Salem-Spencer set}. 
\end{theorem}

\subsection{Restricted-Splitting Tensor Power}
\label{sec:restricted_splitting}

In addition to \emph{values}, we also define the \emph{restricted-splitting values} to capture the ability of the subtensor of $T_{i,j,k}^{\otimes n}$ with a specific split distribution on Z-variable blocks. We first define the \emph{restricted-splitting tensor power}. (It is a new concept introduced in this paper.)

\begin{definition}
  \label{def:restricted_splitting}
  Let $T_{i,j,k}$ be a level-$\lvl$ component and $\alphz^{(i,j,k)}(k_l)$ be a Z-marginal split distribution of $(i,j,k)$. Consider the level-$\lastlvl$ partition of $T_{i,j,k}^{\otimes n}$. We zero out every level-$\lastlvl$ Z-variable block $Z_{\hat K}$ which is not consistent with $\alphz^{(i,j,k)}$. The remaining subtensor is denoted as $T_{i,j,k}^{\otimes n}[\alphz^{(i,j,k)}]$. It is called \emph{restricted-splitting tensor power} of $T_{i,j,k}$.

 We similarly define $T_{i,j,k}^{\otimes n}[\alphx^{(i,j,k)}]$ and $T_{i,j,k}^{\otimes n}[\alphy^{(i,j,k)}]$ to capture the cases when the split distribution of X or Y dimension is restricted.
\end{definition}

In this paper, we often use the notation $\tilde\alpha_{i,j,k}$ instead of $\alphz^{(i,j,k)}$ to denote the Z-split distribution of component $(i, j, k)$. Under this notation, the restricted dimension is Z by default if not otherwise stated.

Furthermore, we define the value with restricted-splitting distribution:
\[
  V^{(6)}_\tau(T_{i,j,k}, \tilde\alpha_{i,j,k}) = \limsup_{n\to\infty} V^{(6)}_\tau(T_{i,j,k}^{\otimes n}[\tilde\alpha_{i,j,k}])^{1/n}.
\]
It is easy to verify that the above definition is equivalent to
\[
  V^{(6)}_\tau(T_{i,j,k}, \tilde\alpha_{i,j,k}) = \limsup_{n\to\infty} \max
  \BKmid{\bk{\sum_{i=1}^m (a_i b_i c_i)^{\tau}}^{1/(6n)}}{\sym_6\left(T_{i,j,k}^{\otimes n}[\tilde\alpha_{i,j,k}]\right) \degen \bigoplus_{i=1}^m \angbk{a_i, b_i, c_i}}.
\]
We call this concept \emph{restricted-splitting values}. We also denote by $V^{(3)}_\tau(T_{i,j,k}, \tilde\alpha_{i,j,k})$ the 3-symmetrized restricted-splitting value:
\[
  V^{(3)}_\tau(T_{i,j,k}, \tilde\alpha_{i,j,k}) = \limsup_{n\to\infty} V^{(3)}_\tau(T_{i,j,k}^{\otimes n}[\tilde\alpha_{i,j,k}])^{1/n}.
\]

\begin{remark}
  In the prior works, a lower bound of the value of each level-$\lvl$ component, $\valsix(T_{i,j,k})$, was determined using the values of level-$\lastlvl$ components. The concept of \emph{values} thus serves as a bridge between different levels, acting as the interface for the recursive analysis. In this paper, we introduce \emph{restricted-splitting values} to play the same role. This is why we use a different notation $\tilde\alpha_{i,j,k}$ instead of $\alphz^{(i,j,k)}$: When we analyze the restricted-splitting value $\valsix(T_{i,j,k}, \tilde\alpha_{i,j,k})$, the splitting restriction $\tilde\alpha_{i,j,k}$ is regarded as a predetermined input parameter given in advance. Its role is different from the component distribution $\alpha$ and the splitting distribution $\alpha^{(i,j,k)}$, which we treat as variables to be carefully selected in order to optimize the value's lower bound.
\end{remark}

\subsection{Hashing Methods}
\label{sec:hashing}
The hashing method is an important step in the laser method. It was first introduced in \cite{coppersmith1987matrix}. Below we mostly consider the hashing method applied on the level-$\lvl$ partition of the CW tensor power, $\T \defeq \big( \CW_q^{\otimes 2^{\lvl-1}} \big)^{\otimes n}$. Initially, the tensor has many triples, and each X, Y, or Z-block appears in multiple triples. We zero out some blocks during the hashing method to remove the triples containing those blocks. We aim to let some types of blocks only appear in a single triple.

Most previous works use the hashing method with Salem-Spencer set to zero out some blocks, so that finally each block $X_I$, $Y_J$, or $Z_K$ only appears in at most one retained triple $(X_I, Y_J, Z_K)$. Here, X, Y, and Z variables are symmetric, so we call this setting \emph{symmetric hashing}. In our paper, besides the symmetric setting, we also use a generalized asymmetric setting appeared in~\cite{coppersmith1987matrix} so that an X or Y-block only appears in a single retained triple, but a Z-block can be in multiple retained triples.

Both symmetric and asymmetric hashing begins by choosing a distribution $\alpha$ over level-$\lvl$ components. Given such distribution $\alpha$, let $\alphx, \alphy, \alphz$ be its inducing marginal distributions. As the goal of the hashing method, we want to keep a number of triples consistent with $\alpha$ while zeroing out others; we call the triples consistent with $\alpha$ \emph{good} triples.

The very first step is to zero out variable blocks inconsistent with $\alphx, \alphy$, and $\alphz$, since these blocks do not appear in good triples. After that, let $\numxblock$ be the number of remaining X-blocks $X_I$, and define $\numyblock, \numzblock$ similarly. We expect that the distribution $\alpha$ satisfies the following:
\begin{itemize}
\item $\numxblock = \numyblock \ge \numzblock$.
\item Let $\numtriple$ be the number of triples whose blocks are consistent with $\alphx, \alphy, \alphz$ (i.e., the remaining blocks so far). For every block $X_I$ or $Y_J$, we require the number of triples containing it to be exactly $\numtriple / \numxblock$, which is the same for all such blocks.
\item Let $\numalpha$ be the number of triples consistent with the joint distribution $\alpha$, i.e., the number of good triples. For every block $X_I$ or $Y_J$, we require that the number of \emph{good} triples containing it equals the same number $\numalpha / \numxblock$.
\item The number of triples containing every $Z_K$ is $\numtriple / \numzblock$, and the number of good triples containing every $Z_K$ is $\numalpha / \numzblock$.
\end{itemize}

Pick $M$ as a prime which is at least $4 \numtriple / \numxblock$, and construct a Salem-Spencer set $B$ of size $M^{1-o(1)}$ in which no three numbers form an arithmetic progression (modulo $M$). Select $n+1$ independently uniformly random integers $0 \leq b_0, w_t < M$ for $t\in \{0,\cdots,n\}$. For blocks $X_I, Y_J, Z_K$, compute the hash functions:
\begin{align*}
  \hashx(I) & = b_0 + \Big(\sum_{t=1}^n w_t\cdot I_t\Big) \bmod M, \\
  \hashy(J) & = b_0 + \Big(w_0 + \sum_{t=1}^n w_t\cdot J_t\Big) \bmod M, \\
  \hashz(K) & = b_0 + \frac{1}{2}\Big(w_0 + \sum_{t=1}^n  w_t\cdot (2^\lvl - K_t)\Big) \bmod M.
\end{align*}
(Since $M$ is odd, division by 2 modulo $M$ is well defined.)
We can see that for any triple $(X_I, Y_J, Z_K)$ in $\mathcal T$, $\hashx(I) + \hashy(J) \equiv 2 \hashz(K) \pmod M$. Zero out all blocks $X_I, Y_J, Z_K$ whose hash values $\hashx(I)$, $\hashy(J)$, or $\hashz(K)$ are not in $B$, then all remaining triples $(X_I, Y_J, Z_K)$ must satisfy $\hashx(I) = \hashy(J) = \hashz(K)\in B$ by \cref{thm:salem-spencer}.

We may think the hash function maps all variable blocks into buckets $b \in \BK{0, \ldots, M - 1}$; for a triple $(X_I, Y_J, Z_K)$, it is retained in this zeroing-out step only if the three variable blocks are mapped to the same bucket $b \in B$.

It is easy to calculate the expected number of remaining triples after the above zeroing-out step. For each of the $\numalpha$ good triples $(X_I, Y_J, Z_K)$, the probability that $\hashx(I) = \hashy(J) = \hashz(K) = b$ is $M^{-2}$ since $\hashz(K)$ can be determined by $\hashx(I)$ and $\hashy(J)$. ($\hashx(I)$ and $\hashy(J)$ are independent because of the randomness of $w_0$.) So the expected number of remaining good triples with hash value $b$ is $\numalpha / M^2$. Multiplied by the size $|B| = M^{1 - o(1)}$ of the Salem-Spencer set, we get $\numalpha \cdot M^{-1 - o(1)}$ which is the expected number of remaining good triples in total.

For each $b\in B$, we have a list of remaining (not necessarily good) triples $(X_I, Y_J, Z_K)$ satisfying $\hashx(I) = \hashy(J) = \hashz(K) = b$. If all remaining triples with hash value $b$ were disjoint (i.e., do not share variables), then our goal could be achieved easily. Otherwise, we resolve collisions by zeroing out some blocks. This second zeroing-out step depends on the setting: whether we allow sharing Z-blocks or not.

\paragraph{Asymmetric Hashing.}

We first see the case where sharing Z-blocks is allowed. Then what we need to do is just to eliminate remaining triples sharing an X or Y-block. We greedily find a pair of triples sharing X or Y-blocks, and zero out any involved\footnote{For example, when $(X_I, Y_J, Z_K)$ and $(X_I, Y_{J'}, Z_{K'})$ share a block $X_I$, we may zero out all of $X_I, Y_J, Y_{J'}$, or just any of them. But we cannot zero out Z-blocks.} X or Y-blocks to resolve the collision; this process is repeated until no X or Y-blocks are shared. After that, no two remaining triples can share X or Y-blocks. Finally, we zero out every $X_I$ if its triple $(X_I, Y_J, Z_K)$ is not consistent with $\alpha$ (i.e., the triple is not good).

To analyze the expected number of remaining good triples, we only need to count the number of good triples (with hash value $b$) that do not share X or Y-blocks with any other triple. These triples will not be removed regardless of the order of checking triples in the greedy process.

Fix a hash value $b \in B$. Initially there are $\numalpha M^{-2}$ good triples mapped to $b$ in expectation. Then, assume $(X_I, Y_J, Z_K)$ and $(X_I, Y_{J'}, Z_{K'})$ are two triples sharing an X-block, where the former one is good. If they were mapped to the same value $b$, the good triple $(X_I, Y_J, Z_K)$ no longer meets the requirement and we need to substract one from the total number of good triples.\footnote{Although zeroing out $X_I$ may affect good triples other than $(X_I, Y_J, Z_K)$ and $(X_I, Y_{J'}, Z_{K'})$, its loss will be counted when we regard it as the former triple in the pair.}
This probability for a fixed triple pair is $M^{-3}$ according to the following lemma:

\begin{restatable}[Implicit in \cite{coppersmith1987matrix}]{lemma}{HashIndependence}
  \label{lemma:hash_independence}
  For two different triples $(X_I, Y_J, Z_K)$ and $(X_{I'}, Y_{J'}, Z_K)$ sharing a Z-block, $\Pr[\hashx(I') = \hashz(K) \mid \hashx(I) = \hashz(K)] = M^{-1}$. Moreover, the probability that $\hashx(I) = \hashx(I') = \hashz(K) = b$ (which implies $\hashy(J) = \hashy(J') = b$) for a fixed $0 \le b < M$ is exactly $M^{-3}$. Same for triple pairs sharing X or Y-blocks.
\end{restatable}
\begin{proof}
  We first show that the events $\hashx(I) = \hashz(K)$ and $\hashx(I') = \hashz(K)$ are independent. Fixing the Z-block $Z_K$, we define
  \[
    h(I) \defeq \hashx(I) - \hashz(K) \equiv \frac{w_0}{2} + \sum_{t=1}^n w_t \cdot \bk{I_t + \frac{1}{2} K_t - 2^{\lvl - 1}}.
  \]
  It is a pairwise independent uniform hash function for all $I$. Therefore, $\Pr[\hashx(I) = \hashx(I') = \hashz(K)] = \Pr[h(I) = h(I') = 0] = M^{-2}$. It follows that $\Pr[\hashx(I') = \hashz(K) \mid \hashx(I) = \hashz(K)] = M^{-1}$.

  Next, conditioned on $\hashz(K) = \hashx(I) = \hashx(I')$, the probability that $\hashz(K) = b$ for a certain $b$ is exactly $1/M$ according to the randomness of $b_0$. (The definition of hash functions here is slightly different than \cite{coppersmith1987matrix} for shorter analysis.\footnote{In \cite{coppersmith1987matrix}, they do not have the random constant term $b_0$; they use Chebyshev's inequality to bound the number of good triples.}) Hence, $\Pr[\hashz(K) = \hashx(I) = \hashx(I') = b] = \Pr[\hashz(K) = \hashx(I) = \hashx(I')] \cdot M^{-1} = M^{-3}$.
\end{proof}

The number of such triple pairs (where the former one is a good triple) equals $\numxblock \cdot (\numalpha / \numxblock) \cdot (\numtriple / \numxblock) = \numalpha \numtriple / \numxblock$. For each pair, with probability $|B| \cdot M^{-3}$ we lose a good triple. The same loss is counted for triples sharing Y-blocks. Thus the expected number of remaining good triples is at least
\[
  |B| \cdot \bk{\numalpha M^{-2} - 2 \cdot \frac{\numalpha \numtriple}{\numxblock} \cdot M^{-3}}
  \ge M^{1 - o(1)} \cdot \frac{1}{2} \numalpha M^{-2} \ge \numalpha M^{-1 - o(1)},
\]
where the first inequality holds as $\numtriple / \numxblock \cdot M^{-1} \le 1/4$. A typical value of $M$ is $\Theta(\numtriple / \numxblock)$, which keeps at least $\numxblock \numalpha / \numtriple \cdot 2^{-o(n)}$ good triples.

\paragraph{Hash Loss.}
Ideally, when $R \defeq \numalpha / \numtriple = 2^{-o(n)}$, almost every of the $\numxblock$ X-blocks can survive the hashing process. (Strictly, $\numxblock \cdot 2^{-o(n)}$ X-blocks are retained, while the factor $2^{-o(n)}$ will disappear as we only care the exponent when $n \to \infty$.) Such utilization rate of X-blocks is the best possible outcome of the hashing process. However, when the factor $R = 2^{-\Omega(n)}$ is non-negligible, $1-R$ fraction of the X-blocks are wasted, which we call the \emph{hash loss}. The quantity of the hash loss is measured by $R$.

In the following lemma, we prove that $\numtriple = \poly(n) \cdot \max_{\alpha' \in \distShareMargin} \numalpha[\alpha']$. We see that if $\alpha$ is the distribution with the maximum number of triples, we do not suffer from the hash loss ($R = 1/\poly(n) = 2^{-o(n)}$).

\begin{lemma}
  \label{lem:numtriple_singledist}
  For any fixed distribution $\alpha$ and its marginals $\alphx, \alphy, \alphz$, we have $\numtriple = \poly(n) \cdot \max_{\alpha' \in \distShareMargin} \numalpha[\alpha']$.
\end{lemma}
\begin{proof}
  For all triples with marginal distributions $\alphx, \alphy, \alphz$, we categorize them by its joint distribution. The number of categories is $\poly(n)$ as all entries of the joint distribution must be multiples of $1/n$. Then, let $\alpha' = \argmax_{\alpha' \in \distShareMargin} \numalpha[\alpha']$ be the category with the maximum number of triples, it contains at least $1/\poly(n)$ fraction of all $\numtriple$ triples, which concludes the proof.
\end{proof}

\paragraph{Symmetric Hashing.} In most of the prior work, they do not allow Z-blocks to be shared among triples. In this case, we only need to change the greedy process a bit: Not only when X or Y-blocks are shared, but also when a Z-block $Z_K$ is shared between triples $(X_I, Y_J, Z_K)$ and $(X_{I'}, Y_{J'}, Z_{K})$, we zero out all five involved blocks $X_I, X_{I'}, Y_J, Y_{J'}, Z_K$ or just any of them. This process is repeated until there are no triples sharing any variable blocks.

The symmetric hashing is only applied when $\numxblock = \numyblock = \numzblock$. The analysis is again similar: we count the number of good triples (i.e., triples consistent with $\alpha$) that do not share any blocks with other triples. The expected number of remaining good triples is at least
\[
  |B| \cdot \bk{
    \numalpha M^{-2} - 3 \cdot \frac{\numalpha \numtriple}{\numxblock} \cdot M^{-3}
  }
  \ge M^{1-o(1)} \cdot \frac{1}{4} \numalpha M^{-2}
  \ge \numalpha M^{-1-o(1)}.
\]
When we choose $M = \Theta(\numtriple / \numalpha)$, we can keep $\numxblock \numalpha / \numtriple \cdot 2^{-o(n)}$ good triples. Similar to above, when $R = \numalpha / \numtriple \ne 2^{-o(n)}$, we cannot utilize all the variables due to the hash loss.

\paragraph{More general settings.} Above, we discussed the case where the starting tensor is $\T = (\CW_q^{\otimes 2^{\lvl - 1}})^{\otimes n}$. In fact, we only need to use the fact that $\T$ is a tensor power of a $P$-partitioned tensor ($P = 2^{\lvl}$ in our setting, see \cref{sec:partition-tensor} for definition). Formally, suppose $T_1$ is a $P$-partitioned tensor, we can apply the hashing method to $T_1^{\otimes n}$ as long as the requirements (e.g., $\numxblock = \numyblock \ge \numzblock$) are met, obtaining independent triples consistent with the given distribution $\alpha$ over $T_1$'s components. More generally, we can also apply the hashing method onto a subtensor of the product of tensor powers $T_1^{\otimes n_1} \otimes T_2^{\otimes n_2} \otimes \cdots \otimes T_s^{\otimes n_s}$, as long as:
\begin{itemize}
\item $s$ is a constant integer, and $T_1, \ldots, T_s$ are $P$-partitioned tensors with the same $P$.
\item For each region $r \in [s]$, a joint distribution $\alpha^{(r)}$ over the components of $T_r$ is given in advance, with marginals $\alphx^{(r)}$, $\alphy^{(r)}$, and $\alphz^{(r)}$.
\item Let $\numxblock$ be the number of X-blocks $X_I$ that is consistent with $\alphx^{(r)}$ in all regions $r \in [s]$. Similarly define $\numyblock$ and $\numzblock$. We require $\numxblock = \numyblock \ge \numzblock$ (for the asymmetric hashing) or $\numxblock = \numyblock = \numzblock$ (for the symmetric hashing).
\item Let $\numtriple$ be the number of triples consisting of variable blocks satisfying the distribution constraint (there are $\numxblock, \numyblock, \numzblock$ many such blocks). The number of such triples containing each $X_I$ or $Y_J$ is the same number $\numtriple / \numxblock$. The number of such triples containing each $Z_K$ is $\numtriple / \numzblock$.
\item Let $\numalpha$ be the number of triples consistent with $\alpha^{(r)}$ in all regions $r$. We call these triples \emph{good} triples. The number of good triples containing each $X_I$ or $Y_J$ should be the same number $\numalpha / \numxblock$. That of Z-blocks $Z_K$ is $\numalpha / \numzblock$.
\end{itemize}
Then, we can apply the hashing method to obtain $\numxblock \numalpha / \numtriple \cdot 2^{-o(n)}$ independent good triples where $n \defeq n_1 + \cdots + n_s$. The proof of the more general setting is the same as above, so we omit it here.

\section{Improving the Second Power of CW Tensor}
\label{sec:2nd}
This section offers a formal analysis of the second power to illustrate our main ideas. We will obtain a bound $\omega < 2.375234$ by analyzing the second power, which beats the previous best bound on the second power ($\omega < 2.375477$ in \cite{coppersmith1987matrix}).\footnote{Since the optimization problem for the second power can be solved exactly by hand, such an improvement can only come from new ideas, rather than new heuristics for optimization.} The technique will be generalized to higher powers in later sections. To help the reader keep track of the notation, we will summarize all the notations that will be used in \Cref{table:sec-4}.

\begin{table}
  \renewcommand{\arraystretch}{1.2}
  \begin{center}
    \begin{tabular}{|c l|}
      \hline
      \textbf{Notation}  & \textbf{Meaning} \\ \hline
      $\ell$ & The level we are considering ($\ell = 1, 2$). \\
      $\T$ & The $n$-th power of first/second power of CW tensor $\T = (\CW_q^{2^{\ell - 1}})^{\otimes n}$\;\;($\ell = 1, 2$).\\
      $T_{i,j,k}$ & The component of CW tensor $\CW_q^{\otimes 2^{\ell - 1}} \vert_{X_i, Y_j, Z_k} $. \\
      $I, J, K$ & Index sequences of level-2 variable blocks. \\
      $\hat{I}, \hat{J}, \hat{K}$ & Index sequences of level-1 variable blocks. \\
      $X_I, Y_J, Z_K$ & Level-2 variables blocks. \\
      $X_{\hat{I}}, Y_{\hat{J}}, Z_{\hat{K}}$ & Level-1 variables blocks. \\
      $(X_I, Y_J, Z_K)$ & A level-2 triple of three variable blocks. \\
      $\T \vert_{X_I, Y_J, Z_K}$ & The subtensor of $\T$ over the triple $(X_I, Y_J, Z_K)$. \\
      $S_{i,j,k}$ & The set of positions with component $(i, j, k)$ of a triple: \\
                         & \qquad $S_{i,j,k} = \{t \in [n] \mid I_t = i, J_t = j, K_t = k\}$. \\
      $\alpha(i,j,k)$ & The joint distribution $\alpha(i,j,k) = \frac{1}{n}|S_{i,j,k}|$ of a triple. \\
      $\alphx, \alphy, \alphz$ & The marginal distributions of $\alpha$. \\
      $D_{\alpha}$ & The set of all joint distributions that have the same marginals as $\alpha$. \\
      $\split(\hat{K}, S)$ & The split distribution  of $\hat{K}$ when restricted to positions in $S$. \\
      $\numxblock, \numyblock, \numzblock$ & Number of X/Y/Z blocks consistent with $\alphx, \alphy, \alphz$. \\
      $\numalpha$ & The number of level-2 triples with joint distribution $\alpha$. \\
      $\numtriple$ & The number of level-2 triples with marginal distributions $\alphx, \alphy, \alphz$. \\
      $\pcomp$ & The probability that a uniformly random triple obeying $\alpha$ and containing $Z_K$ \\
                         & \qquad is consistent with a fixed $Z_{\hat{K}} \in Z_K$. \\
			\hline
		\end{tabular}
		\caption{Summary of Notations in Section 4} \label{table:sec-4}
	\end{center}
\end{table}

\subsection{Coppersmith-Winograd Algorithm}
\label{sec:recap_cw}
The framework of our improved algorithm is similar to \cite{coppersmith1987matrix}. They analyzed the CW tensor's first and second power ($\lvl = 1$ or $\lvl = 2$) using the following steps:
\begin{enumerate}
\item Lower bound the value of each level-$\lvl$ component $T_{i,j,k}$ (the subtensor of $\CW_q^{\otimes 2^{\lvl -1}}$ over $X_i, Y_j, Z_k$ where $i + j + k = 2^\lvl$, as defined in Section~\ref{subsec:leveled}). Throughout this section,  we have either $\lvl = 1$ or $\lvl = 2$. 

\item Choose a distribution $\alpha$ over all level-$\lvl$ components $T_{i,j,k}$'s. Let $\alpha(i,j,k)$ be its probability mass on $T_{i,j,k}$. This distribution $\alpha$ has to be symmetric about its X, Y, and Z dimensions, i.e., $\alpha(i,j,k) = \alpha(i,k,j) = \alpha(j,i,k) = \alpha(j,k,i) = \alpha(k,i,j) = \alpha(k,j,i).$

\item Apply the symmetric hashing method (see \Cref{sec:hashing}) on $\T \defeq (\CW_q^{\otimes 2^{\lvl-1}})^{\otimes n}$ to zero out some (level-$\lvl$) variable blocks. All remaining triples are guaranteed to be independent and consistent with $\alpha$.

Formally, if $(X_I, Y_J, Z_K)$ is a triple retained in the hashing method, then (1) $\alpha(i,j,k)$ equals the proportion of positions $t \in [n]$ which satisfies $(I_t, J_t, K_t) = (i,j,k)$, and (2) the blocks $X_I, Y_J$, and $Z_K$ do not appear in any other retained triples.
\item Degenerate each retained triple to a direct sum of matrix multiplication tensors. Specifically, we will use the value lower bounds from Step 1 and obtain the value lower bound of each triple. (This implicitly gives a degeneration into matrix multiplications.)
\end{enumerate}
After these steps, $\mathcal T$ has degenerated into a direct sum of matrix multiplication tensors, which leads to a lower bound on the value of $\CW_q^{\otimes 2^{\lvl-1}}$. Finally, $\tauthm$ is applied to obtain an upper bound of $\omega$. Below, we instantiate such procedure for the first and second power separately.

\paragraph{Analysis of the First Power.} For the first power, $\T = \CW_q^{\otimes n}$. The level-1 components $T_{i,j,k}$'s ($i + j + k = 2$) are simply matrix multiplication tensors. So we can skip the first and the fourth step. In the second step, a symmetric distribution $\alpha$ over $T_{0,1,1}, T_{0,0,2}$ (and their permutations) is selected:
\begin{align*}
  \alpha(0, 1, 1) = \alpha(1, 0, 1) = \alpha(1, 1, 0) &= a, \\
  \alpha(0, 0, 2) = \alpha(0, 2, 0) = \alpha(2, 0, 0) &= b,
\end{align*}
where $3a + 3b = 1$. Let $\alphx, \alphy, \alphz$ be its marginal distributions. By \Cref{thm:unique-dis-level1}, $\alpha$ is the only distribution consistent with these marginal distributions, which allows us to use the hashing method described in \Cref{sec:hashing} without a hash loss. We can see that $\alphx{(0)} = a + 2b, \, \alphx{(1)} = 2a, \, \alphx{(2)} = b$. Since $\alpha$ is symmetric, $\alphy$ and $\alphz$ are the same as $\alphx$. Let $\numxblock$ be the number of X-blocks consistent with $\alphx$.

After hashing and zeroing out, the number of remaining triples is
\[
  \numxblock \cdot 2^{-o(n)} = \binom{n}{(a + 2b)n, \; 2an, \; bn} \cdot 2^{-o(n)}.
\]
We have degenerated $\mathcal T \defeq \CW_q^{\otimes n}$ into the direct sum of the retained triples, while each retained triple is isomorphic to a matrix multiplication tensor
\[(T_{0,1,1} \otimes T_{1,0,1} \otimes T_{1,1,0})^{\otimes an} \otimes (T_{0,0,2} \otimes T_{0,2,0} \otimes T_{2,0,0})^{\otimes bn} \cong \angbk{q^{an}, q^{an}, q^{an}}.\]
This leads to the bound
\begin{align*}
  \valthree(\CW_q) &\ge \lim_{n\to\infty} \bk{q^{3\tau an} \binom{n}{(a + 2b)n, \; 2an, \; bn}}^{1/n} = \frac{q^{3\tau a}}{(a+2b)^{a+2b} \cdot (2a)^{2a} \cdot b^b}.
\end{align*}
Let $q = 6$ and $b = 0.016$, it gives $\omega < 2.38719$.

\paragraph{Analysis of the Second Power.} For the second power, $\T= (\CW_q^{\otimes 2})^{\otimes n}$. We first lower bound the value of level-2 components $T_{i,j,k}$ ($i + j + k = 4$).

\begin{lemma}[\cite{coppersmith1987matrix}]
  \label{lem:order-2-values}
  The values of level-2 components $T_{i,j,k}$ are:
  \begin{itemize}
  \item $\valthree(T_{0, 0, 4}) = 1$, $\valthree(T_{0, 1, 3}) = (2q)^{\tau}$, $\valthree(T_{0, 2, 2}) = (q^2 + 2)^{\tau}$;
  \item $\valthree(T_{1, 1, 2}) \geq 2^{2/3} q^{\tau} (q^{3\tau} + 2)^{1/3}$.
  \end{itemize}
\end{lemma}

\begin{remark} \label{remark:alpha_sym}
    Note the value $\valthree(T_{i,j,k})$ is defined (in \Cref{def:value}) using the symmetrization of $T_{i,j,k}$, i.e. $\sym_3(T_{i,j,k}) = T_{i,j,k} \otimes T_{j,k,i} \otimes T_{k,i,j}$. This implies that to use these values, we must ensure that $\alpha(i,j,k) = \alpha(j,k,i) = \alpha(k,i,j)$ in Step 2.\footnote{These constraints are incompatible with our algorithm which relies on the asymmetric hashing method. We will relax them in \Cref{sec:non-rot} and \Cref{sec:compa}.} Specifically, $\valthree(T_{i,j,k})$ captures the capability of $\sym_3(T_{i,j,k})^{\otimes m}$ for matrix multiplication, as $m\rightarrow \infty$. If we want to apply it to $T_{i,j,k}^{\otimes \alpha(i,j,k) n} \otimes T_{j,k,i}^{\otimes \alpha(j,k,i) n} \otimes T_{k,i,j}^{\otimes \alpha(k,i,j) n}$, we must ensure $\alpha(i,j,k)n = \alpha(j,k,i) n = \alpha(k,i,j) n = m$. 
\end{remark}

Next, in step 2, we pick a symmetric distribution $\alpha$ over all level-2 components $T_{i,j,k}$ where $i + j + k = 4$:
\begin{align*}
  &\alpha(0,0,4) = \alpha(0,4,0) = \alpha(4,0,0) = a, \\
  &\alpha(0,1,3) = \alpha(0,3,1) = \alpha(1,0,3) = \alpha(1,3,0) = \alpha(3,0,1) = \alpha(3,1,0) = b, \\
  &\alpha(0,2,2) = \alpha(2,0,2) = \alpha(2,2,0) = c, \\
  &\alpha(1,1,2) = \alpha(1,2,1) = \alpha(2,1,1) = d.
\end{align*}

Step 3 uses the symmetric hashing method to zero out some level-2 triples. 

\begin{remark} \label{remark:hash_loss1}
  Here is one subtlety (that readers may skip for the first read). Zeroing out can only distinguish between different marginal distributions but not joint distributions (See e.g. \cite{alman2021}). Hence all joint distributions consistent with $\alphx, \alphy, \alphz$ will remain after zeroing out. If $\alpha$ is not the maximum entropy distribution among all joint distributions consistent with $\alphx, \alphy, \alphz$ (which we denote by $D^*(\alphx, \alphy, \alphz)$), the symmetric hashing method would incur an extra hash loss. For Coppersmith and Winograd's analysis, $\alpha$ indeed equals $D^*(\alphx, \alphy, \alphz)$.
\end{remark}
Notice that $\alphx{(0)} = 2a + 2b + c, \, \alphx{(1)} = 2b + 2d, \, \alphx{(2)} = 2c + d, \, \alphx{(3)} = 2b, \, \alphx{(4)} = a$. Same for $\alphy$ and $\alphz$. Hence after Step 3, the number of retained triples is
\[
  \numxblock \cdot 2^{-o(n)} = \binom{n}{(2a + 2b + c)n, \, (2b + 2d)n, \, (2c + d)n, \, 2bn, \, an} \cdot 2^{-o(n)},
\]
where $\numxblock$ is the number of level-2 X-blocks consistent with $\alphx$. The remaining subtensor of $\T \defeq (\CW_q^{\otimes 2})^{\otimes n}$ is the direct sum of the remaining triples, while each remaining triple is isomorphic to
\[
  \begin{aligned}
    \T^\alpha &\defeq \bigotimes_{i+j+k=4} T_{i,j,k}^{\otimes n\alpha(i,j,k)} = \bigotimes_{\substack{
               (i,j,k) \in \{
               (0,0,4), \\ (0,1,3), (0,3,1), \\ (0,2,2), (1,1,2)
    \}
    }} \sym_3(T_{i,j,k}^{n\alpha(i,j,k)}).
  \end{aligned}
  \numberthis \label{eq:def_Talpha}
\]
In Step 4, we apply \Cref{lem:order-2-values} to each symmetrized term, by the super-multiplicative property of values, we get a lower bound of $\valthree(\CW_q^{\otimes 2})$:
\begin{align*}
  \valthree(\CW_q^{\otimes 2}) &\ge \lim_{n\to\infty} \binom{n}{(2a + 2b + c)n, \, (2b + 2d)n, \, (2c + d)n, \, 2bn, \, an}^{1/n} \cdot \prod_{i+j+k=4} \valthree(T_{i,j,k})^{\alpha(i,j,k)} \\
  &= \frac{1}{(2a + 2b + c)^{2a + 2b + c} (2b + 2d)^{2b + 2d} (2c + d)^{2c + d} (2b)^{2b} a^a} \cdot \prod_{i + j + k = 4} \valthree(T_{i,j,k})^{\alpha(i,j,k)}.
\end{align*}
Coppersmith and Winograd~\cite{coppersmith1987matrix} found that when $a = 0.000233,~ b = 0.012506,~ c = 0.102546,~ d = 0.205542$, and $q = 6$, we can get $\omega < 2.375477$. 

\subsection{Non-rotational Values}
\label{sec:non-rot}

As explained in \Cref{remark:alpha_sym}, to apply $\valthree(T_{i,j,k})$, we need the distribution $\alpha$ to be symmetric. However, to apply the asymmetric hashing method in our improved algorithm, at least for some $T_{i,j,k}$'s, the distribution $\alpha$ has to be asymmetric. Hence we must consider the values of some $T_{i,j,k}$'s without such symmetrization. We call them \emph{non-rotational values}. 

\begin{definition} \label{def:non-rot}
  The \emph{non-rotational value} of a tensor $T$, denoted by $V_\tau^{(\text{nrot})}(T)$, is defined as
  \[
    \valone(T) \defeq \limsup_{m\to\infty} \max
    \BKmid{\Big(\sum_{i=1}^s (a_i b_i c_i)^\tau\Big)^{1/m}}{T^{\otimes m} \degen \bigoplus_{i=1}^s \angbk{a_i,b_i,c_i}}.
  \]
  The only difference between this definition and the original definition of values is that we do not allow $T$ to be symmetrized before degeneration.
\end{definition}

Also, recall the definition of \emph{restricted-splitting values} $\valthree(T_{i,j,k}, \splres)$ in \Cref{sec:restricted_splitting}. We further define the non-rotational version of it.

\begin{definition} The \emph{non-rotational restricted-splitting value} is defined as 
\[
  \valone (T_{i,j,k}, \splres) = \limsup_{n\to\infty} \valone(T_{i,j,k}^{\otimes n}[\splres])^{1/n}.
\]
\end{definition}

For matrix multiplication tensors, $T_{0,j,k}$, their non-rotational value matches its symmetrized value. This is stated in the following lemma. Moreover, for $T_{1,1,2}$, we only give a lower bound on its symmetrized restricted-splitting value. Since almost all the contributions are from the most typical splitting distribution, it is not surprising that its restricted-splitting value matches the original value. We defer the proof of this lemma to \Cref{appendix:missing_proofs}. 

\begin{restatable}{lemma}{MoreValues}
  \label{lem:non-rot-values}
  We have:
  \begin{enumerate}[label=\textup{(\alph*)}]
  \item $\valone(T_{0,0,4}) = 1$,
  \item $\valone(T_{0,1,3}) = (2q)^\tau$,
  \item $\valone(T_{2,2,0}) = \valone(T_{0,2,2}, \splresB) = \valone(T_{2,0,2}, \splresB) = (q^2 + 2)^\tau$,
  \item $\valthree(T_{1,1,2}, \splresA) \ge 2^{2/3} q^{\tau} (q^{3\tau} + 2)^{1/3}$,
  \end{enumerate}
  where
  \begin{align*}
    \splresA(0) = \splresA(2) = \frac{1}{2+q^{3\tau}}, &\qquad \splresA(1) = \frac{q^{3\tau}}{2 + q^{3\tau}}; \numberthis \label{eq:tilde_A} \\
    \splresB(0) = \splresB(2) = \frac{1}{2+q^2}, &\qquad \splresB(1) = \frac{q^2}{2 + q^2}. \numberthis \label{eq:tilde_B}
  \end{align*}
\end{restatable}

\subsection{Compatibility} \label{sec:compa}

In \Cref{lem:non-rot-values}, we presented the restricted-splitting values for those $T_{i,j,k}$'s with $k = 2$, and these values match their original values in \Cref{lem:order-2-values}.  Intuitively, this means that, for any level-2 triple $(X_I, Y_J, Z_K)$, within $Z_K$, the only useful level-1 Z-blocks are those with a specific splitting distribution over the positions $\{t \in [n] \mid K_t = 2\}$. When one such level-1 Z-block is useful for $(X_I, Y_J, Z_K)$, we say that they are \emph{compatible}. 

Before we give the formal definition, we first set up some notations. For each triple $(X_I, Y_J, Z_K)$, we let $S_2$ be the set of positions $t$ where $K_t = 2$. We define $S_{i,j,k} = \BK{t \in [n] \mid (I_t, J_t, K_t) = (i,j,k)}$. Then $S_2 = S_{0,2,2} \cup S_{2,0,2} \cup S_{1,1,2}$. Let $Z_{\hat{K}}$ be a level-1 Z-block in $Z_K$. We define its split distribution over a subset of positions as follows.

\begin{definition} Fix a level-1 Z-block $Z_{\hat{K}}$ where $\hat{K}$ is the level-1 index sequence $(\hat{K}_1, \hat{K}_2, \cdots, \hat{K}_{2n})$. For any subset $S \subseteq [n]$, we define $\split(\hat K, S)$ to be the (marginal) split distribution of positions of $S$ in $\hat K$.
\[
  \split(\hat K, S)(k_l, k_r) = \frac{1}{|S|}\Big|\{t \in S \mid \hat K_{2t - 1} = k_l \land \hat K_{2t} = k_r \}\Big|.
\]
Since we are only considering the split distribution of $2$, we require that $S \subseteq S_2$ and $k_l = 0, 1, 2$. It has support $\{(0,2), (1,1), (2,0)\}$. For simplicity, we write $\split(\hat K, S)(k_l) \coloneqq \split(\hat K, S)(k_l, 2 - k_r)$. ($\split(\hat I, S)$ and $\split(\hat J, S)$ can be defined similarly.)
\end{definition}

Now we are ready to state the condition for $Z_{\hat{K}}$ to be useful for $(X_I, Y_J, Z_K)$. 

\begin{definition}[Compatibility]
  A level-1 block $Z_{\hat K}$ is said to be \emph{compatible} with a level-2 triple $(X_I, Y_J, Z_K)$ obeying distribution $\alpha$ if the following conditions are satisfied:
  \begin{itemize}
  \item $\split(\hat K, S_{1,1,2}) = \splresA$;
  \item $\split(\hat K, S_{0,2,2}) = \split(\hat K, S_{2,0,2}) = \splresB$.
  \end{itemize}
  Here $\splresA$ and $\splresB$ are defined as in \Cref{lem:non-rot-values}.
  \label{def:complv2}
\end{definition}

\Cref{def:complv2} directly implies that, in order for a level-1 block $Z_{\hat K}$ to be compatible with any triple $(X_I, Y_J, Z_K)$ obeying $\alpha$, it must satisfy
\[
  \numberthis \label{eq:avg_split_correct_sec4}
  \begin{aligned}
    \split(\hat K, S_2) &= \frac{\alpha(1,1,2)}{\alpha(1,1,2) + 2 \alpha(0,2,2)} \cdot \splresA + \frac{2 \alpha(0,2,2)}{\alpha(1,1,2) + 2\alpha(0,2,2)} \cdot \splresB \eqdef \splresAverage.
  \end{aligned}
\]

We will call this quantity $\splresAvg$. To see that this equality holds, note that $\split(\hat{K}.S_{1,1,2}) = \splresA$ and $|S_{1,1,2}| / |S_2| = \frac{\alpha(1,1,2)}{\alpha(1,1,2) + 2\alpha(0,2,2)}$. Similarly for $\split(\hat{K}, S_{0,2,2})$ and $\split(\hat{K}, S_{2,0,2})$. \\

Next, we will analyze the probability for a fixed $Z_{\hat{K}} \in Z_K$ to be compatible with a random triple $(X_I, Y_J, Z_K)$. (Here the randomness is over $I, J$.) We will call it $ p_{\textup{comp}}$. This probability is exactly the ``combination loss'': Suppose $Z_K$ is only in one remaining triple after hashing. Inside $Z_K$, each $Z_{\hat{K}}$ has only $p_{\textup{comp}}$ probability of being compatible with that triple. Hence intuitively, $1 - p_{\textup{comp}}$ fraction of $Z_{\hat{K}}$'s are simply wasted.

\begin{lemma}
  \label{lem:p_comp_sec4}
  Fix a level-2 block $Z_K$ and a level-1 block $Z_{\hat K} \in Z_K$ satisfying \eqref{eq:avg_split_correct_sec4}. Let $(X_I, Y_J, Z_K)$ be a uniformly random triple among all triples that obey $\alpha$ and contain $Z_K$. Suppose $\alpha(0,2,2) = \alpha(2,0,2) = c$ and $\alpha(1,1,2) = d$. Then we have
  \[
    \begin{aligned}
      p_{\textup{comp}} &\defeq \phantom{{}={}} \Pr_{X_I, Y_J}[Z_{\hat K}\textup{ is compatible with }(X_I, Y_J, Z_K)] \\
      &=\binom{cn}{[cn \cdot \splresB(k')]_{k'=0,1,2}}^2 \binom{dn}{[dn \cdot  \splresA(k')]_{k'=0,1,2}} \bigg/ \binom{(2c + d)n}{[(2c + d)n \cdot \splresAvg(k')]_{k'=0,1,2}},
    \end{aligned}
  \]
  where $\splresA$, $\splresB$, and $\splresAverage$ are defined in \eqref{eq:tilde_A}, \eqref{eq:tilde_B}, and \eqref{eq:avg_split_correct_sec4}, respectively. This probability is the same for any fixed $Z_{\hat K}$.
\end{lemma}

\begin{proof}
Consider the following distribution: $(X_I, Y_J, Z_K)$ is a random level-2 triple consistent with $\alpha$, and $Z_{\hat K}$ is a random level-1 Z-block inside $Z_K$ satisfying \eqref{eq:avg_split_correct_sec4}. We have
\[
  \begin{aligned}
    p_{\textup{comp}}
    &= \Pr_{I, \hat K}[Z_{\hat K} \text{ is compatible with } (X_I, Y_J, Z_K)] \\
    &= \E_{I}\Big[ \Pr_{\hat K} [Z_{\hat K} \text{ is compatible with } (X_I, Y_J, Z_K)] \Big] \\
    &= \frac{\#(Z_{\hat K} \in Z_K \text{ compatible with } (X_I, Y_J, Z_K))}
      {\#(Z_{\hat K} \in Z_K \text{ satisfying \eqref{eq:avg_split_correct_sec4}})} \quad \text{for any fixed }(X_I, Y_J, Z_K).
  \end{aligned}
\]
In the second line, the content inside the expectation notation is independent of $I$ due to symmetry. Thus we can calculate this probability for any fixed triple $(X_I, Y_J, Z_K)$ consistent with $\alpha$:
\[
  \#(Z_{\hat K} \in Z_K \text{ compatible with } (X_I, Y_J, Z_K)) = \binom{cn}{[cn \cdot \splresB(k')]_{k'=0,1,2}}^2 \binom{dn}{[dn \cdot \splresA(k')]_{k'=0,1,2}}
\]
and
\[
  \#(Z_{\hat K} \in Z_K \text{ satisfying \eqref{eq:avg_split_correct_sec4}}) = \binom{(2c + d) n}{[ (2c + d) n \cdot \splresAvg(k')]_{k'=0,1,2}}.
\]
This finishes the proof.
\end{proof}

\subsection{Variant of the Coppersmith-Winograd Algorithm} \label{sec:modi}

We now present a slightly modified version of the Coppersmith-Winograd algorithm, which results in the same bound of $\omega$. It contains an additional zeroing-out step which explicitly emphasizes that only compatible Z-blocks contribute to the algorithm. Illustrating this important idea is to prepare for our improved algorithm in the next subsection. Compared to the original version of the CW algorithm, here are two main differences:
\begin{itemize}
\item We only ensure that $\alpha(1,1,2) = \alpha(1,2,1) = \alpha(2,1,1)$ for $T_{1,1,2}$. For all other $T_{i,j,k}$'s we do not require such symmetry since we will apply their non-rotational value.\footnote{\label{note:T112} Here $T_{1,1,2}$ is special because it is the only component that is not a matrix multiplication tensor. As a result, in \Cref{lem:non-rot-values}, it is the only component without a non-rotational value bound. Hence we need such symmetry to apply its value bound.}  This is crucial for our improved algorithm which benefits from asymmetric hashing. 
\item For those $T_{i,j,k}$ with $k = 2$, we replace the use of its value with its restricted-splitting value. Such change explicitly emphasizes the fact that only a few compatible level-1 blocks are used in each level-2 triple $(X_I, Y_J, Z_K)$.
\end{itemize}

\paragraph*{Modified Algorithm.} Same as before, we consider the tensor power $\T = (\CW_q^{\otimes 2})^{\otimes n}$. In Step 1, instead of lower bounding values $V^{(3)}_{\tau}(T_{i,j,k})$ for each $i,j,k$, we use the bounds given by \Cref{lem:non-rot-values}. Specifically, for $T_{1,1,2}$, the value we apply still requires the distribution $\alpha$ to be symmetric about $T_{1,1,2}$. 

In Step 2, we have to ensure that $\alpha(1,1,2) = \alpha(1,2,1) = \alpha(2,1,1)$. For other $i,j,k$'s, $\alpha$ may not be symmetric. Moreover, we also require that we get the same number of level-2 X, Y, and Z variables blocks obeying $\alpha$, i.e., $\numxblock = \numyblock = \numzblock$. 

As a result, in Step 3, we can still apply symmetric hashing. After hashing, we have independent level-2 triple $(X_I, Y_J, Z_K)$'s. Every level-2 variable block can only appear in a single triple. That is, we get a subtensor $\bigoplus_{I,J,K} \T \vert_{X_I,Y_J,Z_K}$, where the direct sum is taken over all remaining triples $(I,J,K)$, and $\T \vert_{X_I, Y_J, Z_K}$ is the subtensor of $\T$ over variable sets $X_I, Y_J, Z_K$. Each $\T \vert_{X_I,Y_J,Z_K}$ is isomorphic to $\T^{\alpha} \coloneqq \bigotimes_{i +j+k=4} T^{\otimes n \alpha(i,j,k)}_{i,j,k}$.

Before Step 4, we additionally zero out all level-1 blocks $Z_{\hat K} \in Z_K$ that are not compatible with $(X_I, Y_J, Z_K)$, where $(X_I, Y_J, Z_K)$ is the unique remaining triple containing the level-2 block $Z_K$. Let $S_{i,j,k}$ and compatibility be defined as in \Cref{def:complv2}. $Z_{\hat{K}}$ survives such zeroing-out only when $\split(\hat{K}, S_{1,1,2}) = \splresA$, $\split(\hat{K}, S_{0,2,2}) = \split(\hat{K}, S_{2,0,2}) = \splresB$. Hence the remaining subtensor in each $\T\vert_{X_I,Y_J,Z_K}$ is isomorphic to
\[
  \T^* \coloneqq \Big( \bigotimes_{\substack{i+j+k=4 \\ k \ne 2}} T_{i,j,k}^{\otimes n \alpha(i,j,k)} \Big) \otimes T_{0,2,2}^{\otimes n \alpha(0,2,2)}[\splresB] \otimes T_{2,0,2}^{\otimes n \alpha(2,0,2)}[\splresB] \otimes T_{1,1,2}^{\otimes n \alpha(1,1,2)}[\splresA].
  \numberthis \label{eq:def_Tstar}
\]
This additional step allows us to restrict the split distribution of all remaining $Z_{\hat{K}}$.

Finally, in Step 4, we will use the following values:
\begin{itemize}
\item For $T_{0,2,2}^{\otimes n \alpha(0,2,2)}[\splresB]$, we use its \emph{non-rotational restricted-splitting value} $\valone(T_{0,2,2}, \splresB)$. Similar for $T_{2,0,2}^{\otimes n \alpha(2,0,2)}[\splresB]$.
\item Recall that $\alpha(1,1,2) = \alpha(1,2,1) = \alpha(2,1,1) = d$. For $\sym_3(T_{1,1,2}^{\otimes n \alpha(1,1,2)})$, we use its \emph{restricted-splitting value} $\valthree(T_{1,1,2}, \splresA)$.
\item For other components $T_{i,j,k}$, we use their \emph{non-rotational value} $\valone(T_{i,j,k})$.
\end{itemize}

Note that the values in \Cref{lem:non-rot-values} match the values in the original analysis. So this analysis gives exactly the same bound on $\omega$ as the original analysis from the original parameter $\alpha$. The difference is that now from \eqref{eq:def_Tstar}, we can explicitly see that only those level-1 blocks $Z_{\hat{K}} \in Z_K$ with a specific splitting distribution are involved.  (More specifically, $T_{0,2,2}$'s and $T_{2,0,2}$'s must split according to $\splresB$ while $T_{1,1,2}$'s must split according to $\splresA$.) All other level-1 blocks are wasted (which we call combination loss) in both the original version and the variant. This motivates our improvement.

\subsection{Our Improved Algorithm} \label{sec:2nd-improved-algo}

Finally, we are ready to present the improved algorithm for the second power. The algorithm follows almost the same steps as \Cref{sec:modi} with the following differences.
\begin{enumerate}
\item Symmetric hashing in Step 3 is replaced with asymmetric hashing. Hence each level-2 block $Z_K$ is now in multiple remaining triples after Step 3. This is the crucial step that compensates for the ``combination loss''. 
\item Additional Zeroing-Out Step 1 is an adaptation of the additional zeroing-out step in \Cref{sec:modi}. The difference is that now $Z_K$ might be in multiple remaining triples. Fix one of them, say $(X_I, Y_J, Z_K)$. We cannot simply zero out all level-1 blocks $Z_{\hat{K}} \in Z_K$ that are incompatible with this triple, because there might be another remaining triple $(X_{I'}, Y_{J'}, Z_K)$ that $Z_{\hat{K}}$ is compatible with.
\item There is one more step which we call the Additional Zeroing-Out Step 2. In this step, we zero out all the level-1 blocks in $Z_K$ that are compatible with more than one remaining triples. This guarantees the independence of these triples. 
\item After the Additional Zeroing-Out Step 2, there will be holes in Z-variables. We have to fix them in Step 4 with the random shuffling technique similar to \cite{doi:10.1137/1.9781611975482.31}.
\end{enumerate}

In this section, we will be analyzing the tensor $\T = (\CW_q^{\otimes 2})^{\otimes n}$.

\paragraph*{Step 1: Lower bound the values of subtensors.} Similar as \Cref{sec:modi}, we use the lower bounds given by \Cref{lem:non-rot-values}. Note that for $T_{1,1,2}$ we only have a bound on its symmetrized value, so in Step 2, we must have $\alpha(1,1,2) = \alpha(1,2,1) = \alpha(2,1,1)$.

\paragraph{Step 2: Choose a distribution.} We specify the component distribution $\alpha$ by
\begin{align*}
  &\alpha(0,2,2) = \alpha(2,0,2) = a, \\
  &\alpha(2,2,0) = b, \\
  &\alpha(1,1,2) = \alpha(1,2,1) = \alpha(2,1,1) = c, \\
  &\alpha(0,0,4) = \alpha(0,4,0) = \alpha(4,0,0) = d, \\
  &\alpha(0,1,3) = \alpha(0,3,1) = \alpha(1,0,3) = \alpha(1,3,0) = \alpha(3,0,1) = \alpha(3,1,0) = e,
\end{align*}
where $2a + b + 3c + 3d + 6e = 1$. Although for $(i,j,k)$'s other than $(1,1,2)$, the distribution $\alpha$ does not necessarily have to be symmetric, we still make some of them symmetric just to reduce the number of our parameters.

Note the joint distribution here may not be the maximum entropy distribution that has the same marginals, i.e. $\alpha \neq D^*(\alphx, \alphy, \alphz)$. This may incur the hash loss mentioned in \cref{sec:hashing} and \cref{remark:hash_loss1}. We will take such loss into account in the analysis (specifically, in \Cref{remark:hash-loss-analysis}).

The marginal distributions can be calculated accordingly:
\[
  \begin{array}{ccccc}
    \alphx(0) = a + 2d + 2e, & \alphx(1) = 2c + 2e, & \alphx(2) = a + b + c, & \alphx(3) = 2e, & \alphx(4) = d, \\
    \alphz(0) = b + 2d + 2e, & \alphz(1) = 2c + 2e, & \alphz(2) = 2a + c, & \alphz(3) = 2e, & \alphz(4) = d.
  \end{array}
\]
By symmetry between X and Y, we always have $\alphy = \alphx$.

Denote by $\numxblock, \numyblock, \numzblock$ the number of X, Y, and Z-blocks in $\mathcal T$ consistent with these marginal distributions, given by
\[
  \numxblock = \numyblock = \binom{n}{[n\alphx(i)]_{0\le i \le 4}}, \qquad \numzblock = \binom{n}{[n\alphz(k)]_{0 \le k \le 4}}.
\]
We require that $\numxblock = \numyblock > \numzblock$.

\paragraph{Step 3: Asymmetric Hashing.} We apply the asymmetric hashing method (\cref{sec:hashing}) to $(\CW_q^{\otimes 2})^{\otimes n}$.
We require that each remaining level-2 X or Y-block appears in a unique remaining triple, while each Z-block is typically in multiple remaining triples.

As explained in Step 2, the distribution $\alpha$ may not equal to $D^*(\alphx, \alphy, \alphz)$. In this case, asymmetric hashing can still be applied with a proper modulus $M$ at the cost of introducing an extra hash loss. (We will explicitly consider it in \Cref{remark:hash-loss-analysis}.)

After hashing, the subtensor of $(\CW_q^{\otimes 2})^{\otimes n}$ over a remaining triple $(X_I, Y_J, Z_K)$ is isomorphic to
\[
  \T^{\alpha} = \bigotimes_{i+j+k=4} T_{i,j,k}^{\otimes n\alpha(i,j,k)}.
\]

\paragraph{Additional Zeroing-Out Step 1.} Fix any level-2 triple $(X_I, Y_J, Z_K)$ and level-1 block $Z_{\hat K} \in Z_K$. Similar to the additional zeroing-out step in \Cref{sec:modi}, this step aims to ensure that, if $Z_{\hat K}$ is not compatible with $(X_I, Y_J, Z_K)$, there will be no terms between $X_I, Y_J$ and $Z_{\hat K}$. In other words, for all $X_{\hat I} \in X_I$, $Y_{\hat J} \in Y_J$ such that $\hat{I} + \hat{J} + \hat{K} = (2, 2, \ldots, 2)$, at least one of the level-1 blocks $X_{\hat I}$, $Y_{\hat J}$ and $Z_{\hat K}$ has to be zeroed out.

In \Cref{sec:modi}, since each $Z_K$ was in a unique triple, we just zeroed out all $Z_{\hat{K}}$'s that are not compatible with $X_I$ and $Y_J$. Now $Z_K$ might be in multiple triples due to asymmetric hashing. Even when $Z_{\hat{K}}$ is not compatible with $X_I$ and $Y_J$, we still cannot zero it out because it might be compatible with some other $X_{I'}$ and $Y_{J'}$. 

The fix is to zero out $X_{\hat{I}}$ or $Y_{\hat{J}}$ instead. For any level-2 block $X_I$ (or $Y_J$), there is a unique remaining triple $(X_I, Y_J, Z_K)$ containing it. Recall that we have defined $S_{i,j,k} = \BK{t \in [n] \mid (I_t, J_t, K_t) = (i,j,k)}$ and $S_2 = S_{0,2,2} \cup S_{2,0,2} \cup S_{1,1,2}$.

Suppose $X_{\hat{I}} \in X_I$ and $Y_{\hat{J}} \in Y_J$ satisfy $\hat{I} + \hat{J} + \hat{K} = (2, 2, \ldots, 2)$. When $J_t = 0$, since $(\hat{I}_{2t - 1}, \hat{I}_{2t}) + (\hat{J}_{2t - 1}, \hat{J}_{2t}) + (\hat{K}_{2t - 1}, \hat{K}_{2t}) = (2, 2)$, we must have $(\hat{I}_{2t - 1}, \hat{I}_{2t}) = (2 - \hat{K}_{2t - 1}, 2 - \hat{K}_{2t})$. This implies that $\split(\hat I, S_{2,0,2})(i') = \split(\hat K, S_{2,0,2})(2 - i')$. Let $\splresB^{(\textup{rev})}$ be the marginal split distribution defined by
\[
  \splresB^{(\revtext)}(i') := \splresB(2 - i').
\]
Suppose $Z_{\hat{K}}$ is not compatible with $X_I, Y_J$ (violating \Cref{def:complv2}) due to $\split(\hat K, S_{2,0,2}) \neq \splresB$. Then by zeroing out all $X_{\hat{I}} \in X_I$ where $\split(\hat I, S_{2,0,2}) \neq \splresB^{(\revtext)}$, we can make sure that there is no term between $X_I, Y_J$ and $Z_{\hat{K}}$.

Similarly, if $Z_{\hat{K}}$ is not compatible with $X_I, Y_J$ due to $\split(\hat K, S_{0,2,2}) \neq \splresB$, we zero out all $Y_{\hat{J}} \in Y_J$ with $\split(\hat{K}, S_{0,2,2}) \neq \splresB^{(\revtext)}$. Finally, if $Z_{\hat{K}}$ is incompatible with $X_I, Y_J$ due to $\split(\hat K, S_{1,1,2}) \neq \splresA$ while $\split(\hat K, S_{0,2,2}) = \split(\hat K, S_{2,0,2}) = \splresB$, we must have $\split(\hat K, S_2) \neq \splresAvg$. In this case, $Z_{\hat{K}}$ cannot be compatible with any $X_I, Y_J$, and we simply zero out $Z_{\hat{K}}$. \bigskip

Formally, in this step, we do the following:

\begin{enumerate}
\item For all level-2 block $X_I$ and level-1 block $X_{\hat{I}} \in X_I$, define $S_{i,j,k} = \BK{t \in [n] \mid (I_t, J_t, K_t) = (i,j,k)}$ with respect to the unique remaining triple $(X_I, Y_J, Z_K)$ that $X_I$ is in. We zero out $X_{\hat{I}}$ if and only if $\split(\hat{I}, S_{2,0,2}) \neq \splresB^{(\revtext)}$.
\item For all level-2 block $Y_J$ and all level-1 block $Y_{\hat{J}} \in Y_J$, define $S_{i,j,k}$ with respect to the unique remaining triple $(X_I, Y_J, Z_K)$ that $Y_J$ is in. We zero out $Y_{\hat{J}}$ if and only if $\split(\hat{J}, S_{0,2,2}) \neq \splresB^{(\revtext)}$.
\item For all level-2 block $Z_K$ and level-1 block $Z_{\hat{K}} \in Z_K$, let $S_2 = \{t \in [n] \mid K_t = 2\}$. We zero out $Z_{\hat{K}}$ if and only if $\split(\hat{K}, S_2) \neq \splresAvg$.
\end{enumerate}

From our discussion above, we can conclude the following lemma. Its formal proof is deferred to \Cref{appendix:missing_proofs}. 

\begin{restatable}{lemma}{AdditionalZ}
  Let $(X_I, Y_J, Z_K)$ be a remaining level-2 triple. For all $X_{\hat{I}} \in X_I, Y_{\hat{J}} \in Y_J, Z_{\hat{K}} \in Z_K$ such that $\hat{I} + \hat{J} + \hat{K} = (2, 2, \ldots, 2)$, if $Z_{\hat{K}}$ is not compatible with $X_I, Y_J$, at least one of $X_{\hat{I}}, Y_{\hat{J}}, Z_{\hat{K}}$ is zeroed out in Additional Zeroing-Out Step 1.
  \label{lem:prop1}
\end{restatable}

Fix a remaining triple $(X_I, Y_J, Z_K)$. In the modified CW algorithm described in \Cref{sec:modi}, after additional zeroing out, the remaining subtensor in $\T\vert_{X_I,Y_J,Z_K}$ is isomorphic to $\T^*$:
\[\T^* \defeq \Big( \bigotimes_{\substack{i+j+k=4 \\ k \ne 2}} T_{i,j,k}^{\otimes n \alpha(i,j,k)} \Big) \otimes T_{0,2,2}^{\otimes n \alpha(0,2,2)}[\splresB] \otimes T_{2,0,2}^{\otimes n \alpha(2,0,2)}[\splresB] \otimes T_{1,1,2}^{\otimes n \alpha(1,1,2)}[\splresA].\]
The following lemma says that this is also the case here.

\begin{lemma}
  \label{lem:prop_2}
  Let $(X_I, Y_J, Z_K)$ be a remaining level-2 triple. Let $\T^{(1)}$ be the tensor obtained after Additional Zeroing-Out Step 1. We have
  \[
    \T^{(1)}\vert_{X_I,Y_J,Z_K} \cong \T^*.
  \]
\end{lemma}

\begin{proof}
Fix any level-1 triple $(X_{\hat I}, Y_{\hat J}, Z_{\hat K})$. In $\T^{(1)}\vert_{X_I,Y_J,Z_K}$, it survives if and only if (1) $\split(\hat K, S_2) = \splresAvg$; and (2) $\split(\hat I, S_{2,0,2}) = \split(\hat J, S_{0,2,2}) = \splresB^{(\text{rev})}$. 

First of all, since $\hat{I} + \hat{J} + \hat{K} = (2,2,\dots,2)$, we know (2) is equivalent to $\split(\hat K, S_{0,2,2}) = \split(\hat K, S_{2,0,2}) = \splresB$. Moreover, given (2) holds, (1) is equivalent to
\[
\begin{aligned}\split(\hat K, S_{1,1,2}) &= \frac{1}{\alpha(1,1,2)} \left(\alpha(2) \cdot \splresAvg - \alpha(2,0,2) \cdot \split(\hat K, S_{2,0,2}) - \alpha(0,2,2) \cdot \split(\hat K, S_{0,2,2})\right) \\
&=  \frac{1}{\alpha(1,1,2)} \left(\alpha(2) \cdot \splresAvg - 2\alpha(0,2,2) \cdot  \splresB \right) \\
&= \splresA. \qquad\qquad \text{(By \cref{eq:avg_split_correct_sec4})}
\end{aligned}
\]
Thus (1) and (2) together are equivalent to $\split(\hat{K}, S_{1,1,2}) = \splresA$ and $\split(\hat K, S_{0,2,2}) = \split(\hat K, S_{2,0,2}) = \splresB$. The survived level-1 triples $(X_{\hat{I}}, Y_{\hat{J}}, Z_{\hat{K}})$ are exactly those in $\T^*$. 
\end{proof}

\paragraph{Additional Zeroing-Out Step 2.} In this step, we zero out all level-1 Z-blocks $Z_{\hat K}$ that are compatible with more than one remaining level-2 triples. Namely, we zero out $Z_{\hat{K}} \in Z_K$ if and only if $Z_{\hat{K}}$ is compatible with both $X_I, Y_J$ and $X_{I'}, Y_{J'}$ for some $I \neq I', J \neq J'$, while $(X_I, Y_J, Z_K)$ and $(X_{I'}, Y_{J'}, Z_K)$ both remained in hashing. We call the obtained tensor $\T^{(2)}$.

\begin{figure}[ht]
  \centering
  \includegraphics[width = 0.95 \textwidth]{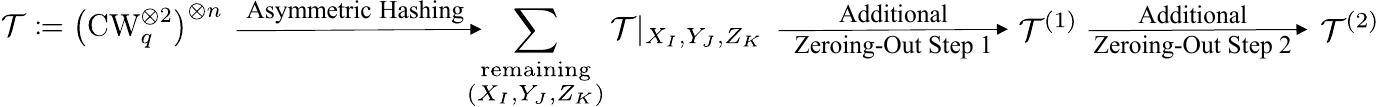}
  \caption{Tensor Degeneration Process}
  \label{fig:tensor6}
\end{figure}

Now, for a remaining triple $(X_I, Y_J, Z_K)$, the subtensor $\T^{(2)}\vert_{X_I,Y_J,Z_K}$ may not be isomorphic to $\T^*$, (in contrast with $\T^{(1)}\vert_{X_I,Y_J,Z_K}$ and \Cref{lem:prop_2}). But it is almost $\T^*$ except that some level-1 Z-blocks are zeroed out. We call these level-1 Z-blocks ``holes''. In \Cref{sec:analysis}, we will show that the fraction of holes $p_{\textup{hole}} < 1/2$. \footnote{Strictly, $p_{\text{hole}}$ is defined as the probability of each level-1 Z-block being a hole. The formal definition will be presented later.} Using this fact, we can fix the holes in Step 4.

\paragraph{Step 4: Degenerate Each Triple Independently and Fix Holes.} If there are no such holes, we can simply follow Step 4 in \Cref{sec:modi}. We will use the same values as \Cref{sec:modi}:
\begin{itemize}
\item For $T_{0,2,2}^{\otimes n \alpha(0,2,2)}[\splresB]$, we use its \emph{non-rotational restricted-splitting value} $\valone(T_{0,2,2}, \splresB)$. Similar for $T_{2,0,2}^{\otimes n \alpha(2,0,2)}[\splresB]$.
\item Recall that $\alpha(1,1,2) = \alpha(1,2,1) = \alpha(2,1,1) = c$. For $\sym_3(T_{1,1,2}^{\otimes n \alpha(1,1,2)})$, we use its \emph{restricted-splitting value} $\valthree(T_{1,1,2}, \splresA)$.
\item For other components $T_{i,j,k}$'s, we use their \emph{non-rotational value} $\valone(T_{i,j,k})$.
\end{itemize}

If we ignore the holes, we can directly use these values to get the bound
\begin{align}
  V_{\tau}^{\nrot}(\T^*) \geq &\left({\prod_{\substack{i + j + k = 4 \\ 0 \in \BK{i,j,k} \ \land \ k \neq 2}}} V_{\tau}^{\nrot}(T_{i,j,k})^{n\alpha(i,j,k) }\right) \cdot V^{(3)}_{\tau}(T_{1,1,2}, \splresA)^{n\alpha(1,1,2)} \notag \\ &\cdot V_{\tau}^{\nrot}(T_{0,2,2}, \splresB)^{n\alpha(0,2,2)} \cdot V_{\tau}^{\nrot}(T_{2,0,2}, \splresB)^{n\alpha(2,0,2)} \eqdef v^*. \label{eq:value_LB}
\end{align}

Any lower bound $V_{\tau}^{\nrot}(\T^*) \geq v$, by \Cref{def:non-rot}, gives a degeneration $\T^*\ \degen \ \bigoplus_{i=1}^s \angbk{a_i,b_i,c_i}$ with $\sum_{i=1}^s (a_i b_i c_i)^\tau \geq v$. \footnote{Strictly speaking, it gives a degeneration $(\T^*)^{\otimes m} \ \degen \ \bigoplus_{i=1}^s \angbk{a_i,b_i,c_i}$ with $\big(\sum_{i=1}^s (a_i b_i c_i)^\tau\big)^{1/m} \geq v$. It is easy to see that as $n\rightarrow \infty$, we can without loss of generality let $m = 1$ for $\T^*$.} To fix the holes, it is not enough to only use the lower bound \eqref{eq:value_LB}. We need to open the black box and use the following lemma about the corresponding degeneration. We also defer its proof to \Cref{appendix:missing_proofs}.

\begin{restatable}{lemma}{EqualSize}
  \label{lem:equal-size}
  The degeneration given by the lower bound \eqref{eq:value_LB} produces matrix multiplication tensors of the same size, i.e., $\T^* \ \degen \ \bigoplus_{i=1}^s \angbk{\bar{N},\bar{M},\bar{P}}$.
  Moreover, the degeneration is simply a zeroing out.
\end{restatable}

Now we start to fix the holes. We call a matrix multiplication tensor with holes \emph{broken}. We will fix them with the following lemma. The idea is straightforward: Suppose half of the Z-variables in a matrix multiplication tensor are holes. Then this broken tensor gives us the correct answer to half of the entries in the result matrix. If we can randomly permute the holes and repeat multiple times, then with high probability, we will get all the entries correct. The proof of \Cref{lemma:mat-hole-lm} will be given in \Cref{sec:mat-hole-lm}.

\begin{restatable}[Matrix Hole Lemma]{lemma}{MatrixHoleLemma}
  \label{lemma:mat-hole-lm}
  Let $s$ be an integer. For each $i \in [s]$, $T'_i$ is a broken copy of the matrix multiplication tensor $\angbk{\widebar N, \widebar M, \widebar P}$ in which $(1 - \eta_i)$ fraction of Z-variables are holes. If $\sum_{i=1}^s \eta_i \ge \log(\widebar N\widebar P) + 1$, then
  \[
    \bigoplus_{i=1}^s T'_i \;\degen\; \angbk{\widebar N,\widebar M,\widebar P}.
  \]
  That is, from the direct sum of broken matrix multiplication tensors, we can degenerate to an unbroken one.
\end{restatable}

In this step, we first observe that each remaining triple gives an independent broken copy of $\T^*$. For triple $(X_I, Y_J, Z_K)$ and level-1 block $Z_{\widehat{K}} \in Z_K$, we formally define $p_{\text{hole}}(\widehat{K}, I, J, K)$ as the probability that $Z_{\widehat{K}}$ is a hole, conditioned on $(X_I, Y_J, Z_K)$ is retained in the hashing step, and provided that $Z_{\hat K} \in Z_K$ is compatible with $(X_I, Y_J, Z_K)$. In \Cref{sec:analysis}, we will show that $p_{\text{hole}}(\widehat{K}, I, J, K) < 1 / 2$ for any fixed $\left(X_I, Y_J, Z_K\right)$ and $Z_{\widehat{K}}$.

We then apply \Cref{lem:equal-size} to degenerate these broken copies of $\T^*$ into broken matrix multiplication tensors of the same size. Since the degeneration is simply zeroing out, each entry in the broken matrix multiplication tensor is mapped from a single variable in $\T^*$. Since $p_{\text {hole}}(\widehat{K}, I, J, K) < 1/2$, each entry in the Z-matrices is a hole with probability less than a half. There are no X or Y holes.

Now, we have degenerated $\T = (\CW_q^{\otimes 2})^{\otimes n}$ to a direct sum of $m$ broken matrix multiplication tensors $T'_1 \oplus T'_2 \oplus \cdots \oplus T'_m$. Each $T'_i$ is a broken copy of $\angbk{\widebar N,\widebar M,\widebar P}$ in which some Z-variables are marked as holes. Let $\eta_i$ be the fraction of non-hole variables over all Z-variables inside the $i$-th matrix multiplication tensor. By our previous arguments, we know $\E[\eta_i] \ge 1/2$. We can then apply \Cref{lemma:mat-hole-lm} to fix the holes.

Formally, we divide $T'_1, \ldots, T'_m$ into groups in which the sum of $\eta_i$'s satisfies $\log(\widebar N\widebar P) + 1 \le \sum \eta_i \le \log(\widebar N\widebar P) + 2$. We independently apply \Cref{lemma:mat-hole-lm} to each group of tensors, and each group will degenerate to a complete $\angbk{\widebar N, \widebar M, \widebar P}$.

Note $\widebar N\widebar P$ is less than the number of Z-variables in $\T$, which is bounded by $(q+2)^{2n}$. So $\log(\widebar N\widebar P) \le 2n \log (q+2) = O(n)$. Through this procedure, we fix all the holes by losing a factor of only $O(n)$ on the number of matrix multiplication tensors. This gives a lower bound of the value $\valone((\CW_q^{\otimes 2})^{\otimes n})$, which can be used to further obtain an upper bound of $\omega$. A more concrete calculation and the numerical result will be presented in the next subsection.

\begin{remark} \label{remark:diffculty}
In this analysis, we require that $\alpha(1,1,2) = \alpha(1,2,1) = \alpha(2,1,1)$ because we do not have a lower bound for $T_{1,1,2}$'s non-symmetric value in \Cref{lem:non-rot-values}. We also required $\alpha$ to be symmetric for $T_{0,0,4}, T_{0,1,3}$ just to minimize the number of parameters. 

We would like to point out this is not optimal. Breaking the symmetry of $T_{0,0,4}, T_{0,1,3}$ could help improve the bound. Another natural attempt would be to further break the symmetry for $T_{1,1,2}$: The challenge is that we do not have a non-rotational value for $T_{1,1,2}$. One attempt is to do another symmetrization of $\T^{(2)}$ and consider $\sym_3\left(\T^{(2)}\right) = \T^{(2)} \otimes \left(\T^{(2)}\right)^{\rot} \otimes \left(\T^{(2)}\right)^{\rot \; \rot}$. But as $\T^{(2)}$ has holes in its Z-variables, the symmetrization $\sym_3\left(\T^{(2)}\right)$ would have holes not only in Z-variables but also in X and Y-variables, then one can no longer apply \Cref{lemma:mat-hole-lm} to fix them. This is why we consider non-rotational values in this section to avoid such symmetrization. To solve this issue, we will introduce a more general Hole Lemma in \cref{sec:hole_lemma}, which can fix the holes in $\T^{(2)}$ before we symmetrize it. After that, we may break the symmetry for $T_{1,1,2}$, $T_{1,2,1}$ and $T_{2,1,1}$, obtaining better bounds than the current section. (See Section~\ref{sec:level-2-global}.)
\end{remark}

\subsection{Analysis}
\label{sec:analysis}

In the previous subsection, we described our improved algorithm for degenerating $\T \defeq (\CW_q^{\otimes 2})^{\otimes n}$ to independent matrix multiplication tensors. This subsection adds some calculations over the parameters and completes the analysis to obtain a bound of $\omega$.

\paragraph{Asymmetric Hashing.} First, let us recall our notations. 

\begin{itemize}
\item $\numxblock, \numyblock, \numzblock$ denote the number of X, Y, and Z-blocks consistent with $\alphx, \alphy, \alphz$, respectively.
\item $\numalpha$ is the number of triples $(X_I, Y_J, Z_K)$ whose joint distribution is consistent with $\alpha$, while $\numtriple$ is the number of triples $(X_I, Y_J, Z_K)$ whose marginal distributions are consistent with $\alphax, \alphay, \alphaz$.
\item $p_{\textup{comp}}$ is the probability that a uniformly random triple obeying $\alpha$ and containing $Z_K$ is consistent with a fixed level-1 block $Z_{\hat{K}} \in Z_K$.
\item We will use $\numretain$ to denote the number of retained triples with joint distribution $\alpha$ that we get after asymmetric hashing.
\end{itemize}

\begin{remark} \label{remark:hash-loss-analysis}
  Since $\alpha$ may not be the maximum entropy joint distribution over all joint distributions consistent with marginals $\alphx, \alphy, \alphz$ (i.e. $\alpha \ne D^*(\alphx, \alphy, \alphz)$), we might have $\numtriple > \numalpha$, which incurs the hash loss (see \cref{sec:hashing} for details). As a result, after hashing and zeroing out, only $\numxblock \cdot \frac{\numalpha}{\numtriple}$ triples consistent with $\alpha$ are retained.
\end{remark}

By \Cref{lem:p_comp_sec4}, since here $\alpha(0,2,2) = \alpha(2,0,2) = a, \alpha(1,1,2) = c$, we have
\[
  p_{\textup{comp}} =\binom{an}{[an \cdot \splresB(k')]_{k'=0,1,2}}^2 \binom{cn}{[cn \cdot  \splresA(k')]_{k'=0,1,2}} \bigg/ \binom{(2a + c)n}{[(2a + c)n \cdot \splresAvg(k')]_{k'=0,1,2}},
\]
(Note this is the same equation as \Cref{lem:p_comp_sec4} but with a different set of parameters for $\alpha$.) The parameters we select will satisfy the additional assumption
\[\numxblock \le \numzblock / p_\text{comp} \numberthis \label{eq:redundant_assumption},\]
so that after we apply asymmetric hashing method, each level-2 Z-block $Z_K$ will be matched to at most $\numxblock / \numzblock \leq 1 / p_{\text{comp}}$ many level-2 X/Y blocks on average. Then by the definition of $p_{\text{comp}}$, each level-1 block $Z_{\hat{K}} \in Z_K$ will in expectation be consistent with at most one of them.

We will set the hash modulus $M$ to a prime in $[4 \numtriple / \numxblock, \, 8 \numtriple / \numxblock]$. As explained in \Cref{remark:hash-loss-analysis}, $\numtriple$ might be larger than $\numalpha$. Then we apply asymmetric hashing (\Cref{sec:hashing}), getting $\E[\numretain] \geq \numxblock \cdot \frac{\numalpha}{\numtriple} \cdot 2^{-o(n)}$ many triples (i.e., loosely speaking, for an X-block consistent with $\alphx$, with probability  $\frac{\numalpha}{\numtriple}$, it is in one of the retained triples).

\paragraph{Probability of Being a Hole.} Let $(X_I, Y_J, Z_K)$ be a triple retained in the hashing step, and $Z_{\hat K} \in Z_K$ be some level-1 block that is compatible with $(X_I, Y_J, Z_K)$. Now we analyze the probability for $Z_{\hat K}$ to be a hole, denoted by $p_{\textup{hole}}(\hat K, I, J, K)$.

A necessary condition for $Z_{\hat K}$ to be a hole is that there exist some other blocks $X_{I'}$ and $Y_{J'}$ $(I' \ne I)$ such that (1) $(X_{I'}, Y_{J'}, Z_K)$ forms a triple consistent with $\alpha$; (2) $\hashx(I') = \hashx(I) = \hashz(K)$, i.e., $I'$ has the same hash value as the triple $(X_I, Y_J, Z_K)$; and (3) $Z_{\hat K}$ is compatible with $(X_{I'}, Y_{J'}, Z_K)$. (Note that $J'$ is determined by $I'$.) We will count the expected number of such $I'$ using the following two facts:
\begin{itemize}
    \item From \Cref{lemma:hash_independence}, we know that for arbitrary two triples $(X_{I}, Y_J, Z_K)$ and $(X_{I'}, Y_{J'}, Z_K)$ consistent with $\alpha$, $\Pr[\hashx(I') = \hashz(K) \mid \hashx(I) = \hashz(K)] = 1/M$.
    \item The choice of $M$ and Assumption \eqref{eq:redundant_assumption} implies that
\[ \frac{p_\text{comp}}{M} \le \frac{p_\text{comp}}{4 \numtriple / \numxblock} \le \frac{\numzblock}{4\cdot \numtriple}. \]
\end{itemize}
Hence 
\begin{align*}
  p_{\textup{hole}}(\hat K, I, J, K)
  &\le \sum_{I' \ne I} \Pr[\hashx(I') = \hashx(I) \mid \hashx(I) = \hashz(K)] \cdot \ind[Z_{\hat K} \textup{ is compatible with } (X_{I'}, Y_{J'}, Z_K)] \\
  & = \sum_{I' \ne I} \ind[Z_{\hat K} \textup{ is compatible with } (X_{I'}, Y_{J'}, Z_K)] \;/\; M \\
  & < \sum_{I'} \ind[Z_{\hat K} \textup{ is compatible with } (X_{I'}, Y_{J'}, Z_K)] \;/\; M \\
  & = \frac{\numalpha}{\numzblock} \Pr_{I'}[Z_{\hat K} \textup{ is compatible with } (X_{I'}, Y_{J'}, Z_K)] \; / \; M \\
  & = \frac{p_{\textup{comp}} \cdot \numalpha}{M \cdot \numzblock} \;<\; \frac{\pcomp \cdot \numtriple}{M \cdot \numzblock} \;<\; 1/2,
\end{align*}
where $I'$ on the fourth line is chosen uniformly at random to let $(X_{I'}, Y_{J'}, Z_K)$ form a triple consistent with $\alpha$.

\paragraph{Bounding the Value.} Recall that $\numretain$ is the number of retained triples in asymmetric hashing and $\E[\numretain] \geq \numxblock \cdot \frac{\numalpha}{\numtriple} \cdot 2^{-o(n)}$.  Here each of them is a broken copy of $\T^*$. From \Cref{lem:equal-size} (whose proof is in \Cref{appendix:missing_proofs}), we know that for some integer $s$, $\T^*$ can be degenerated into $s$ many matrix multiplication tensors of the same size, let us say $\langle\widebar N, \widebar M, \widebar P\rangle$. Moreover, $s$ satisfies that 
\[
  \begin{aligned}
    (\bar N \bar M \bar P)^\tau \cdot s
    &\ge \valone(T_{2,2,0})^{n\alpha(2,2,0)} \cdot \valone(T_{0,2,2}, \splresB)^{2n\alpha(0,2,2)} \cdot \valthree(T_{1,1,2}, \splresA)^{3n\alpha(1,1,2)} \\
    &\phantom{{} \ge {}} \times
      \prod_{\substack{(i,j,k) \in \{(0,0,4), \\ (0,1,3), (0,3,1)\}}} \valone(T_{i,j,k})^{3n \alpha(i,j,k)}
    \cdot 2^{-o(n)} \;\ge\; \alphaval^{n - o(n)}, \\
    \textup{where}\qquad \alphaval &\defeq (q^2+2)^{\tau(2a+b)} \bk{2^{2/3} q^{\tau} (q^{3\tau} + 2)^{1/3}}^{3c} (2q)^{6\tau e}.
  \end{aligned}
\]
Here we follow the notation of \cite{alman2021} and use $\alphaval$ to denote the lower bound on $\displaystyle\lim_{n \rightarrow \infty}((\bar N \bar M \bar P)^\tau \cdot s)^{1/n}$, that is, the total volume of matrix multiplication tensors we finally get normalized by taking the $n$-th root.

From $\T$, we can get in total $m = s \numretain$ matrix multiplication tensors with holes, and further fix the holes to obtain $m'$ copies of $\angbk{\widebar N, \widebar M, \widebar P}$ without holes. Recall that for all $i \in [m]$, we let $\eta_i$ be the fraction of non-hole entries in the $i$-th matrix multiplication tensor we get. From \Cref{lemma:mat-hole-lm}, we know that 
\[
  \begin{aligned}
    m' &\ge \frac{\sum_{i=1}^m \eta_i}{O(n)} - 1 \\
    \E[m'] &\ge \frac{\E[m]}{2 \cdot O(n)} \\
       &= \E[\numretain] \cdot s \cdot 2^{-o(n)}.
  \end{aligned}
\]
Let $v \defeq m' (\widebar N \widebar M \widebar P)^\tau$ be the total volume of the matrix multiplication tensors we get. We have
\begin{align*}
\E[\numretain] &= \numxblock \cdot \frac{\numalpha}{\numtriple}\\
  \E[v] &\ge \E[m'] \cdot (\bar N \bar M \bar P)^\tau \\ &\ge \numxblock \cdot \frac{\numalpha}{\numtriple} \cdot s \cdot (\bar N \bar M \bar P)^\tau \cdot 2^{-o(n)}.
\end{align*}

\paragraph*{How to Verify Our Numerical Results.}
Before we list our parameters that lead to the bound of $\omega$, let us first see how to verify a given set of parameters. Recall that our improved algorithm degenerates the tensor $\T$ into independent matrix multiplication tensors $\angbk{\widebar N, \widebar M, \widebar P}^{\oplus m'}$. Each possible outcome of the algorithm $v$ can lead to a lower bound $\valone(\CW_q^{\otimes 2}) \ge v^{1/n}$, thus $\valone(\CW_q^{\otimes 2}) \ge \E[v]^{1/n}$ by probabilistic method. Following previous works (e.g. Section 4 of \cite{alman2021}), we will use the following set of notations:
\begin{itemize}
\item $\alphabx \defeq \displaystyle \lim_{n\rightarrow\infty} \numxblock^{1/n}$, i.e. the number of level-2 X/Y-blocks normalized by taking $n$-th root.
\item Similarly, $\alphabz \defeq \displaystyle \lim_{n\rightarrow\infty} \numzblock^{1/n}$ is that of the Z-blocks.
\item $\displaystyle \alphant \defeq \lim_{n\rightarrow\infty} \numalpha^{1/n}$ is the number of triples that have joint distribution $\alpha$. 
\item $\displaystyle \alphap \defeq \lim_{n \rightarrow \infty} \pcomp^{1/n}$ is the (normalized) probability that a uniformly random triple obeying $\alpha$ and containing $Z_K$ is consistent with a fixed level-1 block $Z_{\hat{K}} \in Z_K$.  
\item In such notations, $\displaystyle \max_{\alpha' \in \distShareMargin} \alphant' = \lim_{n\rightarrow \infty} \numtriple^{1/n}$ as proved in \cref{lem:numtriple_singledist}.
\end{itemize}
Hence,
\begin{align*}
  & \valone(\CW_q^{\otimes 2})
  \;\ge\; \lim_{n\to\infty}\E[v]^{1/n} \\
  \ge{} & \lim_{n\to\infty} \numxblock^{1/n} \cdot \lim_{n\to\infty} \bk{\frac{\numalpha}{\numtriple}}^{1/n} \cdot \alphaval \\
  \ge{} & \frac{\alphabx \alphant \alphaval}{\max_{\alpha' \in \distShareMargin} \alphant'}. \numberthis \label{eq:vbound-s4}
\end{align*}
Assumption \eqref{eq:redundant_assumption} can be verified by checking
\[
  \lim_{n \to \infty} \numxblock^{1/n} = \alphabx \le \lim_{n \to \infty} \frac{\;\numzblock^{1/n}\;}{p_{\text{comp}}^{1/n}} = \frac{\alphabz}{\alphap}. \numberthis \label{eq:redundant_variant}
\]
By definition, we write the closed forms
\begin{align*}
  \alphabx &= 2^{H(\alphax)} = \prod_{i = 0}^{4} \alphax(i)^{-\alphax(i)}, \\
  \alphabz &= 2^{H(\alphaz)} = \prod_{k = 0}^{4} \alphaz(k)^{-\alphaz(k)}, \\
  \alphant &= 2^{H(\alpha)} = \prod_{i + j + k = 4} \alpha(i, j, k)^{-\alpha(i, j, k)}, \\
  \alphap &= \lim_{n \to \infty} \pcomp^{1/n} = 2^{2a H(\splresB) + c H(\splresA) - (2a + c) H(\splresAvg)},
\end{align*}
where $\splresA, \splresB$ are defined in \Cref{lem:non-rot-values} and are independent of $\alpha$; $\splresAvg$ depends on $a = \alpha(0, 2, 2)$ and $c = \alpha(1, 1, 2)$. Specifically, $\splresAvg = \frac{c}{c+2a} \splresA + \frac{2a}{c + 2a} \splresB$ as defined in \eqref{eq:avg_split_correct_sec4}.

Lastly, $\max_{\alpha' \in \distShareMargin} \alphant[\alpha']$ can be computed via a convex optimization. So far, assume the distribution $\alpha$ is given in advance, the following process (\cref{alg:verify_sec_power}) can compute the corresponding lower bound of $\valone(\CW_q^{\otimes 2})$ via our improved method.

\begin{figure}[h]
  \begin{center}
    \begin{tcolorbox}
      \captionof{algocf}{Verifying the Lower Bound}{\label{alg:verify_sec_power}}
      \vspace{-0.5em}
      Assume constants $a, b, c, d, e$ are given and $\alpha$ is defined by
      \begin{align*}
        &\alpha(0,2,2) = \alpha(2,0,2) = a, \\
        &\alpha(2,2,0) = b, \\
        &\alpha(1,1,2) = \alpha(1,2,1) = \alpha(2,1,1) = c, \\
        &\alpha(0,0,4) = \alpha(0,4,0) = \alpha(4,0,0) = d, \\
        &\alpha(0,1,3) = \alpha(0,3,1) = \alpha(1,0,3) = \alpha(1,3,0) = \alpha(3,0,1) = \alpha(3,1,0) = e.
      \end{align*}
      \begin{enumerate}
      \item Compute $\alphabx, \alphabz, \alphant, \alphap$ and verify that $\alphabx \le \alphabz / \alphap$.
      \item\label{step:conv_prog_entr} Solve the following convex optimization problem:
        \[
          \begin{array}{cc}
            \textup{maximize} & \alphant[\alpha'] = 2^{H(\alpha')} \\
            \textup{subject to} & \alpha' \in \distShareMargin. \\
          \end{array}
        \]
        Then we get $\max_{\alpha' \in \distShareMargin} \alphant[\alpha']$.
      \item Calculate the lower bound by \eqref{eq:vbound-s4}.
      \end{enumerate}
    \end{tcolorbox}
  \end{center}
\end{figure}

Note that in Step~\ref{step:conv_prog_entr}, the objective $\alphant[\alpha']$ is log-concave which allows efficient solvers that guarantee optimality. Although we have limitations on $\alpha$ so that it is determined by 5 variables $a, b, c, d, e$ (4 of which are free variables), $\alpha'$ has a much larger degree of freedom. For example, we require $\alpha(1, 1, 2) = \alpha(2, 1, 1) = \alpha(1, 2, 1)$ since we need to use the 3-rotational value of $T_{1,1,2}$, but $\alpha'$ is allowed to violate such symmetry. Hence, $\max_{\alpha' \in \distShareMargin} \alphant[\alpha']$ is usually strictly larger than $\alphant[\alpha]$.

\paragraph{Finding Good Parameters.} To find these parameters $a, b, c, d, e$, we need to optimize the following program:
\begin{align*}
  \begin{array}{cl}
    \text{maximize} & \displaystyle\frac{\alphabx \alphant \alphaval}{\max_{\alpha' \in \distShareMargin} \alphant'} \\
    \text{subject to} & \alphabx \le \alphabz / \alphap \\
                    & \alpha \textup{ is defined by parameters $a, b, c, d, e$ like above}.
  \end{array}
\end{align*}
However, this optimization problem is non-convex and has a complicated form. Instead of solving it perfectly, we will use heuristics described in \Cref{appendix:heuristics} to get a feasible (but not necessarily optimal) bound. 

\begin{remark}
  Although we use more complicated optimization techniques than prior work, it is clear that our improvement of $\omega$ comes from the new theoretical ideas instead of better calculation -- when applying the prior approach on the second power of the CW tensor, the optimal parameters are easy to derive and prove optimality (see \cite{coppersmith1987matrix}). I.e., better calculation can only improve the bound on higher powers, but not the second power. The best known bound by analyzing the second power remains unchanged since \cite{coppersmith1987matrix}.
\end{remark}

\paragraph{Numerical Results.}
We use \eqref{eq:vbound-s4} together with $\tauthm$ to obtain an upper bound of $\omega$.\footnote{The optimization and verification code for this section is available at \url{https://osf.io/dta6p/?view_only=cf30d3e1ca2f4fe5b4142f65d28b92fd}.} Set
\[
  a = 0.102787, \quad b = 0.102058, \quad c = 0.205540, \quad d = 0.000232, \quad e \approx 0.0125086667.
\]
Given the parameters, the bound can be verified via \cref{alg:verify_sec_power}. According to the definitions, one can calculate
\[
  \alphabx \approx 2.9595937152, \qquad \alphabz \approx 2.9570775659, \qquad 1/\alphap \approx 1.0008517216,
\]
hence \eqref{eq:redundant_variant} is satisfied. Also, by running \cref{alg:verify_sec_power} we can see that $\alphant / \max_{\alpha' \in \distShareMargin} \alphant[\alpha'] \approx 1 - 2.49 \times 10^{-7}$, which means the hash loss is very small (and thus the first heuristic in \cref{appendix:heuristics} is very accurate). The implied bound is $\omega < 2.375234$.

\subsection{Proof of Matrix Hole Lemma.}\label{sec:mat-hole-lm} The last ingredient of our analysis is the proof of \Cref{lemma:mat-hole-lm}. Guided by the explanation in the previous subsection, we formally state its proof. We first recall the lemma:

\MatrixHoleLemma*

\begin{proof}
  We prove by the probabilistic method.
  
  Let $\S_{U}$ denote the symmetric group over $U$. For each $T'_t$, we sample $\sigma_1^{(t)} \in \S_{[\widebar N]}$, $\sigma_2^{(t)} \in \S_{[\widebar M]}$, and $\sigma_3^{(t)} \in \S_{[\widebar P]}$ uniformly. Then, let $T''_t$ be the degeneration of $T'_t$ with the following mappings:
  \[
    x^{(t)}_{i,j} \mapsto \bar x^{(t)}_{\; \sigma^{\!(t)}_1\!(i), \; \sigma^{\!(t)}_2\!(j)}, \qquad
    y^{(t)}_{j,k} \mapsto \bar y^{(t)}_{\; \sigma^{\!(t)}_2\!(j), \; \sigma^{\!(t)}_3\!(k)}, \qquad
    z^{(t)}_{k,i} \mapsto \bar z^{(t)}_{\; \sigma^{\!(t)}_3\!(k), \; \sigma^{\!(t)}_1\!(i)}.
  \]
  Here $T'_t$ is a tensor over $\{x^{(t)}_{i,j}\}, \{y^{(t)}_{j,k}\}$, and $\{z^{(t)}_{k,i}\}$, for $i \in [\widebar N]$, $j \in [\widebar M]$, and $k \in [\widebar P]$; $T''_t$ is over $\{\bar x^{(t)}_{i,j}\}, \{\bar y^{(t)}_{j,k}\}$, and $\{\bar z^{(t)}_{k,i}\}$. Such degeneration is intuitively a ``renaming'' of the variables. The only effect of this degeneration is to shuffle the positions of the holes. One can observe that the preimage of some variable $\bar z^{(t)}_{k,i}$ is uniformly random over all Z-variables in $T'_t$, i.e., $\{z^{(t)}_{k',i'}\}_{k' \in [\widebar P],\, i' \in [\widebar N]}$. As a corollary,
  \[
    \Pr[\bar z^{(t)}_{k, i} \textup{ is a hole in } T''_t] = 1 - \eta_t.
  \]
  Moreover, events of this type are independent for different $t$'s when $k, i$ are fixed.

  If for some pair $(k, i) \in [\widebar P] \times [\widebar N]$, there exists some $t$ such that $\bar z^{(t)}_{k,i}$ is \emph{not} a hole in $T''_t$, we call $(k, i)$ a \emph{good position}; otherwise we call it a \emph{bad position}. The probability of $(k, i)$ being a bad position is
  \[
    \begin{aligned}
      \Pr[(k, i) \textup{ is bad}]
      &= \prod_{t=1}^{s} (1 - \eta_t) \le \prod_{t=1}^s \exp(-\eta_t) \\
      &\le \exp(-\log(\widebar N\widebar P)-1) < \exp(-\ln(\widebar N\widebar P)-1) = 1/(e\widebar N\widebar P).
    \end{aligned}
  \]
  From this formula, we know $\E[\textup{number of bad positions}] \le 1/e < 1$. By probabilistic method, we know there is a sequence $\{(\sigma_1^{(t)}, \sigma_2^{(t)}, \sigma_3^{(t)})\}_{t = 1, 2, \ldots, s}$ such that every position is good. We fix this sequence and continue with our construction.

  For each position $(k,i)$, we let $t_{k,i}$ be the first $t$ such that $\bar z^{(t)}_{k,i}$ is not a hole in $T''_t$. Then, we make the following degeneration from $\bigoplus_{t=1}^s T''_t$ to $\angbk{\widebar N, \widebar M, \widebar P}$:
  \[
    \begin{aligned}
      \bar x^{(t)}_{i,j} &\mapsto x^*_{i,j}, \qquad (\forall i, j, t) \\
      \bar y^{(t)}_{j,k} &\mapsto y^*_{j,k}, \qquad (\forall i, j, t) \\
      \bar z^{(t)}_{k,i} &\mapsto
                      \begin{cases}
                        z^*_{k,i}, & \text{if } t = t_{k,i}, \\
                        0, & \text{otherwise}.
                      \end{cases}
    \end{aligned}
  \]
  Here we denote by $\{x^*_{i,j}\}, \{y^*_{j,k}\}$, and $\{z^*_{k,i}\}$ the variable sets of the result tensor $\angbk{\widebar N, \widebar M, \widebar P}$, where $i \in [\widebar N], j \in [\widebar M]$, and $k \in [\widebar P]$. One can verify that the degeneration above correctly produces $\angbk{\widebar N, \widebar M, \widebar P}$. Notice that the X and Y-variables in different $T''_t$ are mapped to the same matrices through \emph{identification}; Z-variables are zeroed out properly to avoid duplicated terms. This degeneration, combined with the fact that $T'_t \degen T''_t$, concludes our proof.
\end{proof}

\paragraph{Discussion.} This section is only a minimum working example of compensating for the combination loss. It gives the bound $\omega < 2.375234$, which already improves upon the bound $\omega < 2.375477$ of Coppersmith-Winograd \cite{coppersmith1987matrix} for the second power. However, we can get a better bound for the second power via more careful modifications.

For example, here we require the distribution for $(1,1,2)$ to be symmetric to get around the difficulty mentioned in \Cref{remark:diffculty}. In \Cref{sec:hole_lemma}, we will extend \Cref{lemma:mat-hole-lm} to broken tensors like $\T^*$ (not only matrix multiplication tensors). With its help, we first fix holes in $\T^*$, then symmetrize the tensor, and finally degenerate it into matrix multiplication tensors. This would allow us to also break the symmetry for $(1,1,2)$. 

Another further improvement will be modifying the split restrictions $\splresA$ and $\splresB$. Recall that in this section, they are set to the distributions induced by the original parameters used in \cite{coppersmith1987matrix}. In later sections, we will modify these distributions $\splresA$ and $\splresB$. Although the restricted-splitting values will decrease, $p_{\textup{comp}}$ will be significantly smaller, so that each Z-block $Z_K$ can be shared by more triples. By setting proper parameters, we finally get the bound $\omega < 2.374631$ from the second power, whose parameters will be given in \cref{sec:level-2-global}. We will present the method for this result and its generalization to higher powers in the rest of the paper.

\section{Hole Lemma}
\label{sec:hole_lemma}
In this section, we will present a generalization of \Cref{lemma:mat-hole-lm}. The reason we need this generalization is mentioned in \Cref{remark:diffculty}, and now we will explain it in more detail: Suppose we have a tensor $\Thole$ that is the direct sum of many matrix multiplications of the same size but with holes in Z-variables (which is the case in \Cref{sec:2nd}), there is no problem in degenerating it into these matrix multiplications with holes and fixing them using \Cref{lemma:mat-hole-lm}. Now consider some other tensor $\Thole'$ that also has holes only in Z-variables. The complication arises when $\Thole'$ itself cannot be degenerated into enough matrix multiplications without symmetrizing it, but $\sym_3\left(\Thole'\right) = \Thole' \otimes \left(\Thole'\right)^{\rot} \otimes \left(\Thole'\right)^{\rot \; \rot}$ can. While $\Thole'$ only has holes in its Z-variables, the symmetrized tensor $\sym_3\left(\Thole'\right)$ will have holes in all its X, Y, Z-variables, so we can no longer apply \Cref{lemma:mat-hole-lm}. In \Cref{sec:2nd}, we avoided this complication using non-rotational values; but to generalize it to higher levels, it is easier to introduce a new hole lemma, and directly fix $\Thole'$ before symmetrizing it.

Instead of randomly permuting entries in a matrix, we will permute all lower-level blocks in a tensor. This only works for a specific type of tensor with a certain symmetry. We call them \emph{standard form tensors}. It will capture the tensor $\Thole'$ we need to handle in the laser method. Specifically, in \Cref{sec:2nd}, fixing any block triple $(X_I, Y_J, Z_K)$, we are handling $\Thole' = \T^{(2)}\vert_{X_I, Y_J, Z_K}$ which, except for the holes in its Z-variables, has the same structure as
\[\T^* = \Big( \bigotimes_{\substack{i+j+k=4 \\ k \ne 2}} T_{i,j,k}^{\otimes n \alpha(i,j,k)} \Big) \otimes T_{0,2,2}^{\otimes n \alpha(0,2,2)}[\splresB] \otimes T_{2,0,2}^{\otimes n \alpha(2,0,2)}[\splresB] \otimes T_{1,1,2}^{\otimes n \alpha(1,1,2)}[\splresA].\]

\begin{remark}
Notice that in such $\T^*$, we are only restricting the Z-split distribution of those $T_{i,j,k}$'s with $k = 2$. This was because there are two ways to split $2$ (up to reflection symmetry), i.e., $2 + 0$ (or $0 + 2$) and $1 + 1$. On the contrary, there is only one way to split $0$, $1$, and $3$, so their split distributions are trivial. This reason is specific to the second level. We might as well just restrict the split distribution of all $T_{i,j,k}$'s in $\T^*$, as what we will do in \Cref{def:standard_form_tensor}.
\end{remark}

\subsection{Standard Form Tensor}

Our definition of \emph{standard form tensor} captures this type of tensor like $\T^*$ above.

\begin{definition}[Standard form tensor]
  \label{def:standard_form_tensor}
  Let $(i_1, j_1, k_1), \, (i_2, j_2, k_2), \, \dots \, , \, (i_m, j_m, k_m)$ be a sequence of $m$ different level-$\lvl$ components, i.e., for all $t \in [m]$, $i_t + j_t + k_t = 2^\lvl$. Let $n_1, n_2, \dots, n_m$ be positive integers and $\tilde{\alpha}_1, \tilde{\alpha}_2, \dots, \tilde{\alpha}_m$ be Z-split distributions.

  The level-$\lvl$ \emph{standard form tensor} with parameters $\BK{n_t, i_t, j_t, k_t, \tilde\alpha_t}_{t \in [m]}$ is defined as
  \[
    \mathcal T^*
    \defeq \bigotimes_{t \in [m]} T_{i_t, j_t, k_t}^{\otimes n_t} [\tilde\alpha_t].
  \]
\end{definition}

As explained before, in \Cref{sec:2nd}, each $\Thole' = \T^{(2)}\vert_{X_I, Y_J, Z_K}$ is a copy of standard form tensor $\T^*$ with holes. These holes are lower-level blocks $Z_{\hat{K}}$ that have been zeroed out. Throughout this section, we will call the lower-level blocks \emph{small blocks}. In $\T^*$, each small Z-block is either zeroed out, forming a hole, or remains completely intact. Before we fix the holes in $\T^*$, we have to first define the small blocks formally.

\begin{definition}[Small blocks in $\T^*$]
Let $\T^*$ be a standard form tensor with parameters $\BK{n_t, i_t, j_t, k_t, \tilde\alpha_t}_{t \in [m]}.$ Set $N = \sum_{t=1}^m n_t$. A small Z-block in $\T^*$ is indexed by sequence $\hat K = (\hat{K}_1, \hat{K}_2, \dots, \hat{K}_{2N})$ such that the following holds:
$$\text{For all $1 \leq t \leq m$ and $s \in \left(\sum_{t'=1}^{t - 1} n_{t'}, \sum_{t'=1}^t n_{t'}\right]$, } \hat{K}_{2s - 1} + \hat{K}_{2s} = k_t.$$
The corresponding small Z-block in $\T^*$ is defined as the variable set $Z_{\hat{K}_1} \times Z_{\hat{K}_2} \times \dots \times Z_{\hat{K}_{2N}} \eqdef Z_{\hat K}$.
We similarly define small X-blocks and small Y-blocks. Throughout this section, whenever we mention ``blocks,'' we are refering to the \emph{small blocks}.
\end{definition}

All small X and Y-blocks defined in this way exist as variables in $\T^*$. However, some of the Z-blocks do not exist, because there is also restrictions of Z-split distributions $\tilde{\alpha}_t$. Not all small Z-blocks satisfy this restriction. This leads to the following definition.

\begin{definition}[Available small Z-blocks]
We say a small Z-block $Z_{\hat K}$ is available in $\T^*$ if for all $1 \leq t \leq m$ and $k' + k'' = k_t$, we have
\[ \tilde{\alpha}_t(k') = \frac{1}{n_t}\left|\left\{s \in \left(\sum_{t'=1}^{t - 1} n_{t'}, \sum_{t'=1}^t n_{t'}  \right] \ \middle \vert \ (\hat{K}_{2s - 1}, \hat{K}_{2s}) = (k', k'') \right\}\right|. \]
\end{definition}

All Z-variables in $\T^*$ is exactly the union of available small Z-blocks in $\T^*$. Now we define the standard form tensor with holes:

\begin{definition}[Broken standard form tensor]
  \label{def:broken_standard_form_tensor}
  A tensor $\mathcal T'$ is a \emph{broken copy of $\mathcal T^*$}, if it can be obtained from $\mathcal T^*$ by zeroing out several available small Z-blocks. 
  \begin{itemize}
  \item The zeroed-out small Z-blocks are called \emph{holes}, while other available small Z-blocks are called \emph{non-hole blocks}. 
  \item The proportion of holes among all available small Z-blocks in $\mathcal T^*$ is called the \emph{fraction of holes}; conversely, the proportion of non-hole blocks among all available Z-blocks is called the \emph{fraction of non-holes}.
  \end{itemize}
\end{definition}

\subsection{Fixing the Holes}

\begin{lemma}[Hole Lemma]
  \label{lemma:hole_lemma}
  Let $\T^*$ be a level-$\lvl$ standard form tensor with parameters $\BK{n_t, i_t, j_t, k_t, \tilde\alpha_t}_{t \in [m]}$ ($i_t + j_t + k_t = 2^\lvl$ for all $t$) and $N = \sum_{t=1}^m n_t$. Assume that we have $s$ broken copies of $\mathcal T^*$ named $\mathcal T'_1, \ldots, \mathcal T'_s$, with fractions of non-holes $\eta_1, \ldots, \eta_s$, satisfying $\sum_{t=1}^s \eta_t \ge N\lvl + 1$. Then there is a way of degenerating them into one complete copy of $\mathcal T^*$, i.e.,
  \[
    \bigoplus_{t=1}^s \mathcal T'_t \degen \mathcal T^*.
  \]
\end{lemma}

Since each $\T'_t$ has holes only in Z-blocks, the main idea of the proof is to utilize the symmetric structure in $\T^*$ to randomly permute every $\T'_t$, ensuring that each Z-block is non-hole in at least one (permuted) $\T'_t$.

To begin our proof, we define the collection of permutations $G$ that we will apply to $\T^*$, which we call the \emph{shuffling group} of the standard form tensor $\T^*$:

\begin{definition}
  Let $\T^*$ be a level-$\lvl$ standard form tensor with parameters $\BK{n_t, i_t, j_t, k_t, \tilde\alpha_t}_{t \in [m]}$. Denoting by $\Sym_A$ the set of permutations over $A$, we define the \emph{shuffling group} $G$ of $\T^*$ as
  \[
    G \defeq \Sym_{[n_1]} \times \Sym_{[n_2]} \times \cdots \times \Sym_{[n_{m}]}.
  \]
  Every element $\shuf = (\shufi[1], \ldots, \shufi[m]) \in G$ first induces a permutation over $[2N]$:
  \[
    \shuf : 2\bk{\sum_{t' = 1}^{t - 1} n_{t'} + a} + b \quad \mapsto \quad 2\bk{\sum_{t' = 1}^{t - 1} n_{t'} + \shufi[t](a)} + b
    \qquad
    \forall t \in [m], \, a \in [n_t], \, b \in \BK{0, 1}.
  \]
  (That is, it first considers every 2 consecutive numbers as a unit, so $[2N]$ is regarded as $N$ units; then, it permutes the first $n_1$ units according to $\shufi[1] \in \Sym_{[n_1]}$, permutes the next $n_2$ units according to $\shufi[2] \in \Sym_{[n_2]}$, and so on.) Further, it induces a permutation over all small blocks: It maps each small block of $\T^*$, say $Z_{\hat K}$, by reordering the entries of $\hat K$ according to the permutation $\shuf : [2N] \to [2N]$ above:
  \[
    \phi(Z_{\hat K}) \defeq Z_{\hat K'},
    \quad \textup{where} \quad
    \hat K_{i} = \hat K'_{\phi(i)}
    \quad
    \forall i \in [2N].
  \]
  It also permutes X and Y-blocks while $\phi(X_{\hat I})$ and $\phi(Y_{\hat J})$ are defined in the same way.
\end{definition}

So far, each element $\shuf \in G$ gives block-level permutations over small blocks $X_{\hat I}$, $Y_{\hat J}$, $Z_{\hat K}$. However, different from X and Y-blocks, some of the Z-blocks do not appear as variables of $\T^*$; only the available ones appear. So we add an observation that $\shuf$ is also a permutation over all \emph{available} Z-blocks:

\begin{claim}
  For every $\shuf \in G$ and Z-block $Z_{\hat K}$, let $Z_{\hat K'} = \shuf(Z_{\hat K})$. $Z_{\hat K}$ is available if and only if $Z_{\hat K'}$ is available. This implies that $\shuf$ restricted on all available Z-blocks is still a permutation.
\end{claim}

\begin{proof}
  The availability of $Z_{\hat K}$ only depends on the frequency of occurrence of $(\hat{K}_{2s - 1}, \hat{K}_{2s})$ accross all $s \in \left(\sum_{t' = 1}^{t - 1} n_{t'}, \sum_{t' = 1}^t n_{t'}\right]$ for every $t \in [m]$. This quantity remains unchanged between $\hat K$ and $\hat K'$, as our permutation $\shuf$ shuffles each pair of indices $(\hat{K}_{2s-1}, \hat{K}_{2s})$ as a unit, and its shuffling destination stays within the region $\phi(s) \in \left(\sum_{t' = 1}^{t - 1} n_{t'}, \sum_{t' = 1}^t n_{t'}\right]$. Thus, the claim follows.
\end{proof}

A similar observation is that $(X_{\hat I}, Y_{\hat J}, Z_{\hat K})$ is a triple in $\T^*$ if and only if $(\phi(X_{\hat I}), \phi(Y_{\hat J}), \phi(Z_{\hat K}))$ is a triple. Further, we keep the relative arrangement of variables within each block unchanged while shuffling the blocks. This induces a variable-level permutation of all variables of $\T^*$ which we still denote by $\shuf$; we claim that it keeps the structure of $\T^*$:

\begin{claim}
  Let $\shuf(x)$, $\shuf(y)$, $\shuf(z)$ denote the images of variables $x, y, z \in \T^*$ under the variable-level permutation $\shuf$. Then, there is a term $(x, y, z)$ in $\T^*$ if and only if the term $(\shuf(x), \shuf(y), \shuf(z))$ is in $\T^*$, which means $\shuf$ is an automorphism of $\T^*$.
\end{claim}

Suppose $\T'$ is a broken copy of $\T^*$, i.e., some of the available Z-blocks are holes in $\T'$. After using $\shuf \in G$ to rename the variables in $\T'$, we get another broken tensor denoted by $\shuf(\T')$, but the hole blocks in $\shuf(\T')$ is likely to be different from $\T'$. The last property we concern about $G$ is that a random $\shuf \in G$ can move a hole block $Z_{\hat K}$ to a random place, so that every available block of $\shuf(\T')$ has equal probability to become a hole:

\begin{claim}
  Let $Z_{\hat K}$ be a fixed available Z-block in $\T^*$. Picking $\shuf \in G$ uniformly at random, the image $Z_{\hat K'} = \shuf(Z_{\hat K})$ follows the uniform distribution over all available Z-blocks in $\T^*$.
\end{claim}

\begin{proof}
  For every pair of available Z-blocks $Z_{\hat K}$ and $Z_{\hat K'}$, the number of permutations in $G$ that maps $Z_{\hat K} \mapsto Z_{\hat K'}$ is given by
  \[
    \prod_{t = 1}^{m} \prod_{k' = 0}^{2^{\lvl - 1}} (\tilde\alpha_{t}(k') \cdot n_t)!
  \]
  which is independent of $Z_{\hat K}, Z_{\hat K'}$. Thus the claim holds.
\end{proof}

As an implication, for a broken standard form tensor $\T'$ with a fraction $\eta$ of non-holes, and for a fixed available block $Z_{\hat K}$, we have
\[
  \numberthis
  \label{eq:random_hole_probability}
  \Pr_{\shuf \in G} \Bk{Z_{\hat K} \textup{ is non-hole in } \shuf(\T')} = \eta.
\]
(Here $\shuf$ is taken uniformly at random from $G$.)
With the help of \eqref{eq:random_hole_probability}, we now state the proof of \cref{lemma:hole_lemma}.

\begin{proofof}{\cref{lemma:hole_lemma}}
  We use the probabilistic method. First, the number of available small Z-blocks $Z_{\hat K}$ is at most $(2^{\lvl-1} + 1)^N \le 2^{N \lvl}$.

  Sample $s$ permutations $\phi_1, \ldots, \phi_s \in G$ independently uniformly at random from $G$. Let $\mathcal T''_t \defeq \phi_t(\mathcal T'_t)$ be a tensor isomorphic to $\mathcal T'_t$, obtained by renaming $\T'_t$'s variables according to $\shuf_t$. From \eqref{eq:random_hole_probability}, we know that for a fixed available Z-block $Z_{\hat K}$,
  \[
    \begin{aligned}
      \Pr\Bk{Z_{\hat K} \textup{ is a hole in } \mathcal T''_t \;\; \forall t \in [m]}
      &= \prod_{t=1}^s (1 - \eta_t) \le \prod_{t=1}^s e^{-\eta_t} \le e^{- N\lvl-1} \\
      &\le \frac{1}{e} \cdot \frac{1}{\textup{number of available Z-blocks}}.
    \end{aligned}
  \]
  Taking summation over all available small Z-blocks, we know the expected number of $Z_{\hat K}$ that is a hole in all $\mathcal T''_t$ is at most $1/e < 1$. Thus, there exists a sequence of permutations $\phi_1, \cdots, \phi_s$, ensuring that every available block $Z_{\hat K}$ has at least one $t \in [s]$ for which $Z_{\hat K}$ is non-hole in $\mathcal T''_t$, denoted as $t(\hat K)$.

  We do the following degeneration from $\bigoplus_{t=1}^{s} \mathcal T''_t$:
  \begin{itemize}
  \item For each available small block $Z_{\hat K}$, recall that it is not a hole in $\mathcal T''_{t(\hat K)}$. Zero out this small block in all other $\mathcal T''_{t'}$ where $t' \ne t(\hat K)$.
  \item Identify all $\mathcal T''_t$ after the previous step of zeroing out.
  \end{itemize}
  Here ``identify'' means to glue all copies together (see \cref{sec:degen_and_value}). Notice that there are no holes in the X and Y-variables, so the X and Y-variables of all broken copies $\mathcal T''_t$ are the same. After the first step of zeroing out, each available Z-block appears in exactly one broken copy, so in the target tensor
  \[
    \T^* = \sum_{(X_{\hat I}, Y_{\hat J}, Z_{\hat K})} \T^* \vert_{X_{\hat I}, Y_{\hat J}, Z_{\hat K}},
  \]
  each term on the RHS will be added exactly once from the copy $\T''_{t(\hat K)}$ in which the block $Z_{\hat K}$ is not zeroed out. It implies that the obtained tensor after the identification step is exactly $\mathcal T^*$.
\end{proofof}

The above lemma shows that when the sum of non-hole fractions $\sum \eta_t \ge N \lvl + 1$, the direct sum of broken tensors can degenerate to a single copy of $\T^*$. When $\sum \eta_t$ is much larger, we can naturally obtain multiple copies of $\T^*$, as described below:

\begin{cor}
  \label{cor:hole_lemma}
  Let $\T^*$ be a level-$\lvl$ standard form tensor with parameters $\BK{n_t, i_t, j_t, k_t, \tilde\alpha_t}_{t \in [m]}$ ($i_t + j_t + k_t = 2^\lvl$ for all $t$) and $N = \sum_{t=1}^m n_t$. Assume that we have $s$ broken copies of $\mathcal T^*$ named $\mathcal T'_1, \ldots, \mathcal T'_s$, with fractions of non-holes $\eta_1, \ldots, \eta_s$. Then there is a way of degenerating them into $s'$ complete copies of $\T^*$, i.e.,
  \[
    \bigoplus_{t=1}^s \mathcal T'_t \degen (\mathcal T^*)^{\oplus s'}, \quad \textup{where} \quad s' \defeq \floor{\frac{\sum_{t=1}^s \eta_t}{N \lvl + 2}}.
  \]
\end{cor}
\begin{proof}
  We divide $\mathcal T'_1, \ldots, \mathcal T'_s$ into subsets such that the sum of $\eta_t$ in each subset satisfies $N\lvl+1 \le \sum \eta_t \le N\lvl+2$. The number of subsets is at least $s'$. Then, we follow \cref{lemma:hole_lemma} to degenerate each subset into a complete copy of $\mathcal T^*$.
\end{proof}

\section{Improving High-Power Global Values}
\label{sec:global}
\label{sec:global_value}
In this section, we explain how to upper bound the value of $\CW_{q}^{\otimes 2^{\lvl-1}}$, based on known restricted-splitting values of the level-$\lvl$ components. Formally, the inputs to this algorithm are:
\begin{itemize}
\item the Z-split distribution $\splres_{i,j,k}$ of component $(i,j,k)$, for all $i+j+k = 2^\lvl$; 
\item the lower bound on $\valsix(T_{i,j,k},\, \tilde\alpha_{i,j,k})$ for all level-$\lvl$ components ($i + j + k = 2^\lvl$);
\end{itemize}
This algorithm will optimize the following parameter:
\begin{itemize}
    \item distribution $\alpha$ over all level-$\lvl$ components $(i,j,k)$, i.e., $i + j + k = 2^\lvl$. We will denote its marginals by $\alphx, \alphy, \alphz$.
\end{itemize}
Finally, it will output:
\begin{itemize}
\item a lower bound on $\valsix(\T)$ where $\T \defeq (\CW_{q}^{\otimes 2^{\lvl-1}})^{\otimes n}$. It implies a lower bound on $\valsix(\CW_q^{\otimes 2^{\lvl-1}})$.
\end{itemize}

\paragraph*{Notations.} We will let $N = n \cdot 2^{\lvl-1}$. The algorithm will focus on both level-$\lvl$ and level-$\lastlvl$ partitions of $\CW_{q}^{\otimes N}$. In this section, for convenience, we will refer to the level-$\lvl$ variable blocks as \emph{large blocks} and the level-$\lastlvl$ blocks as \emph{small blocks}. We will use $I, J, K$ to denote \emph{large index sequences} and $\hat I, \hat J, \hat K$ to denote \emph{small index sequences}. Furthermore, we will consider the large block $X_I$ as a collection of small blocks, using $X_{\hat I} \in X_I$ to indicate the inclusion relationship between small and large blocks.

\subsection{Algorithm Description} \label{sec:global-algo}

\paragraph*{Step 1. Lower bound the values of level-$\lastlvl$ components.} In this section, these lower bounds are given as the input. Hence this step is trivial.

\paragraph*{Step 2. Choose a distribution.} We can choose any parameters $\alpha(i, j, k)$ where $i + j + k = 2^\lvl$. The only constraint is that they have to form a probability distribution. (Details on how we will optimize these parameters will be presented in \Cref{sec:result}.)

\paragraph*{Step 3. Asymmetric Hashing.} Let $M_0$ be a parameter which will be specified in \cref{sec:global_analysis}, and $M$ be a prime number in $[M_0, 2M_0]$. We will apply the asymmetric hashing in \cref{sec:hashing} with $M$ as the modulus. After hashing and zeroing out, for any pair of retained large block triples $(X_I, Y_J, Z_K)$ and $(X_{I'}, Y_{J'}, Z_{K'})$, we have $I \ne I',~ J \ne J'$. That is, those retained triples can only share Z-blocks. Moreover, all retained triples obey the joint distribution $\alpha$. 

(Note that when $\alpha \notin D^*(\alphx, \alphy, \alphz)$, i.e., when it is not the maximum entropy distribution consistent with these marginals, there will be a certain hash loss. We already took this hash loss into account in \cref{sec:hashing}. Also, we will explicitly compute the value of such hash loss in \cref{sec:global_analysis}.)

\paragraph{Additional Zeroing-Out Step 1.} We now fix a triple $(X_I, Y_J, Z_K)$. For each level-$\lvl$ component $(i,j,k)$, we use $S_{i,j,k}$ to denote the set of positions $t \in [n]$ where $(I_t, J_t, K_t) = (i, j, k)$. Moreover, $S_{*,*,k}$ denotes the set of positions $t$ with $K_t = k$.

For a fixed small block $Z_{\hat K} \in Z_K$ and a set $S \subseteq S_{*,*,k}$ for some $k$, we will use $\split_k(\hat K, S)$ to denote the distribution of $\hat K_{2t - 1}$ over all $t \in S$. (Note that for those $t \in S$, we always have $\hat K_{2t} = k - \hat{K}_{2t - 1}$. Hence such distribution captures how those $k$'s split into $\hat K_{2t - 1} + \hat K_{2t}$.) When the index $k$ to split is clear from context, we also omit the subscript and simply write $\split(\hat K, S)$.

For each $k \in [0, 2^\lvl]$, we define its average split distribution
\[
  \splresavg_{*,*,k} (k') = \frac{1}{\alphz(k)} \sum_{i+j = 2^\lvl - k} \alpha(i,j,k) \cdot \splres_{i,j,k} (k').
  \numberthis \label{eq:def_splresavg_star_g}
\]
Here $\splres_{i,j,k}$ are given in the inputs.

\begin{definition} \label{def:global-compatible}
  A small block $Z_{\hat K} \in Z_K$ is said to be \emph{compatible} with a large triple $(X_I, Y_J, Z_K)$ if the following two conditions are satisfied:
  \begin{enumerate}
  \item For each $k \in \midBK{0, \ldots, 2^\lvl}$, $\split(\hat K, S_{*,*,k}) = \tilde\alpha_{*,*,k}$. \label{item:split-match}
  \item For each large component $(i,j,k)$ (i.e., $i + j + k=  2^\lvl$) satisfying $i = 0$ or $j = 0$, $\split(\hat K, S_{i,j,k}) = \tilde\alpha_{i,j,k}$. \label{item:average}
  \end{enumerate}
\end{definition}

\noindent In this step, we will do the following zeroing-out on small blocks:
\begin{itemize}
\item For each $Z_{\hat K}$, we check \cref{item:split-match} and zero out $Z_{\hat{K}}$ if the condition is not satisfied.
\item For each $X_{\hat I} \in X_I$, since $X_I$ (if retained) is in a unique triple $(X_I, Y_J, Z_K)$, we can define the set $S_{i,j,k}$ w.r.t. that triple. For all components $(i,0,k)$, we define 
  \[ \splres^{\textup{(X)}}_{i,0,k}(i') \defeq \splres_{i,0,k}(2^{\lvl - 1} - i'). \]
  This is the X-split distribution corresponding to the Z-split distribution $\splres_{i,0,k}$. If for any $(i, 0, k)$, $\split(\hat I, S_{i, 0, k}) \ne \splres^{\textup{(X)}}_{i,0,k}$, we then zero out this $X_{\hat I}$.
\item For each $Y_{\hat J}$, we do a similar zeroing-out as $X_{\hat{I}}$. For all components $(0,j,k)$, we define
  \[ \splres^{\textup{(Y)}}_{0,j,k}(j') \defeq \splres_{0,j,k}(2^{\lvl - 1} - j'). \]
  Suppose for some component $(0, j, k)$, the split distribution $\split(\hat J, S_{0, j, k}) \neq \splres^{\textup{(Y)}}_{0,j,k}$, then $Y_{\hat J}$ will be zeroed out.
\end{itemize}

\begin{claim}
  \label{lemma:triple_implies_compatible}
  After Additional Zeroing-Out Step 1, a remaining small block $Z_{\hat{K}} \in Z_K$ can form a triple with remaining $X_{\hat{I}} \in X_I$, $Y_{\hat{J}} \in Y_J$ only when $Z_{\hat{K}}$ is compatible with triple $(X_I, Y_J, Z_K)$. 
\end{claim}
\begin{proof}
To see this, first notice that after Additional Zeroing-Out Step 1, all $Z_{\hat{K}}$ that are not zeroed out satisfy \Cref{item:split-match} in \Cref{def:global-compatible}. Hence, if $Z_{\hat{K}}$ is not compatible with $(X_I, Y_J, Z_K)$, it must be that for some $(i,j,k)$ with $i = 0$ or $j = 0$, $\split(\hat{K}, S_{i,j,k}) \neq \splres_{i,j,k}$ (\Cref{item:average}).

Due to symmetry, we only have to discuss the case where $j = 0$. By our zeroing-out rules, we know that $\split(\hat{I}, S_{i,0,k}) = \splres^{\textup{(X)}}_{i,0,k}$. As $j = 0$ and $\hat{I} + \hat{J} + \hat{K} = (2^{\lvl - 1}, 2^{\lvl - 1}, \ldots, 2^{\lvl - 1})$, we necessarily have that $\split(\hat{K}, S_{i,0,k})(k') = \split(\hat{I}, S_{i,0,k})(2^{\lvl - 1} - k') = \splres^{\textup{(X)}}_{i,0,k}(2^{\lvl - 1} - k') = \splres_{i,0,k}(k')$. This shows that $\split(\hat{K}, S_{i,j,k}) = \splres_{i,j,k}$. Contradiction.
\end{proof}
We use $\mathcal T^{(1)}$ to denote the tensor after Additional Zeroing-Out Step 1.

\paragraph*{Additional Zeroing-Out Step 2.} Before we introduce this step, we have to make the following definition.
\begin{definition}
  \label{def:useful_g}
  A small block $Z_{\hat K}$ is said to be \emph{useful} for a large triple $(X_I, Y_J, Z_K)$ if and only if the following two conditions hold:
  \begin{itemize}
  \item $Z_{\hat K} \in Z_K$ and $(I, J, K)$ is consistent with $\alpha$;
  \item For each large component $(i,j,k)$, $\split(\hat K, S_{i,j,k}) = \tilde\alpha_{i,j,k}$.
  \end{itemize}
\end{definition}

\noindent In this step, we will zero out any small Z-block $Z_{\hat K} \in Z_K$ such that
\begin{itemize}
\item $Z_{\hat K}$ is compatible with more than one triple, or
\item $Z_{\hat{K}}$ is not useful for the unique triple $(X_I, Y_J, Z_K)$ that it is compatible with.
\end{itemize}
After such zeroing out, we call the obtained tensor $\T^{(2)}$. We claim $\T^{(2)} \cong \bigoplus_{(X_I, Y_J, Z_K)} \T^{(2)} \vert_{X_I, Y_J, Z_K}$ due to the first zeroing-out rule here.

\begin{figure}[ht]
  \centering
  \includegraphics[width = 0.95 \textwidth]{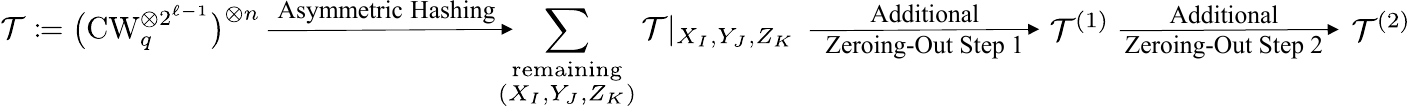}
  \caption{Tensor Degeneration Process}
  \label{fig:tensor7}
\end{figure}

Fixing a triple $(X_I, Y_J, Z_K)$, the structure of $\T^{(2)} \vert_{X_I, Y_J, Z_K}$ is as follows. Suppose no blocks are zeroed out due to the first rule, i.e., being compatible with multiple triples. In this ideal case, $\T^{(2)} \vert_{X_I, Y_J, Z_K}$ is isomorphic to
\[ \T^* \defeq \bigotimes_{i + j + k = 2^\lvl} T_{i,j,k}^{\otimes n \alpha(i,j,k)}[\splres_{i,j,k}]. \]
However, in reality, several small Z-blocks are additionally zeroed out from $\T^*$ due to the first rule, resulting in $\T^{(2)} \vert_{X_I, Y_J, Z_K}$ being a broken copy of $\T^*$ with some holes in its Z-variables.

One can see that $\T^*$ exactly matches \cref{def:standard_form_tensor} with parameters $\midBK{n\alpha(i,j,k), i,j,k, \tilde\alpha_{i,j,k}}_{i+j+k=2^\lvl}$. This will allow us to apply the hole lemma in \cref{sec:hole_lemma} to fix the holes in $\T^{(2)} \vert_{X_I, Y_J, Z_K}$.

\paragraph*{Step 4: Fix the holes and degenerate each triple independently.} Unlike Step 4 of \Cref{sec:2nd-improved-algo}, here we will first fix holes before degenerating into matrix multiplications. This is for two reasons: (1) The hole lemma in \Cref{sec:hole_lemma} allows us to directly fix holes for tensors, and (2) the degeneration here requires symmetrization, which may introduce holes to X/Y variables if they have not been fixed already. (Fixing holes is much easier when they are only in Z-variables.)

For each triple $(X_I, Y_J, Z_K)$, $\T^{(2)}\vert_{X_I, Y_J, Z_K}$ is a broken copy of $\mathcal T^*$. The \emph{fraction of non-holes} in $\T^{(2)}\vert_{X_I, Y_J, Z_K}$ (defined in \cref{def:broken_standard_form_tensor}) is denoted as $\fracnonholeIJK$. Letting
\[
  m' \defeq \floor{\frac{\sum_{(X_I, Y_J, Z_K)} \fracnonholeIJK}{ n \lvl + 2 }},
\]
it follows from \cref{cor:hole_lemma} that
\[
  \bigoplus_{(X_I, Y_J, Z_K) \textup{ remaining}} \mathcal T^{(2)} \vert_{X_I, Y_J, Z_K} \; \degen \; (\mathcal T^*)^{\oplus m'}.
\]

Recall that $\T^* = \bigotimes_{i+j+k = 2^\lvl} T_{i,j,k}^{\otimes n \alpha(i,j,k)}[\tilde\alpha_{i,j,k}]$. By given bounds of restricted-splitting values, we have $\valsix(\mathcal T^*) \ge \prod_{i+j+k=2^\lvl} \valsix(T_{i,j,k}, \tilde\alpha_{i,j,k})^{n\alpha(i,j,k)}$. As conclusion,
\[
  \valsix(\T) \;\ge\; v \;\defeq\; m' \cdot \prod_{i+j+k=2^\lvl} \valsix(T_{i,j,k}, \tilde\alpha_{i,j,k})^{n\alpha(i,j,k)}.
  \numberthis \label{eq:result_mid_g}
\]

The algorithm concludes here, and $v$ is the output that results in a value bound for the CW tensor. However, there is one more implicit step -- symmetrization -- which is hidden under the notation of $\valsix$. This occurs before degenerating into matrix multiplication tensors. This is because, by definition, $\valsix(\T)$ represents the (maximized) total volume of matrix multiplications that $\sym_6(\T)$ degenerates into. Here we are able to perform such symmetrization because we have already fixed all the holes in the broken copies of $\T^*$.

\subsection{Analysis}
\label{sec:global_analysis}

In the previous subsection, we have described our algorithm to degenerate $\T \defeq (\CW_q^{\otimes n})$ and bound its value. This subsection adds calculations and completes the proof.

\paragraph{Asymmetric Hashing.} Similar to the previous sections, we adopt the following notations:
\begin{itemize}
\item $\numxblock = \numyblock \ge \numzblock$ represent the number of large (level-$\lvl$) X, Y, and Z-blocks that are consistent with $\alphx, \alphy, \alphz$, respectively. We have $\numxblock = 2^{n H(\alphx) + o(n)}$ (similar for Y and Z).
\item $\numalpha$ is the number of triples $(X_I, Y_J, Z_K)$ that are consistent with $\alpha$; $\numtriple$ is the number of triples $(X_I, Y_J, Z_K)$ whose marginals are consistent with $\alphx, \alphy, \alphz$. We have $\numalpha = 2^{n H(\alpha) + o(n)}$.
\item $\numretain$ represents the number of retained triples after the asymmetric hashing process (which are all consistent with $\alpha$).
\item Let $\pcomp$ be a parameter to be defined later. Roughly speaking, it is the probability of a small block $Z_{\hat K}$ being compatible with a random triple $(X_I, Y_J, Z_K)$.
\end{itemize}
We let $M_0 = 8 \cdot \max\bk{\frac{\numtriple}{\numxblock}, \frac{\numalpha \cdot \pcomp}{\numzblock}}$ and let $M \in [M_0, 2M_0]$ be a prime. Applying the asymmetric hashing according to \cref{sec:hashing} with modulus $M$, we know the number of retained triples is
\[
  \E[\numretain] \ge \frac{\numalpha}{M} \cdot 2^{-o(n)} = \frac{\numalpha}{M_0} \cdot 2^{-o(n)} = \min\bk{\frac{\numalpha \cdot \numxblock}{\numtriple}, \, \frac{\numzblock}{\pcomp}} \cdot 2^{-o(n)}.
  \numberthis \label{eq:sec6_numretain}
\]

\paragraph{Probability of being compatible.} One of the remaining task is to define and calculate $\pcomp$. We start by defining its prerequisite:

\begin{definition}[Typicalness]
  Let $\gamma$ be a distribution over $\midBK{0, \ldots, 2^{\lvl-1}}^2$ whose probability density function is given by
  \[
    \gamma(k_l, k_r) \defeq \sum_{\substack{i + j + k = 2^\lvl \\ k = k_l + k_r}} \alpha(i, j, k) \cdot \splres_{i,j,k}(k_l).
  \]
  We say a small block $Z_{\hat K}$ is \emph{typical} if and only if the frequency of occurrences of $(\hat K_{2t-1}, \hat K_{2t})$ matches the probability distribution $\gamma$, i.e.,
  \[
    \forall 0 \le k \le 2^{\lvl} \textup{ and } k_l + k_r = k, \quad
    \abs{\BK{t \in [n] \mymiddle (\hat K_{2t - 1}, \hat K_{2t}) = (k_l, k_r)}} = \gamma(k_l, k_r) \cdot n.
  \]
\end{definition}

We show the following equivalent condition of typicalness:

\begin{claim}
  Assume $Z_K$ is a large triple consistent with $\alphz$.
  Then, a small block $Z_{\hat K} \in Z_K$ is \emph{typical} if and only if it was not zeroed out in Additional Zeroing-Out Step 1, i.e., $\Split(\hat K, S_{*, *, k}) = \splresavg_{*, *, k}$ for all $k = 0, 1, \ldots, 2^\lvl$, where $S_{*, *, k}$ is the set of positions with $K_t = k$; $\splresavg_{*, *, k}$ is defined according to \eqref{eq:def_splresavg_star_g}.
\end{claim}

\begin{proof}
  Suppose $Z_{\hat K}$ satisfies $\Split(\hat K, S_{*, *, k}) = \splresavg_{*, *, k}$ for all $k = 0, 1, \ldots, 2^\lvl$. Then, for any $0 \le k \le 2^\lvl$ and $k_l + k_r = k$, we can calculate
  \begin{align*}
    \abs{\BK{t \in [n] \;\middle|\; (\hat K_{2t-1}, \hat K_{2t}) = (k_l, k_r)}}
    &= \abs{S_{*, *, k}} \cdot \Split(\hat K, S_{*, *, k})(k_l) \\
    &= \alphz(k) \cdot n \cdot \splresavg_{*, *, k}(k_l) \\
    &= \alphz(k) \cdot n \cdot \frac{1}{\alphz(k)} \sum_{i + j = 2^{\lvl} - k} \alpha(i,j,k) \cdot \splres_{i,j,k}(k_l) \\
    &= \gamma(k_l, k_r) \cdot n.
  \end{align*}
  Thus, $Z_{\hat K}$ is typical. Similarly, if $Z_{\hat K}$ is typical, we can also determine $\Split(\hat K, S_{*, *, k}) = \splresavg_{*, *, k}$. This concludes the proof.
\end{proof}

In the rest of this subsection, if a triple $(X_I, Y_J, Z_K)$ is consistent with the joint component distribution $\alpha$ that we choose, we say these blocks $X_I, Y_J, Z_K$ are \emph{matchable} to each other. Next, we define $\pcomp$:

\begin{definition}
  Suppose $Z_{\hat K} \in Z_K$ is a typical block, and $X_I$ is an X-block matchable to $Z_{K}$ chosen uniformly at random (i.e., they form a triple $(X_I, Y_J, Z_K)$ consistent with $\alpha$). $\pcomp$ is defined as the probability of $Z_{\hat K}$ being compatible with $X_I, Y_J$. Similar to \cref{sec:2nd}, due to symmetry, $\pcomp$ is independent of which block $Z_{\hat K}$ we choose.
\end{definition}

We calculate $\pcomp$ by the following lemma:

\begin{lemma}
  \label{lemma:pcomp_g}
  We have $\pcomp = \alphap^{n + o(n)}$, where
  \begin{gather*}
    \alphap \defeq
    2^{H(\alphz) - H(\gamma)} \cdot \prod_{\substack{i + j + k = 2^\lvl \\ i = 0 \textup{ or } j = 0}} 2^{\alpha(i,j,k) \cdot H(\splres_{i,j,k})}
    \cdot \prod_{k = 0}^{2^\lvl} 2^{\alpha(\plusplusk) \cdot H(\splresavg_{\plusplusk})}, \qquad \text{and} \\
    \alpha(\plusplusk) \defeq \sum_{\substack{i, j > 0 \\ i + j + k = 2^\lvl}} \alpha(i,j,k), \qquad
    \splresavg_{\plusplusk} \defeq \frac{1}{\alpha(\plusplusk)} \sum_{\substack{i, j > 0 \\ i + j + k = 2^\lvl}} \alpha(i, j, k) \cdot \splres_{i,j,k},
  \end{gather*}
  and $\gamma$ is the typical distribution defined above. ($\splres_{\plusplusk}$ is a split distribution of Z-index $k$.)
\end{lemma}

\begin{proof}
  Fixing a large block $Z_K$ consistent with $\alphz$, we denote by $\typicalset$ the set of typical blocks within the large block $Z_K$. Since $\pcomp$ is identical for all $Z_{\hat K} \in \typicalset$, it will also be the same for a uniformly randomly chosen $Z_{\hat K} \in \typicalset$. Independently, we sample a large X-block $X_I$ that is matchable to $Z_K$ (i.e., they form a triple $(X_I, Y_J, Z_K)$ consistent with $\alpha$), also uniformly at random. Then, we have
  \begin{align*}
    \pcomp &= \Pr_{I, \hat K}[Z_{\hat K} \textup{ is compatible with } X_{I}] \\
           &= \E_{I} \Bk{ \Pr_{\hat K} \Bk{Z_{\hat K} \textup{ is compatible with } X_{I}} } \\
           &= \E_{I} \Bk{ \frac{\abs{\BK{Z_{\hat K'} \in \typicalset \;\middle|\; Z_{\hat K'} \textup{ is compatible with } X_I}}}{\abs{\typicalset}} }.
             \numberthis \label{eq:pcomp_mid_g}
  \end{align*}
  The expectation on the last line is taken over all $X_I$ matchable to $Z_K$. In fact, the content inside the expectation is identical for all $X_I$ due to symmetry. We arbitrarily fix an $X_I$ and further calculate the numerator and denominator, respectively.

  \paragraph{Numerator.} We count the number of desired $Z_{\hat K}$ by counting the number of ways to split all indices $k$ in the sequence $K$. The constraint of being compatible with $X_I$ is equivalent to the following two conditions:
  \begin{enumerate}
  \item[(a)] For $i + j + k = 2^\lvl$ where $i = 0$ or $j = 0$, $\split(\hat K, S_{i, j, k}) = \splres_{i, j, k}$.
  \item[(b)] For $k \in [0, 2^{\lvl}]$, $\split(\hat K, S_{*, *, k}) = \splresavg_{*, *, k}$.
  \end{enumerate}
  (The second condition is a requirement of both being compatible and being typical.) Then we ``subtract'' (a) from (b), obtaining a group of equivalent conditions as follows:
  \begin{enumerate}
  \item[(a)] For $i + j + k = 2^\lvl$ where $i = 0$ or $j = 0$, $\split(\hat K, S_{i, j, k}) = \splres_{i, j, k}$.
  \item[(c)] For $k \in [0, 2^{\lvl}]$, let $S_{\plusplusk}$ denote the set of positions $t \in [n]$ where $(I_t, J_t, K_t)$ satisfies $K_t = k$ and $I_t, J_t \ne 0$. We require $\split(\hat K, S_{\plusplusk}) = \splresavg_{\plusplusk}$, where $\splresavg_{\plusplusk}$ is defined in the statement of \cref{lemma:pcomp_g}.
  \end{enumerate}
  One can see that requiring (a) and (b) is equivalent to requiring (a) and (c). The advantage of doing so is that the set of positions involved in these requirements are disjoint. Specifically, all requirements of types (a) and (c) have the form $\split(\hat K, S) = \splres$ for some position set $S$ and split distribution $\splres$: For (a), they are $S = S_{i,j,k}$ and $\splres = \splres_{i,j,k}$; for (c), they are $S = S_{\plusplusk}$ and $\splres = \splresavg_{\plusplusk}$. Every index position $t \in [n]$ belongs to the set $S$ of exactly one requirement.

  For each of these requirements, namely $\split(\hat K, S) = \splres$, the number of ways to choose $(\hat K_{2t-1}, \hat K_{2t})$ for all $t \in S$ equals
  \[
    \binom{|S|}{|S| \splres(0), \, |S| \splres(1), \, \ldots \, , \, |S| \splres(2^{\lvl-1})}
    = \binom{|S|}{[|S| \cdot \splres(k_l)]_{k_l \in [0, 2^{\lvl - 1}]}}
    = 2^{|S| \cdot H(\splres) + o(|S|)}.
  \]
  Taking product over all requirements, we have
  \begin{align*}
    & \phantom{{}={}} \abs{\BK{Z_{\hat K'} \in \typicalset \;\middle|\; Z_{\hat K'} \textup{ is compatible with } X_I}} \\
    &= \prod_{\substack{i + j + k = 2^\lvl \\ i = 0 \text{ or } j = 0}} 2^{n \alpha(i, j, k) \cdot H(\splres_{i,j,k})} \cdot \prod_{k=0}^{2^\lvl} 2^{n \alpha(\plusplusk) \cdot H(\splresavg_{\plusplusk})} \cdot 2^{o(n)}.
    \numberthis \label{eq:lambda_denominator_g}
  \end{align*}

  \paragraph{Denominator.} To calculate the denominator $|\typicalset|$, we only need to notice the symmetry between different $K$, i.e., the size of $\typicalset$ should be identical for all $K$'s consistent with $\alphz$. Moreover, they are disjoint for different $K$'s. Therefore, we may calculate
  \begin{align*}
    \abs{\typicalset} &= \frac{\abs{\bigcup_{K' \textup{ consistent with } \alphz} \typicalset[K']}}{\numzblock} \\
                      &= \left. \binom{n}{[n\gamma(k_l, k_r)]_{k_l, k_r}} \middle/ \binom{n}{[n \alphz(k)]_{k}} \right. \\
                      &= \frac{2^{n H(\gamma)}}{2^{n H(\alphz)}} \cdot 2^{o(n)}.
  \end{align*}
  Combined with \eqref{eq:pcomp_mid_g} and \eqref{eq:lambda_denominator_g}, we have
  \begin{align*}
    \pcomp &= \frac{\abs{\BK{Z_{\hat K'} \in \typicalset \;\middle|\; Z_{\hat K'} \textup{ is compatible with } X_I}}}{\abs{\typicalset}} \\
           &= \left. \prod_{\substack{i + j + k = 2^\lvl \\ i = 0 \text{ or } j = 0}} 2^{n \alpha(i, j, k) \cdot H(\splres_{i,j,k})} \cdot \prod_{k=0}^{2^\lvl} 2^{n \alpha(\plusplusk) \cdot H(\splresavg_{\plusplusk})} \middle/ \frac{2^{n H(\gamma)}}{2^{n H(\alphz)}} \cdot 2^{o(n)} \right. \\
           &= \alphap^{n} \cdot 2^{o(n)} \;=\; \alphap^{n + o(n)},
  \end{align*}
  where $\alphap$ is defined in the statement of \cref{lemma:pcomp_g}. This concludes the proof.
\end{proof}

\paragraph{Probability of being holes.} Fixing a small block $Z_{\hat K}$ that is \emph{useful} for some retained triple $(X_I, Y_J, Z_K)$ (see \cref{def:useful_g}). Below, we analyze the probability of $Z_{\hat K}$ being a hole.

\begin{claim}
  \label{claim:hole_frac_low}
  Fixing a retained triple $(X_I, Y_J, Z_K)$ (it must be consistent with $\alpha$) and a small block $Z_{\hat K} \in Z_K$ useful for that triple, the probability of $Z_{\hat K}$ being a hole in $\T^{(2)} \vert_{X_I, Y_J, Z_K}$ (i.e., being compatible with a different remaining triple $(X_{I'}, Y_{J'}, Z_K)$) is at most $1/8$.
\end{claim}

\begin{proof}
  A necessary condition of $Z_{\hat K}$ being a hole is that there exists $I' \ne I$ matchable to $K$ such that (1) $Z_{\hat K}$ is compatible with $X_{I'}$; and (2) $I'$ is hashed to the same slot as $K$, i.e., $\hashx(I') = \hashz(K)$. We calculate the expected number of such $I'$ to establish an upper bound on the probability of the existence of such $I'$:
  \begin{align*}
    \Pr\Bk{Z_{\hat K} \textup{ is a hole}}
    &\le \sum_{I' \ne I} \ind\Bk{Z_{\hat K} \textup{ is compatible with } X_{I'}} \cdot \Pr[\hashx(I') = \hashz(K) \mid \hashx(I) = \hashz(K)] \\
    &= \sum_{I' \ne I} \ind\Bk{Z_{\hat K} \textup{ is compatible with } X_{I'}} \cdot \frac{1}{M} \\
    &< \sum_{I'} \ind\Bk{Z_{\hat K} \textup{ is compatible with } X_{I'}} \cdot \frac{1}{M} \\
    &= \frac{\numalpha}{\numzblock} \cdot \Pr_{I' \textup{ matchable to } K} [Z_{\hat K} \textup{ is compatible with } X_{I'}] \cdot \frac{1}{M} \\
    &= \frac{\numalpha \cdot \pcomp}{\numzblock \cdot M} \;\le\; \frac{\numalpha \cdot \pcomp}{\numzblock \cdot M_0} \;\le\; \frac{\numalpha \cdot \pcomp}{\numzblock} \cdot \frac{\numzblock}{8 \cdot \numalpha \cdot \pcomp} \;=\; \frac{1}{8},
  \end{align*}
  where the first equality above holds due to \cref{lemma:hash_independence} (restated below); the third equality holds according to the definition of $\pcomp$ above.
\end{proof}

\HashIndependence*
The proof of this lemma is in \cref{sec:hashing}.

\paragraph{Bounding the value.} We continue to obtain a bound for $\valsix(\T)$ where $\T \defeq (\CW_q^{\otimes 2^{\lvl - 1}})^{\otimes n}$. Recall that for each of the $\numretain$ retained triples $(X_I, Y_J, Z_K)$, the fraction of non-hole blocks in $\T^{(2)} \vert_{X_I, Y_J, Z_K}$ is denoted by $\fracnonholeIJK$. \cref{claim:hole_frac_low} tells that, provided some small block $Z_{\hat K}$ is useful for a retained triple $(X_I, Y_J, Z_K)$, the probability of $Z_{\hat K}$ being a hole is at most $1/8$. Therefore,
\[
  \fracnonholeIJK \defeq 1 - \E_{Z_{\hat K} \textup{ useful for } (X_I, Y_J, Z_K)} \Bk{ \Pr\Bk{Z_{\hat K} \textup{ is a hole in } \T^{(2)} \vert_{X_I, Y_J, Z_K}} } \ge \frac{7}{8}
\]
holds for all $(X_I, Y_J, Z_K)$. According to \cref{cor:hole_lemma}, these broken tensors can degenerate to
\[
  m' \defeq \floor{\frac{\sum_{(X_I, Y_J, Z_K)} \fracnonholeIJK}{n \lvl + 2}}
  \ge \floor{\frac{\numretain \cdot 7}{8(n \lvl + 2)}} = \frac{\numretain}{O(n)}
\]
many copies of standard form tensors $\T^*$. Combined with \eqref{eq:sec6_numretain}, we get
\[
  \E[m'] \ge \min\bk{\frac{\numalpha \cdot \numxblock}{\numtriple}, \, \frac{\numzblock}{\pcomp}} \cdot 2^{-o(n)}.
\]
By the probabilistic method, we know $\T$ can degenerate to at least this many copies of $\T^*$. Finally, we have
\begin{gather*}
  \valsix(\CW_q^{\otimes 2^{\lvl - 1}})^n = \valsix(\T)
  \ge \E[m'] \cdot \valsix(\T^*) \\
  \ge \min\bk{\frac{\numalpha \cdot \numxblock}{\numtriple}, \, \frac{\numzblock}{\pcomp}}
    \cdot \prod_{i + j + k = 2^\lvl} \valsix(T_{i,j,k}, \splres_{i,j,k})^{n \alpha(i,j,k)} \cdot 2^{-o(n)}.
    \numberthis \label{eq:value_before_nth_root}
\end{gather*}

Similar to \cref{sec:2nd}, we define the following notations:
\begin{itemize}
\item $\alphabx \defeq \lim\limits_{n \to \infty} \numxblock^{1/n} = 2^{H(\alphx)}$ is the number of large X-blocks normalized by taking the $n$-th root. Similarly, $\alphabz \defeq \lim\limits_{n \to \infty} \numzblock^{1/n} = 2^{H(\alphz)}$.
\item $\alphant \defeq \lim\limits_{n \to \infty} \numalpha^{1/n} = 2^{H(\alpha)}$ is the number of triples consistent with $\alpha$.
\item With these notations, we have $\max\limits_{\alpha' \in \distShareMargin} \alphant[\alpha'] = \lim\limits_{n \to \infty} \numtriple^{1/n}$ according to \cref{lem:numtriple_singledist}.
\item $\alphap = \lim\limits_{n \to \infty} \pcomp^{1/n}$. Its closed form is given in \cref{lemma:pcomp_g}.
\item $\alphaval \defeq \prod_{i + j + k = 2^\lvl} \valsix(T_{i,j,k}, \splres_{i,j,k})^{\alpha(i,j,k)}$.
\end{itemize}
Then, we take the $n$-th root on both sides of \eqref{eq:value_before_nth_root}, obtaining
\[
  \valsix(\CW_q^{\otimes 2^{\lvl - 1}}) \ge \min\bk{\frac{\alphant \cdot \alphabx}{\max_{\alpha' \in \distShareMargin} \alphant[\alpha']}, \, \frac{\alphabz}{\alphap}} \cdot \alphaval.
  \numberthis \label{eq:numeric_conclusion_g}
\]
Once the distribution $\alpha$ is given, we can verify the lower bound of $\valsix(\CW_q^{\otimes 2^{\lvl - 1}})$ via \cref{alg:verify_g}.

\begin{figure}[h]
  \begin{center}
    \begin{tcolorbox}
      \captionof{algocf}{Verifying a Lower Bound}{\label{alg:verify_g}}
      \vspace{-0.5em}
      Assume $\alpha$ is given. Moreover, for each level-$\lastlvl$ component $T_{i,j,k}$, a pair $(\splres_{i,j,k}, V_{i,j,k})$ is given, indicating that $\valsix(T_{i,j,k}, \splres_{i,j,k}) \ge V_{i,j,k}$.
      \begin{enumerate}
      \item Compute $\alphabx, \alphabz, \alphant, \alphap, \alphaval$ according to their closed forms.
      \item Solve the following convex optimization problem:
        \[
          \begin{array}{cc}
            \textup{maximize} & \alphant[\alpha'] = 2^{H(\alpha')} \\
            \textup{subject to} & \alpha' \in \distShareMargin. \\
          \end{array}
        \]
        This gives us $\max_{\alpha' \in \distShareMargin} \alphant[\alpha']$.
      \item Calculate the lower bound according to \eqref{eq:numeric_conclusion_g}.
      \end{enumerate}
    \end{tcolorbox}
  \end{center}
\end{figure}

We will explain the heuristics for optimizing the parameters $\alpha$ in \cref{sec:result}.

\subsection{Example -- Level-2 Global Value}\label{sec:level-2-global}
As an example, we further analyze the value of $\CW_{q}^{\otimes 2}$. Although in \cref{sec:2nd} we already obtained an improved bound of $\omega<2.375234$, which is better than \cite{coppersmith1987matrix} for the second power of the CW tensor, this bound can still be (significantly) improved by the method in this section. We will break the symmetry in the distributions of components $(1,1,2)$, $(1,2,1)$, and $(2,1,1)$, then use restricted-splitting values based on split distributions different from \cref{sec:2nd}.

\begin{itemize}
\item When we set $\tilde\alpha_{0,2,2}(0) = \tilde\alpha_{0,2,2}(2) = a$ and $\tilde\alpha_{0,2,2}(1) = 1-2a$, using the method from the proof of \cref{lem:non-rot-values}, we have
  \[
    V_\tau^\nrot(T_{0,2,2}, \tilde\alpha_{0,2,2})\geq \lim_{m\to\infty} \bk{\binom{m}{(1-2a)m, \, am, \, am} q^{2(1-2a)m}}^{\tau/m} = \left(\frac{q^{2(1-2a)}}{a^{2a}(1-2a)^{1-2a}}\right)^{\tau}.
  \]
  For $\tilde\alpha_{2,0,2}(0) = \tilde\alpha_{2,0,2}(2) = a$ and $\tilde\alpha_{2,0,2}(1) = 1-2a$, $V_\tau^\nrot\big(T_{2,0,2}, \tilde\alpha_{2,0,2}\big)$ equals the same value.
\item When $\tilde\alpha_{1,1,2}(0) = \tilde\alpha_{1,1,2}(2) = b$ and $\tilde\alpha_{1,1,2}(1) = 1-2b$, similarly to the proof of \cref{lem:non-rot-values},
  \begin{eqnarray*}
    V_\tau^{(3)}(T_{1,1,2}, \tilde\alpha_{1,1,2})
    & \ge & \lim_{m\to\infty} \bk{\binom{m}{m/2}^2 \binom{m}{(1-2b)m,\, bm,\, bm}}^{1/(3m)} \cdot q^{(2(1-2b) + 2b)\tau} \\
    & =& \left(\frac{4}{(1-2b)^{1-2b}b^{2b}}\right)^{1/3}q^{(2-2b)\tau}.
  \end{eqnarray*}
\item For all other components, we use the symmetric Z-marginal split distributions, that is, $\tilde\alpha_{i,j,1}(0) = \tilde\alpha_{i,j,1}(1)=1/2$ and $\tilde\alpha_{i,j,3}(1) = \tilde\alpha_{i,j,3}(2)=1/2$ for all valid $i,j$; for components $(i,j,0)$ or $(i,j,4)$, there is only one Z-marginal split distribution. So the values of all other components (including $(2,2,0),(1,2,1),(2,1,1)$) do not change from \cref{sec:2nd}.
\end{itemize}

\paragraph{Numerical Result.}

By a MATLAB program, we found the following parameters, which can lead to a better bound $\omega<2.374631$ than \cref{sec:2nd}.

\begin{table}[!h] \label{table:result-2nd}
  \captionsetup{font=small}
  \caption{The parameters for the bound $\omega<2.374631$, where $a$ and $b$ are the parameters in Z-marginal split distributions.}
  \centering
  \begin{tabular}{|c|c|c|c|}
    \hline
    \text{Component} & (0,0,4) & (0,4,0) & (4,0,0) \\
    $\alpha$ & 0.00020860 & 0.00024731 & 0.00024731 \\
    \hline
    \text{Component} & (0,1,3) & (0,3,1) & (1,0,3) \\
    $\alpha$ & 0.01211153 & 0.01333318 & 0.01211153 \\
    \hline
    \text{Component} & (1,3,0) & (3,0,1) & (3,1,0) \\
    $\alpha$ & 0.01251758 & 0.01333318 & 0.01251758 \\
    \hline
    \text{Component} & (0,2,2) & (2,0,2) & (2,2,0) \\
    $\alpha$ & 0.10366945 & 0.10366945 & 0.10045791 \\
    $a$ & 0.03477403 & 0.03477403 &  \\
    \hline
    \text{Component} & (1,1,2) & (1,2,1) & (2,1,1) \\
    $\alpha$ & 0.20088623 & 0.20734458 & 0.20734458 \\
    $b$ & 0.00021015 &  &  \\
    \hline
  \end{tabular}
\end{table}

\section{Improving Component Values}
\label{sec:component}
\label{sec:component_value}

In \cref{sec:global_value}, we saw how to obtain lower bounds of $\valsix(\CW_q^{\otimes 2^{\lvl}})$ based on the restricted-splitting values from level-$(\lvl - 1)$. In this section, we will apply similar ideas to analyze the restricted-splitting values of components $T_{i,j,k}$ at level-$(\lvl + 1)$, where $i + j + k = 2^{\lvl + 1}$. Our goal is to obtain a lower bound for $\valsix(T_{i,j,k}, \splresZ)$ for some split distribution $\splresZ$ of the Z-index $k$. This will allow us to plug it into the input of the algorithm in \cref{sec:global_value}. We will focus on $T_{i,j,k}$ with non-zero $i,j,k$ in \cref{subsec:component-algo}; when any of $i,j,k$ is zero, we will still apply the ``merging'' approach from the previous works, which will be discussed in \cref{sec:component_value_zero}.

In previous works (e.g., \cite{williams2012, legall2014}), to obtain a lower bound for a component $T_{i,j,k}$'s value, the tensor $\sym_3(T_{i,j,k}^{\otimes n}) = T_{i,j,k}^{\otimes n} \otimes T_{j,k,i}^{\otimes n} \otimes T_{k,i,j}^{\otimes n}$ was analyzed using the laser method. Note that in this tensor, the X/Y/Z-variables are completely symmetric.

Our new approach is based on the asymmetric hashing, where each Z-block is matched with more than one pair of X/Y-blocks. Hence, it is natural to introduce more asymmetry between X/Y and Z-variables. We will analyze a different tensor $\T_{\mathrm{final}}$ other than $\sym_3(T_{i,j,k}^{\otimes n})$ (explained later). But in the end, as we are lower bounding $\valsix(T_{i,j,k}, \splresZ)$, we will have to show that there exists a degeneration $\sym_6(T^{\otimes n}_{i,j,k}[\splresZ]) \degen \T_{\mathrm{final}}$.

\paragraph*{Choosing an Asymmetric Tensor.} To introduce more asymmetry, rather than analyzing $T_{i,j,k}^{\otimes n} \otimes T_{j,k,i}^{\otimes n} \otimes T_{k,i,j}^{\otimes n}$, we will introduce three parameters $A_1, A_2, A_3 \in [0,1]$ such that $A_1 + A_2 + A_3 = 1$, and then analyze
\begin{align*}
  \T_{\mathrm{asym}} \coloneqq{} &T_{i,j,k}^{\otimes A_1 n} \otimes T_{j,k,i}^{\otimes A_2 n} \otimes T_{k,i,j}^{\otimes A_3 n} \otimes {} \\ 
                               & T_{j,i,k}^{\otimes A_1 n} \otimes T_{k,j,i}^{\otimes A_2 n} \otimes T_{i,k,j}^{\otimes A_3 n}.
\end{align*}
Note that in this tensor, there is a significant asymmetry between X/Y and Z-variables. However, we still let X and Y-variables be symmetric, because in the asymmetric hashing, the X and Y-blocks are still matched one-to-one.

\paragraph*{Restricting Split Distributions.} The natural next step would be to apply the laser method and get a lower bound on $V^{(3)}_\tau(\T_{\mathrm{asym}})$. By doing so, we would actually be lower bounding the volume of matrix multiplications that $\T_{\mathrm{final}} = \sym_3(\T_{\mathrm{asym}})$ can degenerate into. However, to obtain a lower bound on $\valsix(T_{i,j,k}, \splresZ)$, it is necessary to ensure that there exists a degneration $ \sym_6(T_{i,j,k}[\splresZ]) \degen \T_{\mathrm{final}} $. This is not possible without restricting the split distribution of $k$'s in $\T_{\mathrm{asym}}$. We will now add these restrictions and obtain another tensor, denoted as $\T$. To proceed, we need the following claim.

\begin{claim} \label{claim:degen}
Let $T_{i,j,k}$ be a level-$\nextlvl$ component where $i + j + k = 2^{\lvl + 1}$. For any $A_1, A_2, A_3 \in [0,1]$ such that $A_1 + A_2 + A_3 = 1$ and split distributions $\splresZ^{[1]}, \splresZ^{[2]}, \splresZ^{[3]}, \splresZ$ satisfying $A_1 \splresZ^{[1]} + A_2 \splresZ^{[2]} + A_3 \splresZ^{[3]} = \splresZ$, there exists a degeneration 
\[
  T_{i,j,k}^{\otimes n}[\splresZ] \;\degen\;  T_{i,j,k}^{\otimes A_1 n}[\splresZ^{[1]}] \otimes T_{i,j,k}^{\otimes A_2 n}[\splresZ^{[2]}] \otimes T_{i,j,k}^{\otimes A_3 n}[\splresZ^{[3]}].
  \numberthis \label{eq:claim_degen}
\]
\end{claim}
\begin{proof}
First, we partition $T_{i,j,k}^{\otimes n}$ into $T_{i,j,k}^{\otimes A_1 n} \otimes T_{i,j,k}^{\otimes A_2 n} \otimes T_{i,j,k}^{\otimes A_3 n}$. Accordingly, we can index each level-$\lvl$ Z-block in $T_{i,j,k}^{\otimes n}$ by three index sequences $K^{(1)}$, $K^{(2)}$, $K^{(3)}$: $K^{(1)} \in \midBK{0, \ldots, k}^{2A_1 n}$ satisfying $K^{(1)}_{2t-1} + K^{(1)}_{2t} = k$ for all $t \in [A_1 n]$ and similar conditions for $K^{(2)}$ and $K^{(3)}$ whose lengths are $2A_2 n$ and $2A_3 n$, respectively. Starting from $T_{i,j,k}^{\otimes n}[\splresZ]$, we zero out all the level-$\lvl$ Z-blocks indexed by $K^{(1)}, K^{(2)}, K^{(3)}$ that $\split(K^{(r)}, [A_r n]) \ne \splresZ^{[r]}$ for at least one $r \in \{1, 2, 3\}$. The obtained tensor equals the RHS of \eqref{eq:claim_degen}, which concludes the proof.
\end{proof}

One immediate corollary of this claim is
\[
  \sym_6(T_{i,j,k}^{\otimes n}[\splresZ]) \;\degen\; \sym_6(T_{i,j,k}^{\otimes A_1 n}[\splresZ^{[1]}]) \otimes \sym_6(T_{i,j,k}^{\otimes A_2 n}[\splresZ^{[2]}]) \otimes \sym_6(T_{i,j,k}^{\otimes A_3 n}[\splresZ^{[3]}]).
  \numberthis \label{eq:cor_degen}
\]

We will choose three (possibly) different Z-split distributions $\splresZ^{[1]}, \splresZ^{[2]}, \splresZ^{[3]}$. If we take exactly the same split distribution for the index $k$ ($k$ is fixed as the Z-index of the component $T_{i,j,k}$ that we analyze), but apply it to X-variables, we will denote them as $\splresX^{[1]}, \splresX^{[2]}, \splresX^{[3]}$. Similarly for Y variables. Under these notations, we define
\begin{align*}
  \T \coloneqq{} & T_{i,j,k}^{\otimes A_1n}[\splresZ^{[1]}] \otimes T_{j,k,i}^{\otimes A_2n}[\splresY^{[2]}] \otimes T_{k,i,j}^{\otimes A_3n}[\splresX^{[3]}] \otimes{} \\
                 & T_{j,i,k}^{\otimes A_1n}[\splresZ^{[1]}] \otimes T_{k,j,i}^{\otimes A_2n}[\splresX^{[2]}] \otimes T_{i,k,j}^{\otimes A_3n}[\splresY^{[3]}].
\end{align*}
Here for example, $T_{j,k,i}^{\otimes A_2 n}[\splresY^{[2]}]$ means restricting $T_{j,k,i}^{\otimes A_2 n}$ to those level-$\lvl$ Y-blocks with split distribution $\splresY^{[2]}$, which is the same distribution as $\splresZ^{[2]}$ but is applied to Y-variables. We choose to add restrictions in this way because of the following claim.

\begin{claim}
  We have
  \begin{gather*}
    \sym_3(T_{i,j,k}^{\otimes A_1 n}[\splresZ^{[1]}] \otimes T_{j,i,k}^{\otimes A_1n}[\splresZ^{[1]}]) \cong \sym_6(T_{i,j,k}^{\otimes A_1 n}[\splresZ^{[1]}]); \\
    \sym_3(T_{j,k,i}^{\otimes A_2 n}[\splresY^{[2]}] \otimes T_{k,j,i}^{\otimes A_2n}[\splresX^{[2]}]) \cong \sym_6(T_{i,j,k}^{\otimes A_2 n}[\splresZ^{[2]}]); \\
    \sym_3(T_{k,i,j}^{\otimes A_3 n}[\splresX^{[3]}] \otimes T_{i,k,j}^{\otimes A_3n}[\splresY^{[3]}]) \cong \sym_6(T_{i,j,k}^{\otimes A_3 n}[\splresZ^{[3]}]).
  \end{gather*}
\end{claim}
\begin{proof}
  For example, the second equation holds because
  \begin{align*}
    \sym_3(T_{j,k,i}^{\otimes A_2 n}[\splresY^{[2]}] \otimes T_{k,j,i}^{\otimes A_2 n}[\splresX^{[2]}])
    \;\cong{}\; & T_{j,k,i}^{\otimes A_2 n}[\splresY^{[2]}] \otimes T_{k,i,j}^{\otimes A_2 n}[\splresX^{[2]}] \otimes T_{i,j,k}^{\otimes A_2 n}[\splresZ^{[2]}] \otimes{} \\
    & T_{k,j,i}^{\otimes A_2 n}[\splresX^{[2]}] \otimes T_{j,i,k}^{\otimes A_2 n}[\splresZ^{[2]}] \otimes T_{i,k,j}^{\otimes A_2 n}[\splresY^{[2]}] \\
    \;\cong{}\; & \sym_6(T_{i,j,k}^{A_2 n}[\splresZ^{[2]}]).
  \end{align*}
  The other two equations follow from similar calculations.
\end{proof}

Together with \eqref{eq:cor_degen}, this proves our desired statement below.

\begin{claim}\label{claim:degen-global}
There exists a generation from $\sym_6(T_{i,j,k}^{\otimes n}[\splresZ])$ to $\sym_3(\T)$. As a result, $\valsix(T_{i,j,k}, \splresZ)^{6n} \geq \valthree(\T)^{3}$.
\end{claim}
To get our desired lower bound on $\valsix(T_{i,j,k}, \splresZ)$, we only have to lower bound $V^{(3)}_\tau(\T)$, or equivalently, degenerate $\T_{\textup{final}} \defeq \sym_3(\T)$ into matrix multiplication tensors.

\paragraph*{Specification.} Formally, the inputs to the algorithm are:
\begin{itemize}
\item Non-negative real numbers $A_1, A_2, A_3$ that sum up to $1$.
\item For each level-$\lvl$ component $T_{i',j',k'}$, we will give three (possibly different) Z-marginal split distributions of $k'$, $\splres^{(1)}_{i',j',k'}, \splres^{(2)}_{i',j',k'}, \splres^{(3)}_{i',j',k'}$, as part of the input, along with the corresponding lower bounds on restrict-splitting values $\valsix(T_{i',j',k'}, \splres^{(1)}_{i',j',k'})$, $\valsix(T_{i',j',k'}, \splres^{(2)}_{i',j',k'})$, $\valsix(T_{i',j',k'}, \splres^{(3)}_{i',j',k'})$.
\item Three Z-marginal split distributions of $k$, namely $\splresZ^{[1]}, \splresZ^{[2]}, \splresZ^{[3]}$, such that
  \[
    A_1 \splresZ^{[1]} + A_2 \splresZ^{[2]} + A_3 \splresZ^{[3]} = \splresZ.
  \]
\end{itemize}
Note that $\splres^{(1)}_{i',j',k'}, \splres^{(2)}_{i',j',k'}, \splres^{(3)}_{i',j',k'}$ specify how $k'$ splits; they are one level lower than the distributions $\splresZ^{[1]}, \splresZ^{[2]}, \splresZ^{[3]}$ which specify how $k$ splits. We are allowing three different split distributions for each level-$\lvl$ component $T_{i',j',k'}$ because it is natural to allows them to split differently when they are in different factors of $\T$.

Our algorithm will optimize the following parameter:
\begin{itemize}
\item Joint split distributions $\alpha^{(1)}, \alpha^{(2)}, \alpha^{(3)}$ of components $(i,j,k), (j,k,i), (k,i,j)$, respectively, such that the Z-marginal of $\alpha^{(1)}$ equals $\splresZ^{[1]}$, the Y-marginal of $\alpha^{(2)}$ equals $\splresY^{[2]}$, and the X-marginal of $\alpha^{(3)}$ equals $\splresX^{[3]}$.
\end{itemize}
In the end, the algorithm will output:
\begin{itemize}
\item A lower bound on the restricted-splitting value $\valsix(T_{i,j,k}, \splresZ)$. Here $T_{i,j,k}$ is a level-$\nextlvl$ component ($i + j + k = 2^{\lvl + 1}$).
\end{itemize}

\subsection{Algorithm Description} \label{subsec:component-algo}

The algorithm in this section follows the same steps as \Cref{sec:global-algo} but with a few small twists. These twists are necessary to adapt our idea to this tensor $\T$ we constructed.

\paragraph*{Step 1: Lower bound the restricted-splitting values of level-$(\lvl - 1)$ components.} As these lower bounds are given in the input, this step is trivial.

\paragraph*{Step 2: Choose distributions.} We will choose the parameters $\alpha^{(1)}, \alpha^{(2)}, \alpha^{(3)}$, which are split distributions of level-$\nextlvl$ components $(i, j, k)$, $(j, k, i)$, and $(k, i, j)$, respectively.

\paragraph*{Step 3: Asymmetric Hashing.} Same has before, we have a modulus $M$ to be specified later. Then we apply the asymmetric hashing in \cref{sec:hashing} to the tensor $\T$. We start by introducing some notations.

Recalling that
\begin{align*}
  \T \coloneqq{} & T_{i,j,k}^{\otimes A_1n}[\splresZ^{[1]}] \otimes T_{j,k,i}^{\otimes A_2n}[\splresY^{[2]}] \otimes T_{k,i,j}^{\otimes A_3n}[\splresX^{[3]}] \otimes {} \\
                 & T_{j,i,k}^{\otimes A_1n}[\splresZ^{[1]}] \otimes T_{k,j,i}^{\otimes A_2n}[\splresX^{[2]}] \otimes T_{i,k,j}^{\otimes A_3n}[\splresY^{[3]}]
\end{align*}
is a subtensor of $(\CW_q^{\otimes 2^{\lvl}})^{\otimes 2n}$, the variable blocks and index sequences of $\T$ are defined in the same way as those in the CW tensor power. Let $I \in \midBK{0, 1, \ldots, 2^{\lvl}}^{4n}$ be a level-$\lvl$ index sequence of $(\CW_q^{\otimes 2^\lvl})^{\otimes 2n}$, identifying an X-block $X_I$. If $X_I$ is present in the subtensor $\T$ of the CW tensor power, we also call it a level-$\lvl$ X-block of $\T$. Same for level-$\lastlvl$ blocks. Similar to \cref{sec:global_value}, we refer to level-$\lvl$ blocks as \emph{large blocks} and level-$\lastlvl$ blocks as \emph{small blocks} throughout this section.

Every level-$\lvl$ index sequence $I$ identifying an X-block of $\T$ can be divided into 6 parts corresponding to 6 factors of $\T$. We refer to them as \emph{regions}:

\begin{definition}[Regions]
  Let $I \subseteq \midBK{0, 1, \ldots, 2^{\lvl}}^{4n}$ be a level-$\lvl$ index sequence of $(\CW_q^{\otimes 2^{\lvl}})^{\otimes 2n}$. We divide $I$ into 6 subintervals which correspond to the 6 factors in $\T$. We call them \emph{regions} of $I$, numbered from 1 to 6. Let $S^{(r)} \subseteq [4n]$ ($r = 1, \ldots, 6$) denote the set of positions within region $r$, i.e.,
  \begin{align*}
    S^{(1)} &\defeq [1, \; 2A_1 n], & S^{(4)} &\defeq [2n + 1, \; 2(1 + A_1)n], \\
    S^{(2)} &\defeq [2A_1n + 1, \; 2(A_1 + A_2) n], & S^{(5)} &\defeq [2(1 + A_1)n + 1, \; 2(1 + A_1 + A_2)n], \\
    S^{(3)} &\defeq [2(A_1 + A_2)n + 1, \; 2n], & S^{(6)} &\defeq [2(1 + A_1 + A_2)n + 1, \; 4n].
  \end{align*}
  Suppose level-$\lvl$ blocks $X_I, Y_J, Z_K$ exist in $\T$. For every region $r$ and position $2t \in S^{(r)}$, we have $(I_{2t-1} + I_{2t}, J_{2t-1} + J_{2t}, K_{2t-1} + K_{2t}) = \bk{i^{(r)}, j^{(r)}, k^{(r)}}$, where
  \begin{align*}
    \bigbk{i^{(1)}, j^{(1)}, k^{(1)}} &\defeq (i, j, k), & \bigbk{i^{(4)}, j^{(4)}, k^{(4)}} &\defeq (j, i, k), \\
    \bigbk{i^{(2)}, j^{(2)}, k^{(2)}} &\defeq (j, k, i), & \bigbk{i^{(5)}, j^{(5)}, k^{(5)}} &\defeq (k, j, i), \\
    \bigbk{i^{(3)}, j^{(3)}, k^{(3)}} &\defeq (k, i, j), & \bigbk{i^{(6)}, j^{(6)}, k^{(6)}} &\defeq (i, k, j).
  \end{align*}
\end{definition}

The last 3 regions are obtained from the first 3 regions by swapping the order of X and Y dimensions, so we naturally define the following:

\begin{definition}[$\alpha^{(4)}, \alpha^{(5)}, \alpha^{(6)}$ and $A_4, A_5, A_6$.]
  Recall the parameters of the algorithm $\alpha^{(1)}, \alpha^{(2)}, \alpha^{(3)}$ are joint split distributions of components $(i,j,k)$, $(j,k,i)$, and $(k,i,j)$, respectively. We define $\alpha^{(4)}, \alpha^{(5)}, \alpha^{(6)}$ to be split distributions of $(j,i,k)$, $(k,j,i)$, $(i,k,j)$ obtained by swapping the X and Y dimensions of $\alpha^{(1)}, \alpha^{(2)}, \alpha^{(3)}$, respectively. Formally, $\alpha^{(r+3)}(i', j', k') = \alpha^{(r)}(j', i', k')$. We also let $A_{r + 3} \defeq A_r$ and $\splres^{(r + 3)}_{i', j', k'} \defeq \splres^{(r)}_{i', j', k'}$ for $r = 1, 2, 3$.
\end{definition}

$\T$ contains many large triples $(X_I, Y_J, Z_K)$. By applying the asymmetric hashing method, we retain some large triples that do not share X or Y-blocks. Moreover, every retained triple $(X_I, Y_J, Z_K)$ is consistent with the split distribution $\alpha^{(r)}$ in all regions $r \in [6]$. Formally:

\begin{definition}[Consistency with distributions]
  \label{def:consistency-component}
  We say a large (level-$\lvl$) triple $(X_I, Y_J, Z_K)$ is consistent with the split distributions $\alpha^{(r)}$, if for every region $r \in [6]$ and level-$\lvl$ component $(i', j', k')$, there is
  \[
    \frac{1}{|S^{(r)}|/2}|\{2t \in S^{(r)}, t \in \Z \mid (I_{2t-1}, J_{2t-1}, K_{2t-1}) = (i', j', k')\}| \;=\; \alpha^{(r)}(i', j', k').
  \]
  Moreover, we say a large X-block $X_I$ is consistent with marginal split distributions $\alphx^{(r)}$, if for every region $r \in [6]$ and $i' \in [0, 2^{\lvl}]$, there is
  \[
    \frac{1}{|S^{(r)}|/2}\BK{2t \in S^{(r)}, t \in \Z \mid I_{2t-1} = i'} = \alphx^{(r)}(i').
  \]
  Same for Y and Z-blocks.
\end{definition}

(Similar to \Cref{sec:global_value}, to ensure such consistency with joint distribution, we may have to suffer hash loss when the joint distributions $\alpha^{(1)},\alpha^{(2)},\alpha^{(3)}$ are not the maximum entropy distributions given their marginals. We will take such hash loss into account in \Cref{sec:component-analysis}.)

\paragraph*{Additional Zeroing-Out Step 1.}
To state this step, we need to first define compatibility. Recall that a small (level-$\lastlvl$) X-block is identified by an index sequence $\hat I \in \midBK{0, \ldots, 2^{\lvl - 1}}^{8n}$. Denote this small block by $X_{\hat I}$, and similarly, $Y_{\hat J}$ and $Z_{\hat K}$ for small Y and Z-blocks. The meaning of compatibility is the same as \cref{sec:global_value}; however, since different regions have different split distributions, we have to discuss different regions separately.

Fix a triple $(X_I, Y_J, Z_K)$. For region $r \in [6]$, we use $S^{(r)}_{i', j', k'}$ to denote the set of positions $t \in S^{(r)}$ such that $(I_t, J_t, K_t) = (i', j', k')$; use $S^{(r)}_{*, *, k'}$ to denote the set of positions $t \in S^{(r)}$ such that $K_t = k'$. For a small block $Z_{\hat K} \in Z_K$ and a subset $S \subseteq S^{(r)}_{*, *, k'}$ for some $k'$, we use $\split_{k'}(\hat K, S)$ to denote the distribution of $\hat K_{2t-1}$ over all $t \in S$. Since for each $t \in S$, the index $K_t$ splits into $\hat K_{2t-1} + \hat K_{2t}$, this distribution $\split_{k'}(\hat K, S)$ captures the distribution of how those $K_t$ split for $t \in S$. We omit the index $k'$ when it is clear from the context and simply write $\split(\hat{K}, S)$.

For each $k'$, we define the following average split distribution:
\begin{gather} \label{equ:avg-alph}
  \splresavg_{*,*,k'}^{(r)} = \frac{1}{\alphz^{(r)}(k') + \alphz^{(r)}(k^{(r)} \!-\! k')} \cdot \sum_{i' + j' = 2^\lvl - k'} \bk{\alpha^{(r)}(i', j', k') + \alpha^{(r)}(\complementshape)} \cdot \splres^{(r)}_{i',j',k'}.
\end{gather}

\begin{definition}[Compatibility]
  A small block $Z_{\hat K} \in Z_K$ is said to be \emph{compatible} with a large triple $(X_I, Y_J, Z_K)$ if the following two conditions are satisfied:
  \begin{enumerate}
  \item{\label{item:component-compat}} For each $k' \in \midBK{0, \ldots, 2^\lvl}$ and $r \in [6]$, $\split(\hat K, S^{(r)}_{*, *, k'}) = \splresavg^{(r)}_{*, *, k'}$.
  \item For each level-$\lvl$ component $(i', j', k')$ with $i' = 0$ or $j' = 0$, and for all $r \in [6]$, $\split(\hat K, S^{(r)}_{i',j',k'}) = \splres^{(r)}_{i',j',k'}$.
  \end{enumerate}
\end{definition}

\noindent
Based on this definition, we do the following zeroing out on small blocks:
\begin{itemize}
\item For each $Z_{\hat K}$, we check \cref{item:component-compat} and zero out $Z_{\hat K}$ if the condition is not satisfied.
\item For each $X_{\hat{I}} \in X_I$, since $X_I$ (if retained) is in a unique triple $(X_I, Y_J, Z_K)$, we can define our $S^{(r)}_{i',j',k'}$ w.r.t. that triple. For all level-$\lvl$ components $(i',0,k')$, we define
  \[
    \splres^{(r,\textup{X})}_{i',0,k'}(i'_l) \;\defeq\; \splres^{(r)}_{i',0,k'}(2^{\lvl-1} - i'_l).
  \]
  $\splres^{(r, \textup{X})}$ is the X-split distribution corresponding to the Z-split distribution $\tilde\alpha^{(r)}_{i',0,k'}$. If for any $(i',0,k')$ and $r \in [6]$, $\split(\hat I, S^{(r)}_{i', 0, k'}) \ne \splres^{(r,\textup{X})}_{i', 0, k'}$, we then zero out this $X_{\hat{I}}$.
\item For each $Y_{\hat{J}}$, we perform a similar zeroing-out as with $X_{\hat{I}}$. For all level-$\lvl$ components $(0,j',k')$, we define
  \[
    \splres^{(r,\textup{Y})}_{0,j',k'}(j'_l) \;\defeq\; \splres^{(r)}_{0,j',k'}(2^{\lvl-1} - j'_l).
  \]
  Suppose for some component $(0, j', k')$ and region $r \in [6]$, $\split(\hat J, S^{(r)}_{0, j', k'}) \neq \splres^{(r,\textup{Y})}_{0,j',k'}$, we then zero out $Y_{\hat{J}}$.
\end{itemize}

Same as \cref{sec:global_value}, it is easy to verify the following claim:
\begin{claim}
  \label{lemma:triple_implies_compatible_component}
  For our constructed tensor $\T$, after Additional Zeroing-Out Step 1, a remaining small block $Z_{\hat{K}}$ can form a small triple with $X_{\hat{I}} \in X_I$, $Y_{\hat{J}} \in Y_J$ only when $Z_{\hat{K}}$ is compatible with the large triple $(X_I, Y_J, Z_K)$.
\end{claim}
We use the notation $\mathcal T^{(1)}$ to denote the tensor after Additional Zeroing-Out Step 1.

\paragraph*{Additional Zeroing-Out Step 2.} Similar to \cref{sec:global_value}, we make the following definition.

\begin{definition}
  \label{def:useful_c}
  A small block $Z_{\hat K}$ is said to be \emph{useful} for a large triple $(X_I, Y_J, Z_K)$ if the following conditions are met:
  \begin{itemize}
  \item $Z_{\hat K} \in Z_K$, and $(X_I, Y_J, Z_K)$ is consistent with $\alpha^{(r)}$ for all $r \in [6]$.
  \item For each level-$\lvl$ component $(i',j',k')$ and region $r \in [6]$, $\split(\hat K, S^{(r)}_{i',j',k'}) = \tilde\alpha^{(r)}_{i',j',k'}$.
  \end{itemize}
\end{definition}

\noindent
Based on this definition, we zero out any small Z-block $Z_{\hat{K}} \in Z_K$ such that
\begin{itemize}
\item $Z_{\hat{K}}$ is compatible with more than one triple, or
\item $Z_{\hat{K}}$ is not useful for the unique triple $(X_I, Y_J, Z_K)$ that it is compatible with.
\end{itemize}
After such zeroing out, we call the obtained tensor $\T^{(2)}$. We know it is the direct sum of disjoint triples, i.e., $\T^{(2)} = \bigoplus_{(X_I, Y_J, Z_K)} \T^{(2)}\vert_{X_I, Y_J, Z_K}$ (due to the first rule above).

For any retained triple $(X_I, Y_J, Z_K)$, in the ideal case where no small blocks $Z_{\hat{K}}$ are zeroed out due to the first rule mentioned above, we know $\T^{(2)} \vert_{X_I, Y_J, Z_K}$ is isomorphic to
\begin{align*}
  \T^* \defeq{} &(\T^*)^{(1)} \otimes (\T^*)^{(2)} \otimes (\T^*)^{(3)} \otimes{} \\
                &(\T^*)^{(4)} \otimes (\T^*)^{(5)} \otimes (\T^*)^{(6)}, \quad \textup{where}
\end{align*}
\[
  (\T^*)^{(r)} \defeq
  \bigotimes_{i' + j' + k' = 2^\lvl} T_{i',j',k'}^{\otimes (\alpha^{(r)}(i', j', k') + \alpha^{(r)}(\complementshape)) \cdot A_rn} \bigBk{\splres^{(r)}_{i',j',k'}}.
\]
For ease of notation, we define
\[
  \beta^{(r)}(i',j',k') \;\defeq\; \frac{1}{2}\left(\alpha^{(r)}(i', j', k') + \alpha^{(r)}(\complementshape) \right).
\]
It represents the proportion of level-$\lvl$ components $(i', j', k')$ within each region, i.e., $|S^{(r)}_{i',j',k'}| = 2A_r n \cdot \beta^{(r)}(i', j', k')$. Then, we can rewrite
\[
  \T^* \;=\; \bigotimes_{r=1}^6 \bigotimes_{i' + j' + k' = 2^\lvl} T_{i',j',k'}^{\otimes \beta^{(r)}(i', j', k') \cdot 2A_rn} \bigBk{\splres^{(r)}_{i',j',k'}}.
  \numberthis \label{eq:def_tstar_c}
\]
One can see that $\T^*$ matches \cref{def:broken_standard_form_tensor} with parameters
\[
  \BK{2 A_r n \cdot \beta^{(r)}(i', j', k'), \, i', j', k', \, \tilde\alpha^{(r)}_{i',j',k'}}_{i'+j'+k'=2^\lvl, \; r\in[6]}.
\]
Hence we can apply the hole lemma in \cref{sec:hole_lemma} to fix the holes in $\T^{(2)} \vert_{X_I, Y_J, Z_K}$ (similar to \cref{sec:global_value}, the first zeroing-out rule above produces holes in Z-blocks of $\T^{(2)} \vert_{X_I, Y_J, Z_K}$).

\paragraph*{Step 4: Fix the holes and degenerate each triple independently.}

For each retained triple $(X_I, Y_J, Z_K)$, $\T^{(2)} \vert_{X_I, Y_J, Z_K}$ is a broken copy of the standard form tensor $\T^*$. The \emph{fraction of non-holes} in $\T^{(2)} \vert_{X_I, Y_J, Z_K}$, as defined in \cref{def:broken_standard_form_tensor}, is denoted as $\fracnonholeIJK$. Letting
\[
  m' \defeq \floor{\frac{\sum_{(X_I, Y_J, Z_K)} \fracnonholeIJK}{\bk{\sum_{i' + j' + k' = 2^{\lvl}, \; r \in [6]} 2 A_r n \cdot \beta^{(r)}(i', j', k')} \cdot \lvl + 2}} = \floor{\frac{\sum_{(X_I, Y_J, Z_K)} \fracnonholeIJK}{4n \lvl + 2}},
\]
it follows from \cref{cor:hole_lemma} that
\[
  \bigoplus_{(X_I, Y_J, Z_K) \textup{ remaining}} \mathcal T^{(2)} \vert_{X_I, Y_J, Z_K} \; \degen \; (\mathcal T^*)^{\oplus m'}.
\]

Based on the given bounds on restricted-splitting values of level-$\lvl$ components, \cref{eq:def_tstar_c} implies that $\valsix(\T^*) \ge \prod_{r=1}^6 \prod_{i' + j' + k' = 2^{\lvl}} \valsix(T_{i', j', k'}, \splres^{(r)}_{i', j', k'})^{2A_r n \cdot \beta^{(r)}(i', j', k')}$. As conclusion,
\begin{equation}
  \label{eq:result_mid_c}
  \valthree(\T) \;\ge\; v \;\defeq\; m' \cdot \prod_{r=1}^6 \prod_{i' + j' + k' = 2^{\lvl}} \valsix(T_{i', j', k'}, \, \splres^{(r)}_{i', j', k'})^{2A_r n \cdot \beta^{(r)}(i', j', k')}.
\end{equation}
It is worth noting that $\valthree(\T) = \valsix(\T)$ because $\T$ is symmetric about X and Y variables, i.e., $\sym_3(\T)^{\otimes 2} \cong \sym_6(\T)$.

\subsection{Analysis}
\label{sec:component-analysis}

The idea of analysis is again similar to \cref{sec:global_value}.

\paragraph*{Asymmetric Hashing.} First, let us specify our notations.
\begin{itemize}
\item $\numxblock = \numyblock \ge \numzblock$ represent the number of large (level-$\lvl$) X, Y, and Z-blocks that are consistent with $\alphx^{(r)}, \alphy^{(r)}, \alphz^{(r)}$ in all regions $r \in [6]$, respectively. We have $\numxblock = 2^{\bk{\sum_{r=1}^6 A_r n H(\alphx^{(r)}) + o(n)}}$ (similar for Y and Z).
\item $\numalpha$ is the number of triples $(X_I, Y_J, Z_K)$ that are consistent with $\alpha^{(r)}$ in all regions $r \in [6]$. We have $\numalpha = 2^{\bk{\sum_{r=1}^6 A_r n H(\alpha^{(r)}) + o(n)}}$. $N_{\alphx, \alphy, \alphz}$ is the number of triples $(X_I, Y_J, Z_K)$ whose marginal distributions are consistent with $\alphx^{(r)}, \alphy^{(r)}, \alphz^{(r)}$, respectively.
\item $\numretain$ represents the number of retained triples after the asymmetric hashing process.
\item Let $\pcomp$ be a parameter to be defined later. Roughly speaking, it represents the probability of a small block $Z_{\hat K}$ being compatible with a random triple.
\end{itemize}

Similar to \Cref{sec:global_analysis}, we will let $M_0 = 8 \cdot \max\bk{\frac{\numtriple}{\numxblock}, \frac{\numalpha \cdot \pcomp}{\numzblock}}$ and let $M \in [M_0, 2M_0]$ be a prime. Then we apply the asymmetric hashing with modulus $M$. We know
\[
  \E[\numretain] \ge \frac{\numalpha}{M} \cdot 2^{-o(n)} = \frac{\numalpha}{M_0} \cdot 2^{-o(n)} = \min\bk{\frac{\numalpha \cdot \numxblock}{\numtriple}, \, \frac{\numzblock}{\pcomp}} \cdot 2^{-o(n)}.
  \numberthis \label{eq:sec7_numretain}
\]

\paragraph*{Typical distribution.} Recall that in \Cref{sec:global_analysis}, we defined the typical distribution
\[\gamma(k_l, k_r) \coloneqq \sum_{\substack{i + j + k = 2^{\ell} \\ k = k_l + k_r}} \alpha(i,j,k) \cdot \splres_{i,j,k}(k_l),\]
essentially by composing the joint distribution $\alpha(i,j,k)$ of level-$\ell$ components and their Z-split distributions $\splres_{i,j,k}$.

In this section, for each region, we have a level-$\lvl$ split distribution $\alpha^{(r)}(i',j',k')$. For each region $r \in [6]$ and level-$\lvl$ component $(i',j',k')$, we have a split distribution $\splres^{(r)}_{i,j,k}$. Composing these two quantities gives the following distribution:
\[
  \gamma^{(r)}(k_1, k_2, k_3, k_4) \;\defeq \sum_{\substack{i' + j' + k' = 2^{\lvl} \\ k' = k_1 + k_2}} \alpha^{(r)}(i', j', k') \cdot
  \splres_{i',j',k'}^{(r)}(k_1) \cdot \tilde\alpha_{i^{(r)}\!-i',\; j^{(r)}\!-j',\; k^{(r)}\!-k'}^{(r)}(k_3).
  \numberthis \label{eq:typical_def}
\]
We call it the \emph{typical distribution} $\gamma^{(r)}$ in region $r$. It captures how a single (level-$\nextlvl$) index $k^{(r)}$ in the region $r$ of the index sequence is (typically) split into four level-$\lastlvl$ indices $k_1 + k_2 + k_3 + k_4$.

For a small block $Z_{\hat K}$, if $\gamma^{(r)}$ matches the frequency of occurrence of $(\hat K_{4t-3}, \hat K_{4t-2}, \hat K_{4t-1}, \hat K_{4t})$ in region $r$, we say $Z_{\hat K}$ is a \emph{typical block}. Formally, $Z_{\hat K}$ is typical if and only if for every region $r \in [6]$ and $k_1 + k_2 + k_3 + k_4 = k^{(r)}$,
\[
  \abs{\BK{t \in \mathbb{Z}, \, 2t \in S^{(r)} \mymiddle \bigbk{\hat K_{4t-3}, \hat K_{4t-2}, \hat K_{4t-1}, \hat K_{4t}} = (k_1, k_2, k_3, k_4)}} = \gamma(k_1, k_2, k_3, k_4) \cdot \frac{|S^{(r)}|}{2}.
\]
Denote by $\typicalset$ the set of typical blocks $Z_{\hat K}$ within a fixed large block $Z_K$.

Recall the definition of \emph{useful blocks} (\cref{def:useful_c}). We denote by $\usefulset$ the set of small blocks $Z_{\hat K} \in Z_K$ useful for $(X_I, Y_J, Z_K)$. It is worth noting that, fixing a large triple $(X_I, Y_J, Z_K)$, the typicalness and usefulness of a small block $Z_{\hat K} \in Z_K$ do not imply each other. The following lemma shows that typical blocks make up a non-negligible part of $\usefulset$, as expected:

\begin{restatable}{lemma}{typicalnonneg}
  \label{lemma:typical_is_nonnegl}
  Let $(X_I, Y_J, Z_K)$ be a triple consistent with $\alpha^{(r)}$ in all regions $r$. Then,
  \[
    \abs{\usefulset \cap \typicalset} \ge \abs{\usefulset} \cdot 2^{-o(n)}.
  \]
\end{restatable}
We defer its proof to \cref{appendix:missing_proof_7}.

\paragraph{Probability of being compatible.} Throughout the rest of this subsection, if a triple $(X_I, Y_J, Z_K)$ is consistent with $\alpha^{(r)}$ in all regions $r \in [6]$, we say these block $X_I, Y_J, Z_K$ are \emph{matchable} to each other. Next, we define $\pcomp$ similarly to \cref{sec:global_analysis}.

\begin{definition}
  Suppose $Z_{\hat K} \in Z_K$ is a typical block, and $X_I$ is a large X-block matchable to $Z_K$ chosen uniformly at random. $\pcomp$ is defined as the probability of $Z_{\hat K}$ being compatible with $X_I, Y_J$. Due to symmetry, $\pcomp$ is independent of the chosen small block $Z_{\hat K}$ (as long as it is typical).
\end{definition}

We calculate $\pcomp$ by the following lemma.

\begin{lemma} \label{lemma:pcomp_g_c}
Let $\gamma^{(r)}$ be the typical distribution we defined over $\{0,1,\dots,2^{\ell - 1}\}^4$. Let
\[
  \beta^{(r)}(\text{+},\text{+},k') = \sum_{\substack{i',j' > 0 \\i'+j'+k'=2^\lvl}} \beta^{(r)}(i',j',k'), \qquad \tilde\alpha^{(r)}_{\text{+},\text{+},k'} = \frac{1}{\beta^{(r)}(\text{+},\text{+},k')} \cdot \sum_{\substack{i', j' > 0 \\ i'+j'+k'=2^\lvl}} \beta^{(r)}(i',j',k') \cdot \tilde\alpha^{(r)}_{i',j',k'}.
\]
Then, we have
\[\pcomp \le \alphap^{2n} \cdot 2^{o(n)}, \quad \text{where} \quad \alphap \defeq \prod_{r=1}^6 \bigbk{\alphap^{(r)}}^{A_r / 2} = \prod_{r=1}^3 \bigbk{\alphap^{(r)}}^{A_r}, \quad \text{and}\]
\[\alphap^{(r)} \defeq 2^{\bigbk{H(\alphz^{(r)}) - H(\gamma^{(r)})}} \cdot  {\prod_{\substack{i' + j' + k' = 2^\lvl \\ i' = 0 \textup{ or } j' = 0}}} \; 2^{2\beta^{(r)}(i', j', k') \cdot H\bigbk{\splres^{(r)}_{i',j',k'}}} \cdot {\prod_{k'=0}^{2^\lvl}} \; 2^{2\beta^{(r)}(\plusplusk') \cdot H\bigbk{\splres^{(r)}_{\plusplusk'}}}.\]
\end{lemma}
\begin{proof}
  We fix $Z_K$ as a large block consistent with $\alphz^{(r)}$ in all regions $r$, and let $Z_{\hat K} \in \typicalset$ be a uniformly random typical block in $Z_K$. Since $\pcomp$ is the same for all typical blocks, it also has the same value for the random block $Z_{\hat K}$. As in the definition of $\pcomp$, we let $X_{I}$ be a random large X-block that is matchable to $Z_K$, which is independent of $Z_{\hat K}$ conditioned on $Z_K$. We have
\begin{align*}
  \pcomp
  &= \Pr_{\hat K,\; I} \Bk{Z_{\hat K} \textup{ is compatible with } X_{I}} \\
  &= \E_{I} \Bk{\Pr_{\hat K} \Bk{Z_{\hat K} \textup{ is compatible with } X_{I}}} \\
  &= \E_{I} \Bk{\frac
    { \left|\left\{Z_{\hat K'} \in \typicalset \ \middle\vert \ Z_{\hat K'} \text{ is compatible with } X_{I} \right\}\right|}
    { \abs{\typicalset} }} \\
  &\le \E_{I} \Bk{\frac
    { \left|\left\{Z_{\hat K'} \in Z_K \ \middle\vert \ Z_{\hat K'} \text{ is compatible with } X_{I} \right\}\right| }
    { \abs{\typicalset} }}.
\end{align*}
The content inside the expectation is identical for all $X_I$ due to symmetry, so we arbitrarily fix an $X_I$ and continue the calculation.

\paragraph*{Numerator.} We now count the number of (not necessarily typical) small blocks $Z_{\hat{K}}$ compatible with $X_{I'}$. Recall that $S^{(r)}_{i',j',k'}$ denotes the set of positions $t \in S^{(r)} \subset [2n]$ where $(I_t, J_t, K_t) = (i',j',k')$. Let
\[
  S^{(r)}_{*,*,k'} = \bigcup_{i'+j'=2^\lvl - k'} S^{(r)}_{i',j',k'}, \qquad S^{(r)}_{\plusplusk'} = \bigcup_{\substack{i'+j'=2^\lvl - k' \\ i',j' \neq 0}} S^{(r)}_{i',j',k'}.
\]
$Z_{\hat{K}}$ is compatible with $X_{I}$ if and only if:
\begin{enumerate}[label=(\alph*)]
\item For $i' + j' + k' = 2^\lvl$ where $i' = 0$ or $j' = 0$, and for all regions $r \in [6]$, $\split(\hat K, S^{(r)}_{i', j', k'}) = \splres^{(r)}_{i', j', k'}$.
\item For $k' \in [0, 2^{\lvl}]$ and $r \in [6]$, $\split(\hat K, S^{(r)}_{*, *, k'}) = \splresavg^{(r)}_{*, *, k'}$.
\end{enumerate}
(Recall that $\splresavg^{(r)}_{*, *, k'}$ is defined in \eqref{equ:avg-alph}.)
Similar to \cref{sec:global_analysis}, we transform these two conditions into the following equivalent conditions:
\begin{enumerate}[label=(\alph*)]
\item For $i' + j' + k' = 2^\lvl$ where $i' = 0$ or $j' = 0$, and for all regions $r \in [6]$, $\split(\hat K, S^{(r)}_{i', j', k'}) = \splres^{(r)}_{i', j', k'}$.
\item[(c)] For $k' \in [0, 2^{\lvl}]$ and $r \in [6]$, $\split(\hat K, S^{(r)}_{\plusplusk'}) = \splresavg^{(r)}_{\plusplusk'}$.
\end{enumerate}
All requirements of these two types are of the form $\split(\hat{K}, S) = \splres$, for which there are
$$\binom{|S|}{|S|\splres(0), |S|\splres(1), \dots, |S| \splres(2^{\lvl - 1})} = 2^{|S| \cdot H(\splres) + o(|S|)}$$
way to split those $K_t \; (t \in S)$ into $(\hat{K}_{2t - 1}, \hat{K}_{2t})$. Besides, the position sets of all these requirements, $S^{(r)}_{i', j', k'}$ for the first type and $S^{(r)}_{\plusplusk'}$ for the second type, are disjoint. (In fact, these position sets form a partition of $[2n]$.)

Note that $|S^{(r)}_{i',j',k'}| = \beta^{(r)}(i',j',k') \cdot 2A_rn$ and $|S^{(r)}_{\plusplusk'}| = \beta^{(r)}(\plusplusk') \cdot 2 A_r n$. Multiplying the number of ways together for all requirements, we obtain
  \begin{align*}
    & \phantom{{}={}} \abs{\BK{Z_{\hat K'} \in Z_K \;\middle|\; Z_{\hat K'} \textup{ is compatible with } X_I}} \\
    &= \prod_{r = 1}^6 {\prod_{\substack{i' + j' + k' = 2^\lvl \\ i' = 0 \textup{ or } j' = 0}}} \; 2^{\bigabs{S^{(r)}_{i', j', k'}} \cdot H\bigbk{\splres^{(r)}_{i',j',k'}}} \cdot {\prod_{k'=0}^{2^\lvl}} \; 2^{\bigabs{S^{(r)}_{\plusplusk'}} \cdot H\bigbk{\splresavg^{(r)}_{\plusplusk'}}} \cdot 2^{o(n)} \\
    &= \prod_{r = 1}^6 {\prod_{\substack{i' + j' + k' = 2^\lvl \\ i' = 0 \textup{ or } j' = 0}}} \; 2^{A_rn \cdot 2\beta^{(r)}(i', j', k') \cdot H\bigbk{\splres^{(r)}_{i',j',k'}}} \cdot {\prod_{k'=0}^{2^\lvl}} \; 2^{A_rn \cdot 2 \beta^{(r)}(\plusplusk') \cdot H\bigbk{\splresavg^{(r)}_{\plusplusk'}}} \cdot 2^{o(n)}. \label{eq:lambda_denominator_g}
  \end{align*}

\paragraph*{Denominator.} To calculate the denominator $\abs{\typicalset}$, like in \Cref{sec:global}, we use the fact that for all $K$ that are consistent with $\alphz^{(r)}$ for all $r \in [6]$, the sets $\abs{\typicalset}$ are of the same size.
\begin{align*}
    \abs{\typicalset} &= \frac{\abs{\bigcup_{K' \textup{ consistent with } \alphz} \typicalset[K']}}{\numzblock} \\
                      &= \prod_{r=1}^6 \left. \binom{A_r \cdot n}{[A_r n \cdot \gamma^{(r)}(k_1, k_2, k_3, k_4)]_{k_1,k_2,k_3, k_4}} \middle/ \binom{A_r \cdot n}{[A_r n \cdot \alphz^{(r)}(k')]_{k'}} \right. \\
                      &= \prod_{r=1}^6 2^{\bigbk{H(\gamma^{(r)}) - H(\alphz^{(r)})} \cdot A_r \cdot n}  \cdot 2^{o(n)}.
\end{align*}
Putting these two results together and combining with $\alphap^{(r)} = \alphap^{(r+3)}$ ($r = 1, 2, 3$) due to symmetry, we conclude the proof of this lemma.
\end{proof}

\paragraph{Probability of being holes.} Fixing a small block $Z_{\hat K}$ that is \emph{useful} for some retained triple $(X_I, Y_J, Z_K)$ (see \cref{def:useful_c}), we analyze the probability of $Z_{\hat K}$ being a hole. There are two cases:
\begin{itemize}
\item $Z_{\hat K}$ is not typical, in which case we apply the trivial bound $\Pr[Z_{\hat K} \textup{ is a hole}] \le 1$;
\item $Z_{\hat K}$ is typical. We analyze it below.
\end{itemize}
According to \cref{lemma:typical_is_nonnegl}, at least $2^{-o(n)}$ fraction of the useful blocks are typical. For every typical block $Z_{\hat K}$, we use an approach similar to \cref{sec:global_analysis} to bound its probability of being holes:

\begin{claim}
  \label{claim:hole_frac_low_c}
  Fixing a retained triple $(X_I, Y_J, Z_K)$ (it must be consistent with $\alpha^{(r)}$ in all regions $r \in [6]$) and a small typical block $Z_{\hat K} \in \typicalset$ useful for that triple, the probability of $Z_{\hat K}$ being a hole in $\T^{(2)} \vert_{X_I, Y_J, Z_K}$ (i.e., being compatible with a different remaining triple $(X_{I'}, Y_{J'}, Z_K)$) is at most $1/8$.
\end{claim}

\begin{proof}
  A necessary condition of $Z_{\hat K}$ being a hole is that there exists $X_{I'} \ne X_I$ matchable to $Z_K$ such that (1) $Z_{\hat K}$ is compatible with $X_{I'}$; and (2) $I'$ is hashed to the same slot as $K$, i.e., $\hashx(I') = \hashz(K)$. We calculate the expected number of such $I'$ to establish an upper bound on the probability of the existence of such $I'$:
  \begin{align*}
    \Pr\Bk{Z_{\hat K} \textup{ is a hole}}
    &\le \sum_{I' \ne I} \ind\Bk{Z_{\hat K} \textup{ is compatible with } X_{I'}} \cdot \Pr[\hashx(I') = \hashz(K) \mid \hashx(I) = \hashz(K)] \\
    &= \sum_{I' \ne I} \ind\Bk{Z_{\hat K} \textup{ is compatible with } X_{I'}} \cdot \frac{1}{M} \\
    &< \sum_{I'} \ind\Bk{Z_{\hat K} \textup{ is compatible with } X_{I'}} \cdot \frac{1}{M} \\
    &= \frac{\numalpha}{\numzblock} \cdot \Pr_{I' \textup{ matchable to } K} [Z_{\hat K} \textup{ is compatible with } X_{I'}] \cdot \frac{1}{M} \\
    &= \frac{\numalpha \cdot \pcomp}{\numzblock \cdot M} \;\le\; \frac{\numalpha \cdot \pcomp}{\numzblock \cdot M_0} \;\le\; \frac{\numalpha \cdot \pcomp}{\numzblock} \cdot \frac{\numzblock}{8 \cdot \numalpha \cdot \pcomp} \;=\; \frac{1}{8},
  \end{align*}
  where the first equality above holds due to \cref{lemma:hash_independence}; the third inequality holds according to the definition of $\pcomp$ above. It is worth noting that $\pcomp$ is defined for typical blocks $Z_{\hat K}$ while this is also a premise of the current claim~\ref{claim:hole_frac_low_c}.
\end{proof}

Recall that for each of the $\numretain$ retained triples $(X_I, Y_J, Z_K)$, the fraction of non-hole blocks in $\T^{(2)} \vert_{X_I, Y_J, Z_K}$ is represented by $\fracnonholeIJK$. That is, the fraction of useful blocks in $\usefulset$ not being a hole. \Cref{lemma:typical_is_nonnegl} tells that at least $2^{-o(n)}$ fraction of the useful blocks are typical, each of which has $\Omega(1)$ probability not to be a hole due to \Cref{claim:hole_frac_low_c}. Thus, we conclude that $\fracnonholeIJK \ge 2^{-o(n)}$.

\paragraph*{Bounding the value.} We will now obtain the bound for $\valsix(T_{i,j,k}, \splresZ)$. By \Cref{claim:degen-global}, we know that $\valsix(T_{i,j,k}, \splresZ)^{6n} \geq \valthree(\T)^{3}$, i.e., $\valsix(T_{i,j,k}, \splresZ)^{2n} \geq \valthree(\T)$. In our algorithm, we degenerated $\T$ into tensor $\T^{(2)}$, in which every triple $\T^{(2)}\vert_{X_I, Y_J, Z_K}$ is a broken copy of a standard form tensor $\T^*$. The fraction of non-hole blocks in $\T^{(2)}\vert_{X_I, Y_J, Z_K}$ is $\fracnonholeIJK \ge 2^{-o(n)}$ as shown above.

Same as \Cref{sec:global}, according to \cref{cor:hole_lemma}, $\T^{(2)}$ can degenerate into
\[
  m' \defeq \floor{\frac{\sum_{(X_I, Y_J, Z_K)} \fracnonholeIJK}{4n \lvl + 2}}
  \ge \floor{\frac{\numretain \cdot 7}{8(4n \lvl + 2)}} = \frac{\numretain}{O(n)}
\]
many copies of standard form tensors $\T^*$. Together with \eqref{eq:sec7_numretain}, we know
\[
\E[m'] \geq \min\bk{\frac{\numalpha \cdot \numxblock}{\numtriple}, \, \frac{\numzblock}{\pcomp}} \cdot 2^{-o(n)}
\]
and
\begin{align*}
  & \phantom{{} \ge {}} \valsix(T_{i,j,k}, \splresZ)^{2n} \geq \valthree(\T) \\
&\geq  \min\bk{\frac{\numalpha \cdot \numxblock}{\numtriple}, \, \frac{\numzblock}{\pcomp}} \cdot \prod_{r = 1}^6 \prod_{i' + j' + k' = 2^\lvl} \valsix(T_{i',j',k'}, \tilde\alpha^{(r)}_{i', j', k'})^{2 A_r n \beta^{(r)}(i', j', k')} \cdot 2^{-o(n)}. \numberthis \label{eq:value_bound_c}
\end{align*}
We define the following quantities.
\begin{itemize}
\item $\alphabx^{(r)} \defeq 2^{H\bigbk{\alphx^{(r)}}}$, and similarly $\alphaby^{(r)} \defeq 2^{H\bigbk{\alphy^{(r)}}}$, $\alphabz^{(r)} \defeq 2^{H\bigbk{\alphz^{(r)}}}$. Then,
  \begin{align*}
    &\alphabx \;\defeq\; \lim_{n \to \infty} \numxblock^{1/2n} \;=\; \prod_{r = 1}^6 \bigbk{\alphabx^{(r)}}^{A_r/2} \;=\; \prod_{r = 1}^3 \sqrt{\alphabx^{(r)} \cdot \alphaby^{(r)}}^{A_r} \\
    \textup{and} \qquad &\alphabz \;\defeq\; \lim_{n \to \infty} \numzblock^{1/2n} \;=\; \prod_{r = 1}^6 \bigbk{\alphabz^{(r)}}^{A_r/2} \;=\; \prod_{r = 1}^3 \bigbk{\alphabz^{(r)}}^{A_r}
  \end{align*}
  are the number of large X and Z-blocks normalized by taking the $2n$-th root, respectively.
\item $\alphant \defeq \lim\limits_{n \to \infty} \numalpha^{1/2n} = \prod_{r = 1}^6 \bigbk{\alphant^{(r)}}^{A_r/2} = \prod_{r = 1}^3 \bigbk{\alphant^{(r)}}^{A_r}$ is the number of triples consistent with $\alpha^{(r)}$ in all regions $r$, where $\alphant^{(r)} \coloneqq 2^{H(\alpha^{(r)})}$.
\item With these notations, we have
  \[\displaystyle\prod_{r=1}^6 \biggbk{\max\limits_{\alpha' \in \distShareMargin[\alpha^{(r)}]} \alphant[\alpha']}^{A_r / 2} = \prod_{r=1}^3 \biggbk{\max\limits_{\alpha' \in \distShareMargin[\alpha^{(r)}]} \alphant[\alpha']}^{A_r} = \lim\limits_{n \to \infty} \numtriple^{1/2n}\]
  according to \cref{lem:numtriple_singledist}.
\item $\alphap = \lim\limits_{n \to \infty} \pcomp^{1/2n}$. Its closed form is given in \cref{lemma:pcomp_g_c}.
\item $\alphaval \defeq \prod_{r = 1}^6 \bigbk{\alphaval^{(r)}}^{A_r / 2} = \prod_{r = 1}^3 \bigbk{\alphaval^{(r)}}^{A_r}$ where $\alphaval^{(r)} \coloneqq \prod_{i' + j' + k' = 2^\lvl} \valsix(T_{i',j',k'}, \tilde\alpha^{(r)}_{i', j', k'})^{2\beta^{(r)}(i', j', k')}$.
\end{itemize}
Then, we take the $2n$-th root on both sides of \eqref{eq:value_bound_c}, obtaining
\[
 \valsix(T_{i,j,k}, \splresZ) \ge \min\bk{\frac{\alphant \cdot \alphabx}{\prod_{r=1}^3 \bigbk{\max_{\alpha' \in \distShareMargin[\alpha^{(r)}]} \alphant[\alpha']}^{A_r}}, \, \frac{\alphabz}{\alphap}} \cdot \alphaval.
 \numberthis \label{eq:numeric_conclusion_c}
\]
Once the distributions $\alpha^{(r)}$ are given, we can verify the lower bound of $\valsix(T_{i,j,k}, \splresZ)$ via \cref{alg:verify_c}.

\begin{figure}[ht]
  \begin{center}
    \begin{tcolorbox}
      \captionof{algocf}{Verifying the Lower Bound}{\label{alg:verify_c}}
      \vspace{-0.5em}
      Assume $\alpha^{(r)}$ and $A_r$ are given for $r \in [3]$. Moreover, for each level-$\lvl$ component $T_{i',j',k'}$ and region $r \in [3]$, a pair $(\splres^{(r)}_{i',j',k'}, V^{(r)}_{i',j',k'})$ is given, indicating that $\valsix(T_{i',j',k'}, \splres^{(r)}_{i',j',k'}) \ge V^{(r)}_{i',j',k'}$.
      \begin{enumerate}
      \item Compute $\alphabx, \alphabz, \alphant, \alphap, \alphaval$ according to their closed forms.
      \item For $r = 1, 2, 3$, solve the following convex optimization problem:
        \[
          \begin{array}{cc}
            \textup{maximize} & \alphant[\alpha'] = 2^{H(\alpha')} \\
            \textup{subject to} & \alpha' \in \distShareMargin[\alpha^{(r)}].
          \end{array}
        \]
        Then, compute ${\prod_{r=1}^3 \bigbk{\max_{\alpha' \in \distShareMargin[\alpha^{(r)}]} \alphant[\alpha']}^{A_r}}$ based on the optimal solutions.
      \item Calculate the lower bound of $\valsix(T_{i,j,k}, \splresZ)$ according to \eqref{eq:numeric_conclusion_c}.
      \end{enumerate}
    \end{tcolorbox}
  \end{center}
\end{figure}

\subsection{Value from Merging When \texorpdfstring{$\{i,j,k\}$}{\{i,j,k\}} Contains Zero}
\label{sec:component_value_zero}

In \cite{williams2012}, they gave the formula for computing the value of $T_{0,j,k}$ without restricted-splitting constraints:
\begin{lemma}[\cite{williams2012}]
  \label{thm:zero-i}
  For level-$\nextlvl$ component $T_{0,j,k}$ where $j\leq k$, its value is given by:
  \[ \valsix(T_{0,j,k})=\left(\sum_{b\leq j,\; 2\mid (b-j)} \binom{(j+k)/2}{b,(j-b)/2,(k-b)/2}\cdot q^b\right)^{\tau}. \]
  And $\valsix(T_{0,j,k})=\valsix(T_{0,k,j})=\valsix(T_{j,0,k})=\valsix(T_{k,0,j})=\valsix(T_{j,k,0})=\valsix(T_{k,j,0})$.
\end{lemma}

Next, we show that a slightly modified version of it can be used to compute the restricted-splitting value $\valsix(T_{i,j,k}, \splresZ)$. Without loss of generality, we assume $i = 0$, and $(0, j, k)$ is a level-$\nextlvl$ component.

Notice that $T_{0,j,k}^{\otimes n}$ is isomorphic to $\angbk{1, 1, m}$ where $m$ is the number of Z-variables in $T_{0,j,k}^{\otimes n}$. Let $\alpha$ be the split distribution of $(0, j, k)$ with marginal $\splresZ$ (note that $\alpha$ is uniquely determined by $\splresZ$). Let $\mathcal T = T_{0,j,k}^{\otimes n}[\splresZ]$. $\mathcal T$ is also isomorphic to $\angbk{1, 1, m'}$ where $m'$ is the number of Z-variables that are consistent with $\splresZ$, i.e., have not been zeroed out. Let $\numzblock = 2^{nH(\splresZ) + o(n)}$ be the number of level-$\lvl$ Z-blocks consistent with $\splresZ$, we have
\[
  m' = \numzblock \cdot \prod_{j'+k'=2^\lvl} (\valsix(T_{0,j',k'}) \cdot \valsix(T_{0,\; j-j',\; k-k'}))^{n\alpha(0,j',k')/\tau},
\]
Then, the restricted-splitting value is exactly $\valsix(T_{0,j,k}, \splresZ) = \lim\limits_{n\to\infty} (m')^{\tau/n}$.

Thus, we just use the normal values for level-$\lvl$ components $(0,j',k')$ and $(0,j-j',k-k')$ as in \cref{thm:zero-i} to obtain the lower bound. The case where $j = 0$ is similar. When $k = 0$, the restricted-splitting value is equivalent to the value without split restriction, which can also be computed precisely by \cref{thm:zero-i}.

\section{Heuristics and Numerical results}
\label{sec:result}
In \cref{sec:global_value,sec:component_value}, we have introduced approaches to bound the value of the Coppersmith-Winograd tensor or its component, given some distributions as parameters. We have written a program to apply these approaches to give a lower bound of $\valsix(\CW_q)$, which leads to an upper bound of $\omega$. In this section, we will introduce the basic idea of our program.

We start by illustrating the framework of our optimization process. For convenience, we will use the following terminology throughout this section:

\begin{definition}
  For a level-$\lvl$ component $T_{i,j,k}$, a \emph{(restricted-splitting) value pair} is a pair $(V_{i, j, k}, \splres_{i, j, k})$ representing the lower bound $\valsix(T_{i,j,k}, \splres_{i,j,k}) \ge V_{i,j,k}$.
\end{definition}

\begin{algorithm}[!h]
  \caption{Optimization Framework}
  \label{alg:opt_frame}
  \DontPrintSemicolon
  For each level-1 component $T_{i,j,k}$ ($i + j + k = 2$), we obtain its value $\valsix(T_{i,j,k})$ via its closed form.\;
  \For{\textbf{each} level-2 componenet $T_{i,j,k}$ ($i + j + k = 4$)} {
    Choose a split distribution $\alpha^{(i,j,k)}$ of $(i, j, k)$.\;
    Obtain a value pair $(V_{i,j,k}, \splres_{i,j,k})$ (explained later). Both $V_{i,j,k}$ and $\splres_{i,j,k}$ are determined by the chosen distribution $\alpha^{(i,j,k)}$.\;
  }

  \For{$\lvl = 3$ to $\lvl^*$} {
    \For{\textbf{each} level-$\lvl$ component $T_{i,j,k}$ ($i + j + k = 2^{\lvl}$)} {
      \uIf{any of $i, j, k$ is zero} {
        Choose a split distribution $\alpha^{(i,j,k)}$ of $(i,j,k)$.\;
        Apply the merging approach in \cref{sec:component_value_zero} to obtain a value pair $(V_{i,j,k}, \splres_{i,j,k})$, using the chosen distribution $\alpha^{(i,j,k)}$ as its parameter.\;
      } \Else {
        Apply the algorithm in \cref{sec:component_value} with some chosen parameters, specifically $A_1 + A_2 + A_3 = 1$ and split distributions $\alpha^{(r)}$ ($r = 1, 2, 3$) of components $(i, j, k)$, $(j, k, i)$, and $(k, i, j)$. This results in a value pair $(V_{i,j,k}, \splres_{i,j,k})$, which depends on both the chosen parameters and the value pairs from level-$\lastlvl$ obtained in previous steps.
      }
    }
  }

  Apply the algorithm in \cref{sec:global_value} with chosen distribution $\alpha$, obtaining a lower bound $\valsix\big(\CW_q^{\otimes 2^{\lvl^*-1}}\big) \le \Vglob$, which is based on the value pairs from level-$\lvl^*$.\; \label{line:glob_val}
\end{algorithm}

\cref{alg:opt_frame}'s output is the lower bound $\Vglob$ of $\valsix(\CW_q^{\otimes 2^{\lvl^* - 1}})$ obtained in the last line, which is called the \emph{global value}. If the global value is large enough, \tauthm{} can imply $\omega \le 3 \tau$ for the $\tau$ we use. The framework of \cref{alg:opt_frame} is similar to prior works; the main difference is that we use (restricted-splitting) value pairs to connect different levels instead of values.

Note that we are choosing different approaches for different component $T_{i,j,k}$ to obtain its value pair. For the components $T_{i,j,k}$ of level $\lvl \ge 3$ where $0 \notin \midBK{i,j,k}$, we use the asymmetric hashing approach introduced in \cref{sec:component_value}. If $0 \in \midBK{i,j,k}$, we can use the ``merging approach'' to obtain a larger value, as discussed in \cref{sec:component_value_zero}. See, for example, Section 2.4 of \cite{ambainis2015fast} for details. If $\lvl = 2$, the asymmetric hashing approach is not applicable, so we use the symmetric hashing approach as in prior works.

\paragraph*{Component values from the symmetric hashing approach.}
In prior works, the symmetric hashing method is used to obtain a lower bound of $\valsix(T_{i,j,k})$, for some level-$\lvl$ component $T_{i,j,k}$ where $\lvl \ge 2$. However, we need a bound for $\valsix(T_{i,j,k}, \splres_{i,j,k})$, where the split distribution of $k$ is restricted. In fact, the symmetric hashing method is already enough to obtain such a bound.

Let $\T \defeq \sym_6(T_{i,j,k}^{\otimes n})$. The symmetric hashing method specifies a joint split distribution $\alpha$, then applies symmetric hashing to obtain some disjoint triples that are consistent with $\alpha$ (or its rotation) in each region. If we apply the same procedure on $\T' = \sym_6\bk{T_{i,j,k}^{\otimes n}[\splres_{i,j,k}]}$ where $\splres_{i,j,k} = \alphz$, the procedure is not affected. Therefore, a bound of the restricted-splitting value
\[
  \valsix(T_{i,j,k}, \splres_{i,j,k}) \ge 2^{\bk{H(\alphx) + H(\alphy) + H(\alphz)}/3} \cdot \; \smashoperator{\prod_{i'+j'+k'=2^{\lvl-1}}} \;\; \bk{\valsix(T_{i',j',k'}) \cdot \valsix(T_{i-i',\; j-j',\; k-k'})}^{\alpha(i',j',k')}
\]
is obtained.

This method is only used for level-2 components in our algorithm; the values of level-1 components are trivial.

\paragraph{Optimization problem.} \cref{alg:opt_frame} leads to the following optimization problem: maximizing $\Vglob$ (the output of the algorithm) by choosing a feasible set of parameters. Similar to prior works, we need to actually solve this optimization problem in order to give a bound on $\omega$. However, the optimization problem here is more challenging than that in the prior works.

We illustrate the difficulties by recalling the optimization framework in prior works (e.g., \cite{legall2014,alman2021}). They optimize the parameters level by level, component by component. For every component $T_{i,j,k}$, their goal is to maximize the lower bound of $\valsix(T_{i,j,k})$, namely $V_{i,j,k}$, since it is the only property of $T_{i,j,k}$ we used in subsequent levels. They solve an individual optimization (sub-)problem for maximizing $V_{i,j,k}$. Some heuristics are further applied to the subproblem so that it becomes a convex optimization problem, which enables the use of efficient solvers. The convex optimization problems are solved one by one, in the same order as \cref{alg:opt_frame}.

The first difference here is that we cannot simply decouple the optimization problem into different components and solve them separately. Assume for some component $T_{i,j,k}$ we can achieve two value pairs, $(V_{i,j,k}, \splres_{i,j,k})$ and $(V'_{i,j,k}, \splres'_{i,j,k})$. Even if $V_{i,j,k} > V'_{i,j,k}$, we cannot simply say that the first one is better, because the split distribution $\splres_{i,j,k}$ also affects subsequent levels. Even when $V_{i,j,k} = V'_{i,j,k}$, it is still hard to tell whether $\splres_{i,j,k}$ is better than $\splres'_{i,j,k}$ without providing the parameters of subsequent levels.

The second challenge is that our bounds from \cref{sec:global_value,sec:component_value} have more complicated forms than the prior works. This requires us to apply more complicated heuristics so that the objective function can become convex.

In the following subsections, we introduce our approach to address these difficulties.

\subsection{Alternating Optimization}

We still want to decouple the whole optimization problem, focusing on a single component at a time. We use a technique called \emph{alternating optimization} to achieve this. It is a simple idea that has been well studied in machine learning.

Recall that in \cref{alg:opt_frame}, for every component $T_{i,j,k}$, a set of parameters are introduced to obtain its value pair. We denote these parameters by $\param_{i,j,k}$. Similarly, in the last step, parameters $\paramglob$ are used to obtain the global value. Our final objective is the result of \cref{alg:opt_frame} running with these parameters $\midBK{\paramglob, \param_{i,j,k}}_{i + j + k = 2^{\lvl} \le 2^{\lvl^*}}$.

Suppose we already have a feasible solution $\midBK{\paramglob^{(t)}, \param^{(t)}_{i,j,k}}$. We enumerate all components $(i,j,k)$ from lower levels to higher levels, one at a time, running a subroutine to update the parameters from $\param^{(t)}_{i,j,k}$ to $\param^{(t+1)}_{i,j,k}$. (We also do this for $\paramglob$.) After the enumeration, the collection of current parameters becomes $\midBK{\paramglob^{(t)}, \param^{(t + 1)}_{i,j,k}}$. The above procedure is called an \emph{iteration} of alternating optimization. We expect the final objective to increase contiguously if we run many iterations.

We obtain the first feasible solution by using a similar approach of previous works. For each component $T_{i,j,k}$, we solve the following subproblem:
\[
  \begin{array}{ccc}
    \textup{maximize} & V_{i,j,k} \\
    \textup{subject to} & \param_{i,j,k} \textup{ is feasible}.
  \end{array}
\]
Here $V_{i,j,k}$ is regarded as a function of $\param_{i,j,k}$; its specific form depends on the type of $(i,j,k)$. We assume this subproblem can be solved efficiently. By solving it for all components $T_{i,j,k}$ from lower to higher levels, and finally for $\CW_q^{\otimes 2^{\lvl^* - 1}}$, we obtain a collection of feasible parameters. We let them be the starting parameters of the alternating optimization process, namely $\midBK{\param^{(0)}_{i,j,k}}$.

\paragraph{Objective for subproblems.} In later iterations, we specify the objective functions of the subproblems differently. The objective function should express the global value's ``demand'' for $V_{i,j,k}$ and $\splres_{i,j,k}$ through the subsequent levels. A straightforward attempt is to precompute $\partial \log \Vglob / \partial \log V_{i,j,k}$ and $\partial \log \Vglob / \partial \splres_{i,j,k}(k_l)$ for all $k_l \in [0, 2^{\lvl - 1}]$ ($i + j + k = 2^\lvl$), then set the objective function to
\[
  \frac{\partial \log \Vglob}{\partial \log V_{i,j,k}} \cdot \log V_{i,j,k} + \sum_{k_l = 0}^{2^{\lvl - 1}} \frac{\partial \Vglob}{\partial \splres_{i,j,k}(k_l)} \cdot \splres_{i,j,k}(k_l).
  \numberthis \label{eq:ideal_sub_obj}
\]
(Our programs will store the logarithm of values for easier implementation. Consequently, the partial derivatives here are taken with respect to logarithms.)
However, it incurs the following issue. Recall in \cref{sec:global_value}, our obtained bound is
\[
  \Vglob \defeq \min\bk{\frac{\alphant \cdot \alphabx}{\max_{\alpha' \in \distShareMargin} \alphant[\alpha']}, \, \frac{\alphabz}{\alphap}} \cdot \alphaval.
  \tag{\ref{eq:numeric_conclusion_g} revisited}
\]
It contains a min operation. When we take the partial derivative, only one branch of the min operation can keep its gradient information, while the other branch's gradient disappears. The same issue occurs for the algorithm in \cref{sec:component_value} as well. To develop intuition for this, imagine that $T_{i,j,k}$ is already in the last level, i.e., $i + j + k = 2^{\lvl^*}$. If we run \cref{alg:opt_frame} with the current parameters, when it computes $\Vglob$ according to \eqref{eq:numeric_conclusion_g}, the two branches of the min operation are equal. Suppose we have a chance to decrease $\alphap$ while keeping other quantities in \eqref{eq:numeric_conclusion_g} unchanged, by choosing different parameters $\param_{i,j,k}$ that result in a different $\splres_{i,j,k}$. This would probably benefit $\Vglob$ after we solve the subproblem for $\Vglob$ again (if we can slightly incrase the left branch and decrease the right branch, the resulting $\Vglob$ will increase). However, such benefit is not expressed in $\partial \log \Vglob / \partial \splres_{i,j,k}(k_l)$: If the gradient is passed through the left branch of the min operation, all $\splres_{i,j,k}(k_l)$ will have zero gradient.

To resolve this issue, we design the following ``soft-min'' function:
\[
  \mathop{\textup{softmin}}\bk{\frac{\alphant \cdot \alphabx}{\max_{\alpha' \in \distShareMargin} \alphant[\alpha']}, \, \frac{\alphabz}{\alphap}} \defeq \bk{\frac{\alphant \cdot \alphabx}{\max_{\alpha' \in \distShareMargin} \alphant[\alpha']}}^{2/3} \cdot \bk{\frac{\alphabz}{\alphap}}^{1/3}.
\]
The intuition behind this definition is that, we assume we can increase $\alphabz$ and decrease $\alphabx$ while keeping $\alphabx \alphaby \alphabz = \alphabx^2 \alphabz$ and other quantities in \eqref{eq:numeric_conclusion_g} unchanged.

Finally, we specify the objective function for subproblems as \eqref{eq:ideal_sub_obj}, but the partial derivatives are defined by replacing ``min'' with ``soft-min''.\footnote{
  The above description is only for illustrating the basic idea of this heuristic rather than defining the specific objective function. In practice, by observing how $\splres_{i,j,k}$ influences the logarithm of $\alphap$ in the subsequent level, we can derive a better estimation than directly using the gradient. This improved estimation is still concave, allowing us to use convex optimization solvers.
}
The remaining task is to solve the following optimization program for a component $T_{i,j,k}$:
\[
  \begin{array}{ccc}
    \textup{maximize} & \displaystyle\frac{\partial \log \Vglob}{\partial \log V_{i,j,k}} \cdot \log V_{i,j,k} + \sum_{k_l = 0}^{2^{\lvl - 1}} \frac{\partial \log \Vglob}{\partial \splres_{i,j,k}(k_l)} \cdot \splres_{i,j,k}(k_l) \\
    \textup{subject to} & \param_{i,j,k} \textup{ is feasible}.
  \end{array}
  \numberthis\label{eq:rem_opt_prog}
\]
We will introduce the approach in the next subsection.

\subsection{Heuristics within a Component}

The specific form of \eqref{eq:rem_opt_prog} depends on the type of $T_{i,j,k}$, or in other words, the approach used to obtain the bound $V_{i,j,k}$. Here, we illustrate the heuristics we used for the most challenging case: when $V_{i,j,k}$ is obtained using the algorithm from \cref{sec:component_value}. In this case, \eqref{eq:rem_opt_prog} is rewritten as
\[
  \begin{array}{cll}
    \textup{maximize} & \objfun(\log V_{i,j,k}, \splres_{i,j,k}) \\
    \textup{where} & \displaystyle V_{i,j,k} \defeq \min\bk{
                     \prod_{r=1}^3
                     \bk{\frac{\alphant^{(r)} \cdot \alphabx^{(r)}}{\max_{\alpha' \in \distShareMargin[\alpha^{(r)}]} \alphant'}}^{A_r}, \;
                     \prod_{r=1}^3 \bk{\frac{\alphabz^{(r)}}{\alphap^{(r)}}}^{A_r}
                     } \cdot \prod_{r=1}^3 \bk{\alphaval^{(r)}}^{A_r} \\
                      & \splres_{i,j,k} \defeq A_1 \cdot \alphz^{(1)} + A_2 \cdot \alphy^{(2)} + A_3 \cdot \alphx^{(3)} \vspace{0.5em} \\
    \textup{subject to} & A_1 + A_2 + A_3 = 1 \\
                      & A_1, A_2, A_3 \ge 0 \\
                      & \textup{split distributions } \alpha^{(1)}, \alpha^{(2)}, \alpha^{(3)} \textup{ of } (i,j,k), \, (j,k,i), \, (k,i,j).
  \end{array}
  \numberthis\label{eq:opt_comp}
\]
As discussed in the previous subsection, $\objfun$ is a non-negative linear combination of $\log V_{i,j,k}$ and $\splres_{i,j,k}$.

Below, we introduce several heuristics to solve \eqref{eq:opt_comp}. We first divide the parameters into two groups, $\midBK{A_1, A_2, A_3}$ and $\midBK{\alpha^{(1)}, \alpha^{(2)}, \alpha^{(3)}}$, then apply the alternating optimization technique. That is, we first fix $A_1, A_2, A_3$ as constants and optimize $\alpha^{(1)}, \alpha^{(2)}, \alpha^{(3)}$, then do it conversely; this process is repeated multiple times to refine our solution.

\paragraph{Optimize $A_1, A_2, A_3$.} This step is relatively easy, as we can rewrite
\begin{gather*}
  \log V_{i,j,k} = \min \bk{\sum_{r=1}^3 c_r \cdot A_r, \, \sum_{r=1}^3 c'_r \cdot A_r} + \sum_{r=1}^3 c''_r \cdot A_r
\end{gather*}
as a concave function of $A_1, A_2, A_3$, where $c_r, c'_r, c''_r$ are constants that only depend on $\alpha^{(r)}$. Moreover, $\splres_{i,j,k}$ is linear in the parameters $A_1, A_2, A_3$. So $\objfun$ is also concave in $A_1, A_2, A_3$ since it is a non-negative linear combination of two concave functions. This means \eqref{eq:opt_comp} is a concave maximization program (it is easy to check that all constraints on $A_1, A_2, A_3$ are linear) which can be solved precisely and efficiently.

\newcommand{\old}{(\textup{old})}
\newcommand{\new}{(\textup{new})}

\paragraph{Optimize $\alpha^{(1)}, \alpha^{(2)}, \alpha^{(3)}$.} In this step, \eqref{eq:opt_comp} is no longer concave, so we need to apply several heuristics to change its form. Let $\alpha^{(r)\old}$ and $\alpha^{(r)\new}$ ($r = 1, 2, 3$) represent the parameters before and after the current alternating optimization step, i.e., $\alpha^{(r)\old}$ is predetermined and $\alpha^{(r)\new}$ is what we want to optimize. The basic intuition for our heuristics is that we expect $\alpha^{(r)\new}$ to be close to $\alpha^{(r)\old}$, which means several complicated quantities in \eqref{eq:opt_comp} will not change too much. Then we compute these quantities with $\alpha^{(r)\old}$ and regard them as constants. Specifically:
\begin{enumerate}
\item Let $R \defeq \prod_{r=1}^3 (\alphant^{(r)} / \max_{\alpha' \in \distShareMargin[\alpha^{(r)}]} \alphant')^{A_r}$. This factor appears in the form of $V_{i,j,k}$. We calculate this formula with the old parameters $\alpha^{(r)\old}$, and denote the result by $R^{\old}$. Then, we replace $R$ with the fixed constant $R^{\old}$ in \eqref{eq:opt_comp}.
\item Similarly, we calculate $\alphap^{(r)}$ with the old parameters $\alpha^{(r)\old}$, and denote the result by $\alphap^{(r)\old}$. We replace the occurrence of $\alphap^{(r)}$ with $\alphap^{(r)\old}$ in \eqref{eq:opt_comp}.
\end{enumerate}
With the two heuristics above, we estimate $\log V_{i,j,k}$ by
\[
  \log V_{i,j,k} \approx \min \bk{\log R^{\old} + \sum_{r=1}^3 A_r \cdot \log \alphabx^{(r)}, \; \sum_{r=1}^3 A_r (\log \alphabz^{(r)} - \log \alphap^{(r)\old})} + \sum_{r=1}^3 A_r \cdot \log \alphaval^{(r)}.
\]
Note that $\log \alphabx = H(\alphx)$ is a concave function of $\alpha$. Therefore, one can check that our approximation for $\log V_{i,j,k}$ is concave in all the parameters $\alpha^{(1)}, \alpha^{(2)}, \alpha^{(3)}$. Then we can solve \eqref{eq:opt_comp} via convex optimization tools.

Inspired by \cite{alman2021}, we can further refine $\alpha^{(1)}, \alpha^{(2)}, \alpha^{(3)}$ by doing an optimization while fixing their marginals $\alphx^{(r)}, \alphy^{(r)}, \alphz^{(r)}$. Since $\max_{\alpha' \in \distShareMargin[\alpha^{(r)}]} \alphant'$ is determined by the marginal distributions, the factor $R$ becomes a concave function of the parameters, so we no longer use the first heuristic in this further adjustment.\footnote{We also replaced ``min'' with ``soft-min'' during this adjustment step, because it empirically improves the result.}

\subsection{Numerical Results}

We wrote a program to perform the above optimization procedure. The program is written in MATLAB~\cite{MATLAB:2022} while CVX package~\cite{cvx,gb08} and Mosek software~\cite{mosek} are used for convex optimization. We analyzed the eighth tensor power of the Coppersmith-Winograd tensor. For $\tau = \ourbound / 3$, the program shows $\valsix(\CW_5^{\otimes 8}) > 7^8 + 9$. By \cref{thm:Schonhage}, we have $\omega \le \ourbound$. We also have written a standalone program to verify this result, which recursively verifies the bounds of restricted-splitting values.\footnote{Data and code are available at \url{https://osf.io/dta6p/}. Parameters for the second-power result is given in \cref{sec:level-2-global}.}

Our program differs from \cref{alg:opt_frame} in that it optimizes multiple value pairs for every component instead of just one, allowing different components in the subsequent level to use different value pairs. In previous works, a single bound $V_{i,j,k}$ is obtained for each component $T_{i,j,k}$: If there were multiple valid bounds, we could retain the best one and discard the others. However, this is not the case for (restricted-splitting) value pairs, so we made this adjustment. As a result, the optimization process becomes significantly slower when analyzing higher powers, which is why our analysis stops at the eighth power.

When analyzing the second and fourth tensor power of the Coppersmith-Winograd tensor using our asymmetric hashing approach, we also notice the improvement from previous works, as shown in \cref{table:result}.

\begin{table}[!h]
  \captionsetup{font=small}
  \caption{Upper bounds of $\omega$ obtained by analyzing the $m$-th tensor power of the CW tensor.}
  \label{table:result}
  \centering
  \begin{tabular}{|c|c|c|c|}
    \hline
    $m$ & Bounds from Our Methods & Bounds from Prior Works & References \\
    \hline
    1 & N/A & 2.387190 & \cite{coppersmith1987matrix} \\
    \hline
    2 & 2.374631 & 2.375477 & \cite{coppersmith1987matrix} \\
    \hline
    4 & 2.371919 & 2.372927 & \cite{stothers2010, williams2012, legall2014} \\
    \hline
    8 & 2.371866 & 2.372865 & \cite{williams2012, legall2014} \\
    \hline
    16 & N/A & 2.372864 & \cite{legall2014} \\
    \hline
    32 & N/A & 2.372860 & \cite{alman2021} \\
    \hline
  \end{tabular}
\end{table}

We think a slightly better bound can be obtained by analyzing higher tensor powers.

\bibliographystyle{alphaurl}
\bibliography{sample}

\appendix
\section{Missing Proofs from Section 4}
\label{appendix:missing_proofs}

\MoreValues*

\begin{proofof}{\Cref{lem:non-rot-values}}
  (a) and (b) are clear because $T_{0,0,4} \cong \angbk{1,1,1}$ and $T_{0,1,3} \cong \angbk{1,1,2q}$.

  To show (c), we let $\alpha^{(0,2,2)}$ be a split distribution for the component $(0,2,2)$:
  \begin{align*}
    \alpha^{(0,2,2)}(0, 1, 1) &= a'', \\
    \alpha^{(0,2,2)}(0, 0, 2) &= \alpha^{(0,2,2)}(0, 2, 0) = b'',
  \end{align*}
  where $a'' + 2b'' = 1$. We can write the corresponding Z-marginal distribution
  \[
    \alphz^{(0,2,2)}(0) = \alphz^{(0,2,2)}(2) = b'', \qquad \alphz^{(0,2,2)}(1) = a''.
  \]
  Consider the tensor in which all Z-blocks inconsistent with $\alphz^{(0,2,2)}$ are zeroed-out, i.e., $T_{0,2,2}^{\otimes m}[\alphz^{(0,2,2)}]$. It is isomorphic to a matrix multiplication tensor where all remaining Z-variables are utilized:
  \[
    T_{0,2,2}^{\otimes m}[\alphz^{(0,2,2)}] \cong \angbk{1, \; 1, \; \binom{m}{a''m, \, b''m, \, b''m} q^{2a''m}}.
  \]
  The size of this matrix multiplication tensor on the right side is maximized at $b'' = 1 / (2 + q^2)$, which leads to the lower bound of the \emph{non-rotational restricted-splitting value}
  \[
    V_\tau^{\textup{(nrot)}}(T_{0,2,2}, \tilde\alpha_\B) \ge \lim_{m\to\infty} \bk{\binom{m}{a''m, \, b''m, \, b''m} q^{2a''m}}^{\tau/m} = (q^2+2)^\tau,
  \]
  where $\tilde\alpha_\B$ is the optimal choice of $\alphz^{(0,2,2)}$ which is specified in \cref{eq:tilde_B}. Thus (c) holds.

  To show (d), we analyze $\sym_3(T_{1,1,2}^{\otimes m})$ similarly to \cite{coppersmith1987matrix}. Let $\alpha^{(1, 1, 2)}$ be the split distribution of component $(1, 1, 2)$:
\begin{align*}
  \alpha^{(1, 1, 2)}(0, 1, 1) = \alpha^{(1, 1, 2)}(1, 0, 1) &= a', \\
  \alpha^{(1, 1, 2)}(1, 1, 0) = \alpha^{(1, 1, 2)}(0, 0, 2) &= b', 
\end{align*}
where $2a' + 2b' = 1$. To see its restricted-splitting value, we degenerate $\sym_3(T_{1,1,2}^{\otimes m}[\alphz^{(1,1,2)}])$ to independent matrix multiplication tensors, where
\[
  \alphz^{(1,1,2)}(0) = \alphz^{(1,1,2)}(2) = b', \qquad \alphz^{(1,1,2)}(1) = 2a'
\]
is the Z-marginal split distribution of $\alpha^{(1,1,2)}$.

Let $\T$ be the tensor obtained from $T_{1,1,2}^{\otimes m}$ by zeroing out all blocks inconsistent with the marginal distributions of $\alpha^{(1,1,2)}$. $\T$ is a subtensor of $T_{1,1,2}^{\otimes m}[\alphz^{(1,1,2)}]$; it can be obtained by zeroing out X and Y-blocks from $T_{1,1,2}^{\otimes m}[\alphz^{(1,1,2)}]$. Then, we take the 3-symmetrization of $\T$, denoted by $\sym_3(\T) \defeq \T \otimes \T^{\rot} \otimes \T^{\rot\,\rot}$, which is a subtensor of $\sym_3(T_{1,1,2}^{\otimes m}[\alpha^{(1,1,2)}])$. Next, symmetric hashing method is applied on $\sym_3(\T)$ to obtain $\binom{m}{m/2}^2\binom{m}{2a'm, \, b'm, \, b'm}\cdot 2^{-o(m)}$ disjoint triples. Each triple is isomorphic to
\[
  \angbk{q^{(4a' + 2b')m},\, q^{(4a' + 2b')m},\, q^{(4a' + 2b')m}}.
\]
This leads to the lower bound
\[
  V_\tau^{(3)}(T_{1,1,2}, \tilde\alphz^{(1,1,2)}) \ge \bk{\binom{m}{m/2}^2 \binom{m}{2a'm,\, b'm,\, b'm}}^{1/(3m)} \cdot q^{(4a' + 2b')\tau}.
\]
From \cite{coppersmith1987matrix}, we know that when $b' = 1 / (2 + q^{3\tau})$, this is optimized at $V_\tau^{(3)}(T_{1,1,2}, \alphz^{(1,1,2)}) \ge 2^{2/3} q^{\tau} (q^{3\tau} + 2)^{1/3}$. In this case, $\alphz^{(1,1,2)} = \tilde\alpha_\A$, where $\tilde\alpha_\A$ is defined in \cref{eq:tilde_A}. This implies (d).
\end{proofof}

\AdditionalZ*

\begin{proofof}{\Cref{lem:prop1}}
  We prove by contradiction. Let $(X_{\hat I}, Y_{\hat J}, Z_{\hat K})$ be a remaining level-1 triple. According to the zeroing-out rules, we know
  \begin{align}
    \Split(\hat K, S_{2}) &= \tilde\alpha_{\textup{average}}, \label{eq:split_avg_correct} \\
    \Split(\hat I, S_{2,0,2}) &= \tilde\alpha_\B^{(\revtext)}, \label{eq:split_x_correct} \\
    \Split(\hat J, S_{0,2,2}) &= \tilde\alpha_\B^{(\revtext)}. \label{eq:split_y_correct}
  \end{align}
  In component $(2,0,2)$, the split distribution in X leads to the split distribution in Z. So from \cref{eq:split_x_correct} we know $\split(\hat K, S_{2,0,2}) = \tilde\alpha_\B$. Similarly we have $\split(\hat K, S_{0,2,2}) = \tilde\alpha_\B$. Combined with \cref{eq:split_avg_correct}, we infer that $\split(\hat K, S_{1,1,2}) = \tilde\alpha_\A$. (Here we assume $\alpha(1,1,2) > 0$; otherwise the lemma is trivial.) These split distributions showed that $Z_{\hat K}$ must be compatible with $X_I, Y_J, Z_K$.
\end{proofof}

\EqualSize*

\begin{proofof}{\Cref{lem:equal-size}}
  The degeneration of $\T^*$ provided by \eqref{eq:value_LB} is formed by combining the degenerations for several factors, as shown in \Cref{lem:non-rot-values}. For each factor, it is easy to verify that its degeneration in our analysis is a zeroing out and produces a direct sum of equal-sized matrix multiplication tensors. Therefore, the degeneration of $\T^*$, as the tensor product of the zeroing-outs of these factors, is also a zeroing-out which produces equal-sized matrix multiplication tensors.
\end{proofof}

\section{Heuristics for Section 4}

\label{appendix:heuristics}
Consider the following approximated optimization program:
\begin{align*}
  \begin{array}{cl}
    \text{maximize} & \min(\alphabx, \alphabz / \alphap^*) \, \alphaval \\
    \text{subject to} & \alpha \textup{ is defined by parameters $a, b, c, d, e$}.
  \end{array}
\end{align*}
It has the following differences from the original program:
\begin{itemize}
\item The ratio $\alphant / \max_{\alpha' \in \distShareMargin} \alphant'$ is removed from the objective. This ratio is always close to 1 for the optimal $\alpha$ in practice, but it makes the optimization extremely difficult, hence we remove it to simplify the objective function.
\item The feasibility constraint $\alphabx \le \alphabz / \alphap$ is transformed to a minimum in the objective function. In fact, this is not a heuristic: From \cref{sec:global,sec:component} we can see that such modification also leads to valid lower bounds. In this section we did not prove such complicated version of the bound since we aim for better presentation.
\item $\alphap$ is replaced by a fixed reference value $\alphap^*$ (initially it is a function of optimization parameters). It is because $\alphap$ is a complicated function of $a, b, c, d, e$ which stops us to apply convex optimization tools. By replacing it with a fixed number, the new objective become logarithmic concave.
\end{itemize}
Imagine that we already know a distribution $\alpha^{(0)}$ that is close to the optimal value. By substituting $\alphap^*$ by $\alphap^{(0)}$ and solving the approximated program, we can hopefully get a better solution $\alpha^{(1)}$ than $\alpha^{(0)}$. Then, we substitute $\alphap^*$ by $\alphap^{(1)}$ and solve the approximate program again, obtaining its optimal solution $\alpha^{(2)}$. As shown in \cref{alg:optim_sec_power}, we repeat such procedure for a few iterations and take the output of the last iteration as our result solution. The initial solution $\alpha^{(0)}$ can be obtained by solving the approximated program with $\alphap^* = 1$.

\begin{algorithm}[h]
  \caption{Optimization with Heuristics}
  \label{alg:optim_sec_power}
  \DontPrintSemicolon
  $\alpha^{(0)} \gets$ the optimal solution of the approximated program where $\alphap^* = 1$.\;
  \For(\Comment{In practice, $t_{\textup{max}} \le 5$ is sufficient.}){$t \gets 1$ to $t_{\textup{max}}$} {
    $\alpha^{(t)} \gets$ the optimal solution of the approximated program where $\alphap^* = \alphap[\alpha^{(t - 1)}]$.\;
  }
  \label{line:perturb} Slightly perturb the solution $\alpha = \alpha^{(t_{\textup{max}})}$ so that $\alphabx \le \alphabz / \alphap$ holds.\;
  Run \cref{alg:verify_sec_power} on $\alpha$ to obtain the final bound.
\end{algorithm}

Practically, for the solution $\alpha^{(t_{\textup{max}})}$ before perturbation in line~\ref{line:perturb}, we observe $\alphabx \approx \alphabz / \alphap$. Therefore a perturbation is enough to satisfy the constraint without much loss on the solution's quality. In practice, this perturbation step is done manually.

\section{Mising Proofs from Section 7}
\label{appendix:missing_proof_7}
\newcommand{\irjrkr}{\complementshape}

\typicalnonneg*

\begin{proofof}{\Cref{lemma:typical_is_nonnegl}}
  Fix the level-$\lvl$ triple $(X_I, Y_J, Z_K)$, and recall that $S^{(r)}_{i', j', k'}$ is the set of positions $t$ in region $r$ such that $(I_t, J_t, K_t) = (i', j', k')$. We let $S_{i', j', k'}^{(r,\textup{L})} \defeq S_{i', j', k'}^{(r)} \cap (2 \mathbb{N} + 1)$ and $S_{i', j', k'}^{(r, \textup{R})} \defeq S_{i', j', k'}^{(r)} \cap 2\mathbb{N}$ represent the odd and even elements in $S_{i',j',k'}^{(r)}$, respectively. Every $t$ in $S^{(r, \textup{L})}_{i', j', k'}$ is odd, so the corresponding component $(i', j', k')$ on that position appears as the \emph{left part} of some level-$\nextlvl$ component $(i^{(r)}, j^{(r)}, k^{(r)})$; similarly, every $t \in S^{(r, \textup{R})}_{i', j', k'}$ appears as the \emph{right part}. These two sets form a partition $S^{(r)}_{i', j', k'} = S^{(r, \textup{L})}_{i', j', k'} \sqcup S^{(r, \textup{R})}_{i', j', k'}$.

  Recall that we say $Z_{\hat K} \in Z_K$ is useful for $(X_I, Y_J, Z_K)$ when the Z-marginal split distribution in $S^{(r)}_{i',j',k'}$ is the same as $\splres^{(r)}_{i', j', k'}$. Further, we say $Z_{\hat K}$ is \emph{strongly useful} for $(X_I, Y_J, Z_K)$, if its marginal split distribution in $S^{(r, \textup{L})}_{i', j', k'}$ and $S^{(r, \textup{R})}_{i', j', k'}$ are both identical to $\splres^{(r)}_{i', j', k'}$. It is a sufficient condition of usefulness. The set of $Z_{\hat K} \in Z_K$ that are strongly useful for $(X_I, Y_J, Z_K)$ is denoted by $\strongset$.

  The concept of strong usefulness has a tight connection with typicalness. We will show
  \begin{equation}
    \label{eq:typical_is_nonnegl_2side}
    \abs{\strongset \cap \typicalset} \ge \abs{\strongset} \cdot 2^{-o(n)}
  \end{equation}
  and
  \begin{equation}
    \label{eq:2side_is_nonnegl}
    \abs{\strongset} \ge \abs{\usefulset} \cdot 2^{-o(n)}.
  \end{equation}
  It is clear that these two inequalities together can imply \Cref{lemma:typical_is_nonnegl}.
  Next, we first show \eqref{eq:typical_is_nonnegl_2side} by calculating both sides of it.

  \bigskip

  \paragraph{Calculate $\normalfont |\strongset|$.} Below is an equivalent condition of $Z_{\hat K} \in Z_K$ being strongly useful:
  \begin{itemize}
  \item For each region $r \in [6]$ and level-$\lvl$ component $(i', j', k')$, $\split(\hat K, S^{(r, \textup{L})}_{i', j', k'}) = \splres_{i', j', k'}$, while\\ $\split(\hat K, S^{(r, \textup{R})}_{\complementshape}) = \splres_{\complementshape}$.
  \end{itemize}
  For every $r$ and $(i', j', k')$, this condition implies two constraints of the form $\split(\hat K, S) = \splres$, for which we have $\binom{|S|}{[|S| \cdot \splres(k'_l)]_{k'_l}} = 2^{|S| \cdot H(\splres) + o(n)}$ ways to split the Z-indices in $S$. All these constraints apply on disjoint position sets. Multiplying the number of ways over all constraints, we obtain the number of strongly useful blocks $Z_{\hat K} \in Z_K$:
  \begin{align*}
    &\phantom{{}={}} |\strongset| \\
    &= \prod_{r=1}^6 \prod_{i' + j' + k' = 2^{\lvl}} 2^{\Bigbk{\bigabs{S^{(r, \textup{L})}_{i', j', k'}} \cdot H\bigbk{\splres_{i', j', k'}}}} \cdot 2^{\Bigbk{\bigabs{S^{(r, \textup{R})}_{\complementshape}} \cdot H\bigbk{\splres_{\complementshape}}}} \cdot 2^{o(n)} \\
    &= \prod_{r=1}^6 \prod_{i' + j' + k' = 2^{\lvl}} 2^{\Bigbk{\bigabs{S^{(r, \textup{L})}_{i', j', k'}} \cdot \bigbk{H\bigbk{\splres_{i', j', k'}} + H\bigbk{\splres_{\complementshape}}}}} \cdot 2^{o(n)}.
      \numberthis \label{eq:abs_strongset}
  \end{align*}
  (The last equality holds since $\midabs{S^{(r, \textup{L})}} = \midabs{S^{(r, \textup{R})}}$.)

  \paragraph{Calculate $\normalfont |\strongset \cap \typicalset|$.} Consider the following sufficient condition for a block $Z_{\hat K} \in Z_K$ to be both strongly useful and typical:
  \begin{itemize}
  \item For each region $r \in [6]$ and level-$\lvl$ component $(i', j', k')$, and for all $k_1 + k_2 = k'$ and $k_3 + k_4 = k^{(r)} - k'$, there is
    \begin{gather*}
      \abs{\BK{
          t \in S^{(r, \textup{L})}_{i', j', k'} \mymiddle (\hat K_{2t-1}, \hat K_{2t}, \hat K_{2t + 1}, \hat K_{2t + 2}) = (k_1, k_2, k_3, k_4)
        }} \\
      = \splres^{(r)}_{i', j', k'}(k_1) \cdot \splres^{(r)}_{\complementshape}(k_3) \cdot \bigabs{S^{(r, \textup{L})}_{i', j', k'}}.
      \numberthis \label{eq:suf_cond}
    \end{gather*}
    That is, the frequency of occurrence of $(\hat K_{2t-1}, \hat K_{2t}, \hat K_{2t + 1}, \hat K_{2t + 2})$ matches the distribution $\splres^{(r)}_{i', j', k'} \times \splres^{(r)}_{\complementshape}$.
  \end{itemize}
  Suppose some block $Z_{\hat K}$ satisfies the above condition. By summing up \eqref{eq:suf_cond} over all $(i', j', k')$ and all $r \in [6]$, we see $Z_{\hat K}$ is typical. Moreover, the marginals of $\splres^{(r)}_{i', j', k'} \times \splres^{(r)}_{\complementshape}$ on $k_1$ and $k_3$ are identical to $\splres^{(r)}_{i', j', k'}$ and $\splres^{(r)}_{\complementshape}$, respectively, which implies that $Z_{\hat K}$ is strongly useful.

  Next, we count the number of $Z_{\hat K} \in Z_K$ satisfying \eqref{eq:suf_cond} to form a lower bound of $|\strongset \cap \typicalset|$. Similar to above, the constraints from different $(i', j', k')$ and $r$ are applying on disjoint sets of positions. For each of the constraints, the number of ways to split Z-indices $(K_t, K_{t + 1}) = (k', k^{(r)} - k')$ into $(k_1, k_2, k_3, k_4)$ equals
  \begin{gather*}
    \bigbinom{\bigabs{S^{(r, \textup{L})}_{i', j', k'}}}{\Bk{\bigabs{S^{(r, \textup{L})}_{i', j', k'}} \cdot \splres^{(r)}_{i', j', k'}(k_1) \cdot \splres^{(r)}_{\complementshape}(k_3)}_{k_1, k_3}} \\
    = 2^{\bigabs{S^{(r, \textup{L})}_{i', j', k'}} \cdot H\bigbk{\splres^{(r)}_{i', j', k'} \times \splres^{(r)}_{\complementshape}} + o(n)}
    = 2^{\bigabs{S^{(r, \textup{L})}_{i', j', k'}} \cdot \Bigbk{H\bigbk{\splres^{(r)}_{i', j', k'}} + H\bigbk{\splres^{(r)}_{\complementshape}}}} \cdot 2^{o(n)}.
  \end{gather*}
  Multiplying these numbers over all $(i', j', k')$ and $r$, we obtain
  \begin{align*}
    \phantom{{}\ge{}} |\strongset \cap \typicalset|
    &\ge
      \prod_{r=1}^6 \prod_{i' + j' + k' = 2^{\lvl}} 2^{\Bigbk{\bigabs{S^{(r, \textup{L})}_{i', j', k'}} \cdot \bigbk{H\bigbk{\splres_{i', j', k'}} + H\bigbk{\splres_{\complementshape}}}}} \cdot 2^{o(n)},
  \end{align*}
  which equals \eqref{eq:abs_strongset} up to a negligible factor $2^{o(n)}$. Thus, \cref{eq:typical_is_nonnegl_2side} holds.

  \paragraph{Proof of \cref{eq:2side_is_nonnegl}.} A similar argument is applied to show \eqref{eq:2side_is_nonnegl}. We start by transforming \eqref{eq:abs_strongset} into the following equivalent form:
  \begin{align*}
    &\phantom{{}={}} |\strongset| \\
    &= \prod_{r=1}^6 \prod_{i' + j' + k' = 2^{\lvl}} 2^{\Bigbk{\bigabs{S^{(r, \textup{L})}_{i', j', k'}} \cdot H\bigbk{\splres_{i', j', k'}}}} \cdot 2^{\Bigbk{\bigabs{S^{(r, \textup{R})}_{\complementshape}} \cdot H\bigbk{\splres_{\complementshape}}}} \cdot 2^{o(n)} \\
    &= \prod_{r=1}^6 \prod_{i' + j' + k' = 2^{\lvl}} 2^{\bigbk{\bigabs{S^{(r, \textup{L})}_{i', j', k'}} + \bigabs{S^{(r, \textup{R})}_{i', j', k'}}} \cdot H\bigbk{\splres^{(r)}_{i', j', k'}}} \cdot 2^{o(n)}.
      \numberthis \label{eq:abs_strongset_2}
  \end{align*}
  Then, we count the number of useful blocks $Z_{\hat K} \in Z_K$:
  \begin{align*}
    |\usefulset| &= \prod_{r=1}^6 \prod_{i' + j' + k' = 2^{\lvl}}
    \binom{|S^{(r)}_{i', j', k'}|}{[|S^{(r)}_{i', j', k'}| \cdot \splres^{(r)}_{i', j', k'}(k_l)]_{k_l}} \\
                 &= \prod_{r=1}^6 \prod_{i' + j' + k' = 2^\lvl} 2^{\bigabs{S^{(r)}_{i', j', k'}} \cdot H\bigbk{\splres^{(r)}_{i', j', k'}}} \cdot 2^{o(n)} \\
                 &= |\strongset| \cdot 2^{o(n)}.
  \end{align*}
  The last equality holds due to \eqref{eq:abs_strongset_2} and $\bigabs{S^{(r)}_{i', j', k'}} = \bigabs{S^{(r, \textup{L})}_{i', j', k'}} + \bigabs{S^{(r, \textup{R})}_{i', j', k'}}$. Thus, \eqref{eq:2side_is_nonnegl} holds. Combining \eqref{eq:2side_is_nonnegl} with \eqref{eq:typical_is_nonnegl_2side}, we conclude the proof.
\end{proofof}

\end{document}